\newif\ifstoc

\stocfalse

\documentclass[11pt]{article}
\pdfoutput=1
\usepackage{fullpage}           
\usepackage{amsfonts}           
\usepackage{amsthm, amssymb}             
\usepackage{xspace}             
\usepackage{graphicx}
\usepackage{adjustbox}          
\usepackage{multicol}            
\usepackage{url}
\usepackage{enumerate, enumitem}  
\usepackage{subcaption}
\usepackage[usenames]{xcolor} 

\usepackage{amsmath, hyperref, nicefrac}
\usepackage[capitalize]{cleveref}
\usepackage{thm-restate}
\usepackage{parskip}
\usepackage{mathtools}
\usepackage{multirow}

\newtheorem{lemma}{Lemma}

\newtheorem{theorem}{Theorem}
\newtheorem{corollary}[theorem]{Corollary}
\newtheorem{definition}{Definition}
\newtheorem{fact}{Fact}
\newtheorem{proposition}[theorem]{Proposition}
\newtheorem{remark}{Remark}

\title{A New Coreset Framework for Clustering}

\author{Vincent Cohen-Addad\thanks{Google Research, Zurich.} \and 
  David Saulpic\thanks{Sorbonne Universit\'e, Paris}
  \and
  Chris Schwiegelshohn \thanks{Aarhus University} \and
}
\date{}

\colorlet{darkgreen}{green!45!black}

\newcommand{\R}{\mathbb{R}}

\newcommand{\dist}{\text{dist}}
\newcommand{\eps}{\varepsilon}
\newcommand{\opt}{\text{OPT}}
\newcommand{\cost}{\text{cost}}

\newcommand{\calS}{\mathcal{S}}
\newcommand{\calP}{\mathcal{P}}
\newcommand{\calL}{\mathcal{L}}

\newcommand{\calT}{\mathcal{T}}

\newcommand{\E}{\mathbb{E}}

\newcommand{\calE}{\mathcal{E}}

\newcommand{\cand}{\mathbb{C}}
\newcommand{\greedy}{\mathcal{A}}
\newcommand{\A}{\mathcal{A}}
\newcommand{\poly}{\text{poly}}
\newcommand{\alg}{\greedy}
\newcommand{\seeded}{{\color{red}\mathcal{G}}}

\newcommand{\coreset}{\Omega}

\newcommand{\weight}{f}
\newcommand{\inner}{R_I}
\newcommand{\out}{R_O}
\newcommand{\main}{R_M}
\newcommand{\size}{\Gamma}
\newcommand{\polylog}{\text{polylog}}
\newcommand{\outergroup}[1]{G^O}
\newcommand{\structured}{P_1}
\newcommand{\discarded}{\mathcal{D}}
\newcommand{\valuedelta}{\frac{\log^2 1/\eps}{2^{O(z\log z)}\min(\eps^2, \eps^z)}\left(k \log |\cand| + \log \log (1/\eps) + \log(1/\pi)\right)}
\newcommand{\argmin}{\text{argmin}}

\newcounter{sideremark}

\begin{document}
\maketitle
\ifstoc
\else
\begin{abstract}
Given a metric space, the $(k,z)$-clustering problem consists of finding $k$ centers such that the sum of the of distances raised to the power $z$ of every point to its closest
center is minimized. This encapsulates the famous $k$-median ($z=1$) and $k$-means ($z=2$) clustering problems.
Designing small-space sketches of the data that approximately preserves the cost of the solutions, also known as \emph{coresets},
has been an important research direction over the last 15 years. 

In this paper, we present a new, simple coreset framework that simultaneously improves upon the best known bounds for
a large variety of settings, ranging from Euclidean space, doubling metric, minor-free metric, and the general metric cases: with $\size = \min(\eps^{-2} + \eps^{-z}, k\eps^{-2})\polylog (\eps^{-1})$, this framework constructs coreset with size
  \begin{itemize}
  \item  $O\left(\size \cdot k (d + \log k) \right)$ in doubling metrics, improving upon the recent breakthrough of
    [Huang, Jiang, Li, Wu, FOCS' 18], who presented a coreset with size $O(k^3 d /\eps^2)$.
  \item  $O(\size \cdot k \cdot \min(d,\eps^{-2}\log k))$ in $d$-dimensional Euclidean space, improving on the recent results of 
    [Huang, Vishnoi, STOC' 20], who presented a coreset  of size \newline $O (k \log k \cdot \eps^{-2z}\cdot \min(d,\eps^{-2}\log k))$.
  \item  $O(\size \cdot k (t + \log k))$ for graphs with treewidth $t$,
    improving on [Baker, Braverman, Huang, Jiang, Krauthgamer, Wu, ICML'20], who presented a coreset of size $O(k^2 t/\eps^2)$ 
    for $z = 1$.
  \item $O\left(\size \cdot k\left(\log^2 k  + \frac{\log k}{\eps^4}\right)\right)$ for shortest paths metrics of graphs
    excluding a fixed minor.
    This improves on [Braverman, Jiang, Krauthgamer, Wu, SODA'21], who presented a coreset of size  $O(k^2 / \eps^4)$.
  \item Size $ O(\size \cdot k \log n )$ in general discrete metric spaces, improving on the results of
    [Feldman, Lamberg, STOC'11], who presented a coreset  of size $O(k\eps^{-2z}\log n\log k )$.
  \end{itemize}
  \medskip

A lower bound of $\Omega(\frac{k \log n}{\eps})$ for $k$-Median in general metric spaces [Baker, Braverman, Huang, Jiang, Krauthgamer, Wu, ICML'20] implies that in general metrics as well as metrics with doubling dimension $d$, our bounds are optimal up to a $\poly\log(1/\eps) / \eps$ factor. For graphs with treewidth $t$, the lower bound of $\Omega\left(\frac{k t}{\eps}\right)$ of [Baker, Braverman, Huang, Jiang, Krauthgamer, Wu, ICML'20] shows that our bounds are optimal up to the same factor.
\end{abstract}

\fi

\section{Introduction}
\label{sec:intro}

Center-based clustering problems are classic objectives for
the problem of computing a ``good'' partition of a 
set of points into $k$ parts, so that points that are ``close'' 
are in the same part. Finding a good clustering of a dataset helps extracting
important information from a dataset and center based clustering problems
have become the cornerstones of various data analysis approaches and machine 
learning techniques (see formal definition in Section~\ref{sec:prelims}).

Datasets used in practice are often huge, containing hundred of millions of points,
distributed, or evolving over time. Hence, in these settings classical heuristics (such as Lloyd or $k$-means++)
are lapsed; The size of the dataset forbids multiple passes over the input data and
finding a ``compact representation'' of the input data is of primary importance.
The method of choice for this is to compute a \emph{coreset},
i.e. a weighted set of points
of small size that can be used in place of the full input for algorithmic purposes.
More formally, for any $\eps > 0$, an $\eps$-coreset (referred to simply as coreset)
is a set $Q$ of points of the metric space such that any $\alpha$-approximation
to a clustering problem on $Q$, is a 
$\alpha(1+\eps)$-approximation to the clustering problem for the original
point set. Hence, a small coreset is a good compression of the full input set: one can
simply keep in memory a coreset and apply any given algorithm on the coreset rather than
on the input to speed up performances and reduce memory consumption.
Coreset constructions had been widely studied over the
last 15 years.

In this paper, we specifically focus on the $(k,z)$-clustering problem, which encapsulates $k$-median ($z=1$) and
$k$-means ($z=2$). Given two positive integers $k$ and $z$ and a metric space $(X, \dist)$, the $(k,z)$-clustering problem asks
for a set $\calS$ of $k$ points, called \emph{centers}, that minimizes
\[\cost(X, \calS) := \sum_{x \in X} \min_{s \in \calS} \dist(x, s)^z\]

%

The method of choice for designing coreset is \textit{importance sampling}, initiated by the seminal work of Chen~\cite{Chen09}. The basic approach is to devise a non-uniform sampling distribution which picks points proportionally to their cost contribution in an arbitrary constant factor approximation. In a nutshell, the current best-known analysis shows that, for a given set $\calS$ of $k$ centers, it happens with high probability that the sampled instance $\coreset$ with appropriate weights has roughly the same cost as the original instance, i.e. $\cost(\coreset, \calS) \in (1\pm\eps) \cost(X, \calS)$. Then, to show that the set $\coreset$ is an $\eps$-coreset, it is necessary to take a union-bound over these events for all possible set of $k$ centers.
Bounding the size of the union-bound is the main hurdle faced by this approach: indeed, there may be infinitely many possible set of centers.

The state-of-the-art analysis relies on VC-dimension to address this issue. Informally, the VC dimension is a complexity measure of a range space, denoting the cardinality of the largest set such that all subsets are included in the range space. 
The application to clustering considers weighted range spaces, where each point is weighted by its relative contribution to the cost of a given clustering\footnote{For more on these notions, we refer to \cite{FeldmanSS20}.}.
In metric spaces where the weighted range space induced by distances to $k$ centers has VC-dimension $D$, it can be shown that taking $O_{\eps, z}(k\cdot D\log k)$ samples yields a coreset~\cite{FeldmanSS20}, although tighter bounds are achievable in certain cases. For instance, in $d$ dimensional Euclidean spaces $D$ is in $O(kd\log k)$~\cite{BachemLH017}, which would yield coresets of size $O_{\eps, z}(k^2\cdot d\log^2 k)$, but Huang and Vishnoi~\cite{huang2020coresets} showed the existence of a coreset with $O(k\cdot \log^2 k \cdot \eps^{-2z-2})$ points.

This analysis was proven powerful in various metric spaces, such as doubling spaces by Huang, Jiang, Li and Wu~\cite{HuangJLW18}, graphs of bounded treewidth by Baker, Braverman, Huang, Jiang, Krauthgamer, Wu~\cite{baker2020coresets} or the shortest-path metric of a graph excluding a fixed minor by Braverman, Jiang, Krauthgamer and Wu~\cite{BravermanJKW21}. 
However, range spaces of even heavily constrained metrics do not necessarily have small VC-dimension (e.g. bounded doubling dimension does not imply bounded VC-dimension or vice versa~\cite{HuangJLW18,LiL06a}), and applying previous techniques requires heavy additional machinery to adapt the VC-dimension approach to them. Moreover, the bounds provided are far from  the bound obtained for Euclidean spaces:
their dependency in $k$ is at least $\Omega(k^2)$, leaving a significant gap to the best lower bounds of $\Omega(k)$.
We thus ask: 

\textbf{Question.}\emph{  Is it possible to design coresets whose size are near-linear in $k$ for doubling metrics, minor-free metrics, bounded-treewidth metrics?
  Are the current roadblocks specific to the analysis through VC-dimension, or inherent to the problem?}



To answer positively these questions, we present a new framework to analyse importance sampling. Its analysis stems from first principles, and it can be applied in a black-box fashion to any metric space that admits an \emph{approximate centroid set} (see \cref{def:centroid-set}) of bounded size. We show that all previously mentioned spaces satisfy this condition, and our construction improves on the best-known coreset size. 
More precisely, we recover (and improve) all previous results for $(k, z)$-clustering such as Euclidean spaces, $\ell_p$ spaces for $p \in [1,2)$, finite $n$-point metrics, while also giving the first coresets with size near-linear in $k$ and $\eps^{-z}$ for a number of other metrics such as doubling spaces, minor free metrics, and graphs with bounded treewidth.

\subsection{Our Results} 
Our framework requires the existence of a particular discretization of the set of possible centers, as described in the following definition.
We show in the latter sections that this is indeed the case for all the metric spaces mentioned so far.

\begin{definition}\label{def:centroid-set}
Let $(X, \dist)$ be a metric space, $P \subseteq X$ a set of clients and two positive integers $k$ and $z$. Let $\eps > 0$ be a precision parameter. Given a set of centers $\greedy$, a set $\cand$ is an \emph{$\greedy$-approximate centroid set} for $(k, z)$-clustering on $P$ if it satisfies the following property.

For every set of $k$ centers $\calS \in X^k$, there exists $\tilde \calS \in \cand^k$ such that for all points $p \in P$ that satisfies either $\cost(p, \calS) \leq \left(\frac{8z}{\eps}\right)^z \cost(p, \greedy)$ or $\cost(p, \tilde \calS) \leq \left(\frac{8z}{\eps}\right)^z \cost(p, \greedy)$, it holds
\[|\cost(p, \calS) - \cost(p, \tilde{\calS})| \leq \frac{\eps}{z\log(z/\eps)}\left(\cost(p, \calS) + \cost(p, \greedy)\right),\] 
\end{definition}

This definition is slightly different from Matousek's one~\cite{Mat00}, in that we seek to preserve distances only for interesting points, and allow an error $\eps \cost(p, \greedy)$. This is crucial in some of our applications.

\begin{theorem}\label{thm:main}
Let $(X, \dist)$ be a metric space, $P \subseteq X$ a set of clients with $n$ distinct points and two positive integers $k$ and $z$. Let $\eps > 0$ be a precision parameter. Let also $\greedy$ be a constant-factor approximation for $(k,z)$-clustering on $P$.

 Suppose there exists an $\greedy$-approximate centroid set $\cand$ for $(k, z)$-clustering on $P$.
  Then, there exists an algorithm running
  in time $ O(n)$ that constructs with probability at least $1-\pi$
  a
  coreset of size 
\[O\left(\frac{2^{O(z\log z)}\cdot\log^4 1/\eps}{\min(\eps^2, \eps^z)}\left(k \log |\cand| + \log \log (1/\eps) + \log(1/\pi)\right)\right)\]  
   with positive weights for the $(k,z)$-clustering problem.
\end{theorem}

When applying this theorem to particular metric spaces, the running time
 is dominated by the construction of the constant-factor approximation $\greedy$, which can be done for instance in $\tilde O (k|P|)$ given oracle access to the distances using~\cite{MettuP04}\footnote{Although initially stated for $z = 1$ only, this algorithm works for general $z$ as stressed in \cite{huang2020coresets}}.

If one wishes to trade a factor $\eps^{-z}$ for a factor $k$, we also present coresets of size \\  $O(k^2\cdot 2^{O(z)} \frac{\log^3(1/\eps)}{\eps^2}\left(\log k + \log |\cand| + \log(1/\pi)\right)$, as explained in \cref{sec:ksquare}. 

We apply this theorem to several metric spaces, achieving the following (simplified) size bounds (we ignore $\poly \log(1/\eps)$ and $2^{O(z\log z)}$ factors): let $\size = \min (\eps^{-2} + \eps^{-z}, k\eps^{-2})$, see also \cref{table:core}.
\begin{itemize}
\item $O\left(\size \cdot k\left(d + \log k \right)\right)$ for metric spaces with doubling dimension $d$. This improves over the $O(k^3 d \eps^{-2})$ from Huang et al.~\cite{HuangJLW18}. See \cref{cor:coreset-doubling}.
\item Since general discrete metric spaces have doubling dimension $O(\log n)$, this yields coreset of size $O\left(\size\cdot k \log n\right)$. This improves on the bound from Feldman and Langberg~\cite{FeldmanL11} $O\left(\eps^{-2z}k \log k \log n\right)$.
\item $O\left(\size \cdot k\eps^{-2} \cdot \log k\right)$ for Euclidean spaces, see \cref{cor:coreset-euclidean}. This improves on the recent result from Huand and Vishnoi~\cite{huang2020coresets}, who achieve $O\left(\eps^{-2z - 2}k \log^2 k\right)$.
\item $O\left(\size \cdot k \left( \log^2 k  + \frac{\log k}{\eps^4}\right)\right)$ for a family of graphs excluding a fixed minor, see \cref{cor:coreset-minor}. This improves on Braverman et al.~\cite{BravermanJKW21}, whose coreset has size $\widetilde{O}(k^2 / \eps^4)$.
\item $O\left(\size \cdot k \left(\log^2 k + \frac{\log k}{\eps^3}\right)\right)$ for Planar Graphs, which is a particular family excluding a fixed minor for which we can save a $1/\eps$ factor and present a simpler, instructive proof.
\item $O\left(\size \cdot k \left(t + \log k \right)\right)$ in graphs with treewidth $t$, see \cref{cor:coreset-treewidth}. 
This improves upon the work of Baker et al.~\cite{baker2020coresets}, that construct coreset with size $\widetilde{O}(k^2 t/ \eps^2) $.
  \item  $O(k \eps^{-2z}\cdot \min(d,\eps^{-2}\log k))$ in $\mathbb{R}^d$ with $\ell_p$ distance, for $p \in [1, 2)$, see \cref{cor:coreset-lp}. This improves on  
    Huang and Vishnoi~\cite{huang2020coresets}, who presented a coreset  of size $O (k \log k \cdot \eps^{-4z}\cdot \min(d,\eps^{-2}\log k))$. 
\end{itemize}

We note the lower bound $\Omega(\frac{k \log n}{\eps})$ for $k$-Median in general metric spaces from \cite{baker2020coresets}. This means that in the case of metrics with doubling dimension $d$, our bounds are optimal up to a $\poly\log(1/\eps) / \eps$ factor. For graphs with treewidth $t$, another lower bound of $\Omega\left(\frac{k t}{\eps}\right)$ from \cite{baker2020coresets} shows that our bounds are optimal up to the same factor.

\begin{table*}[ht!]
\begin{center}
\begin{tabular}{r|c}
Reference & Size (Number of Points) 
\\
\hline
\hline
\multicolumn{2}{l}{\textbf{Euclidean space}} \\\hline\hline
Har-Peled, Mazumdar (STOC'04)~\cite{HaM04} & $O(k\cdot \varepsilon^{-d}\cdot \log n)$ 
\\
\hline
Har-Peled, Kushal (DCG'07)~\cite{HaK07} & $O(k^3 \cdot \varepsilon^{-(d+1)})$  
\\\hline
Chen (Sicomp'09)~\cite{Chen09} & $O(k^2\cdot d\cdot \varepsilon^{-2}\log n)$ 
\\\hline
Langberg, Schulman (SODA'10)~\cite{LS10} &$O(k^3\cdot d^2\cdot \varepsilon^{-2})$ 
\\\hline
Feldman, Langberg (STOC'11)~\cite{FeldmanL11} & $O(k\cdot d \cdot \log k \cdot \varepsilon^{-2z})$ 
\\\hline
Feldman, Schmidt, Sohler (Sicomp'20)~\cite{FeldmanSS20} & $O(k^3\cdot \log k \cdot \varepsilon^{-4})$ 
\\\hline
Sohler and Woodruff (FOCS'18)~\cite{SohlerW18} & $O(k^2\cdot \log k\cdot \varepsilon^{-O(z)})$  
\\\hline
Becchetti, Bury, Cohen-Addad, & \multirow{2}{*}{$O(k\cdot \log^2 k \cdot \varepsilon^{-8})$} \\
Grandoni, Schwiegelshohn~(STOC'19)~\cite{BecchettiBC0S19} &  
\\
\hline
Huang, Vishnoi (STOC'20)~\cite{huang2020coresets} & $O(k\cdot \log^2 k \cdot \varepsilon^{-2-2z})$ 
\\
\hline
Braverman, Jiang, Krauthgamer, Wu (SODA'21)~\cite{BravermanJKW21}& $\tilde O (k^2 \cdot \eps^{-4})$ 
\\
\hline
\large \textbf{This paper} & $O(k \cdot \log k \cdot \eps^{-2-\max(2,z)})$ 
\\
\hline\hline
\multicolumn{2}{l}{\textbf{General $n$-point metrics, $ddim$ denotes the doubling dimension}} \\\hline\hline
Chen (Sicomp'09)~\cite{Chen09} & $O(k^2\cdot\varepsilon^{-2}\cdot\log^2 n)$ 
\\\hline
Feldman, Langberg (STOC'11)~\cite{FeldmanL11} & $O(k\cdot \log k\cdot  \log n \cdot \varepsilon^{-2z})$  
\\\hline
Huang, Jiang, Li, Wu~(FOCS'18)~\cite{HuangJLW18} & $O(k^3\cdot ddim\cdot \varepsilon^{-2})$ 
\\\hline
\large \textbf{This paper} & $O(k\cdot(ddim+\log k)\cdot \varepsilon^{-\max(2,z)})$ 
\\
\hline
\large \textbf{This paper} & $O(k\cdot\log n \cdot \varepsilon^{-\max(2,z)})$ 
\\
\hline \hline
\multicolumn{2}{l}{\textbf{Graph with $n$ vertices, $t$ denotes the treewidth}} \\
\hline\hline
Baker, Braverman, Huang, Jiang, & \multirow{2}{*}{$\tilde O (k^2\cdot t/\eps^2) $} \\
Krauthgamer, Wu (ICML'20)~\cite{baker2020coresets} &  
\\
\hline
\large \textbf{This paper} & $O(k\cdot(t + \log k)\cdot\varepsilon^{-\max(2,z)})$ \\
\hline \hline
\multicolumn{2}{l}{\textbf{Graph with $n$ vertices, excluding a fixed minor}} \\
\hline\hline
Bravermann Jian, Krauthgamer, Wu (SODA'21)~\cite{BravermanJKW21} & $\tilde O (k^2 \cdot \eps^{-4})$ 
\\
\hline
\large \textbf{This paper} & $O\left(k\cdot(\log^2 k + \frac{\log k}{\eps^4})\cdot \varepsilon^{-\max(2,z)}\right)$ \\
\end{tabular}
\end{center}
\caption{Comparison of coreset sizes for $(k,z)$-Clustering in various metrics. 
Dependencies on $2^{O(z)}$ and $\text{polylog}\varepsilon^{-1}$ are omitted from all references. Additionally, we may trade a factor $\varepsilon^{-z+2}$ for a factor $k$ in any construction with $z>2$.
\cite{HaK07,HaM04} only applies to $k$-means and $k$-median,~\cite{BecchettiBC0S19,FeldmanSS20} only applies to $k$-means. 
\cite{SohlerW18} runs in exponential time, which has been addressed by Feng et al.~\cite{FKW19}.
Aside from~\cite{HaK07,HaM04}, the algorithms are randomized and succeed with constant probability. 
Although the results are claimed only for $k$-Median in \cite{baker2020coresets}, it seems that they can be generalized to any power. The main difference is in the computation of a constant factor approximation.}
\label{table:core}
\end{table*}

\subsection{Overview of Our Techniques}

Our proof is arguably from first principles. We now give a quick overview of its ingredients.
The approach consists in first reducing to a well structured instance, that consists of a set of centers $\greedy$ inducing $k$ clusters, all having roughly the same costs, and where every point is at the same distance of $\greedy$, up to a factor 2. Then we show it is enough to perform importance sampling on all these clusters.

\paragraph{Reducing to a structured instance.}

Like most coreset constructions, we initially compute a constant factor approximation $\greedy$ to the problem. We then deviate from previous importance sampling algorithms
by partitioning points into groups such that the following conditions are satisfied, for a given group $G$:
\begin{itemize}
\item For all clusters, the cost of the intersection of the cluster with the group is at least half the average; i.e. $\forall C_i,~ \cost(C_i \cap G, \greedy) \geq \frac{\cost(G, \greedy)}{2k}$.
\item In every cluster $C_i$, there exists $r_{G, i}$ such that the points in the intersection of the cluster with the group cost $r_{G, i}$ (up to constant factors), i.e. $\forall p \in C_i \cap G, \cost(p, \greedy) = \Theta(r_{G_i})$.
\end{itemize}

We then compute coresets for each group and output the union. In some sense,
this preprocessing step identifies canonical instances for coresets;
any algorithm that produces improved coresets
for instances satisfying the aforementioned regularity condition can be combined with our preprocessing steps to
produce improved coreset in general.

%

\paragraph{Importance Sampling in Groups.}
The first technical challenge is to analyse the importance sampling procedure for structured instances.

The arguably simplest way to attempt to analyse importance sampling is by first showing that for any fixed solution $\calS$ we need a set $\coreset$ of $\delta$ samples to show that with good enough probability
\begin{equation}
\label{eq:estimatorIntro}
 \sum_{p\in \coreset} \cost(p,\calS)\frac{\cost(G,\greedy)}{\cost(p,\greedy)\cdot \delta} = (1\pm \varepsilon) \cdot \cost(G, \calS),
\end{equation}
and then applying a union bound over the validity of \cref{eq:estimatorIntro} for all solutions $S$. This union bound is typically achieved via the VC-dimension.

Using this simple estimator, most analyses of importance sampling procedures require a sample size of at least $k$ points to approximate the cost of a single given solution. To illustrate this, consider an instance where a single cluster $C$ is isolated from all the others. Clearly, if we do not place a center close to $C$, the cost will be extremely large, requiring some point of $C$ to be contained in the sample. One way to remedy this is by picking a point $p'$ proportionate to $\frac{\cost(p',\greedy)}{\cost(\greedy)} + \frac{1}{|C_i|}$ rather than $\frac{\cost(p',\greedy)}{\cost(\greedy)}$, where $C_i$ is the cluster to which $p'$ is assigned, see for instance~\cite{FeldmanSS20}. This analysis always leads to coreset of size quadratic in $k$ at best\footnote{A linear dependency on $k$ can be achieved using a different analysis, see~\cite{FeldmanL11,huang2020coresets} for examples. This approach does not seem to generalize to arbitrary metrics.}. 
Our analysis of importance sampling for structured instances will allow us to bypass both the quadratic dependencies on $k$, and the need of a bound on the VC-dimension of the range space.


Our high level idea is to use two union bounds. The first one will deal with clusters that are very expensive compared to their cost in $\greedy$. The second one will focus on solutions in which clusters have roughly the same cost as they do in $\greedy$.
For the former case, we observe that if a cluster $C_i$ is served by a center in solution $\calS$ that is very far away, then we can easily bound its cost in $\calS$ as long as our sample approximates the size of every cluster.
Specifically, assume that there exists a point $p$ in $C_i$ with distance to $\calS$ at least $ \Omega(1)\cdot \varepsilon^{-1}\cdot \dist(p, c_i)$. Then, since we are working with structured instances, all points of $C_i$ are roughly at the same distance of $c_i$ and that this distance is negligible compared to $\dist(p, \calS)$, all points of $C_i$ are nearly at the same distance of $\calS$.
Conditioned on the event $\calE$ that the sample $\coreset$ 
preserves the size of all clusters, the cost of $C_i$ in solution $\calS$ is preserved as well. Note that this event $\calE$ is independent of the solution $\calS$ and thus we require no enumeration of solutions to preserve the cost of expensive clusters. Proving that $\calE$ holds is a straightforward application of concentration bounds.
   

The second observation is that points with $\dist(p, \calS) \leq \nicefrac{\eps}{z} \cdot \dist(p, \greedy)$ are so cheap that their cost is preserved by the sampling with an error at most $\eps\cdot \cost(\greedy)$. Indeed, their cost in $\calS$ cannot be more than $\eps \cdot \cost(\greedy)$: it is easy to show that the same bound holds for the coreset. 

The intermediate cases, i.e. solutions in which $\calS$ serves clusters at distances further than $\varepsilon/z\cdot\dist(p,\greedy)$, but not so far as to simply use event $\calE$ to bound the cost, is the hardest part of the analysis.
Using a geometric series, we can split the cost ranges into into $\log \frac{z}{\varepsilon^2} \in O(z\log \varepsilon^{-1})$ groups by powers of two.
Due to working with a structured instance, the points within such a group have equal distances, up to a constant factors. This also implies that the cost in such a group is equal, up to a factor of $2^{O(z)}$.
The overall variance of the cost estimator is then of the order $\max_p \left(\varepsilon^{-1}\cdot \dist(p,\greedy)\right)^z\cdot \frac{\cost(\greedy)}{\cost(p,\greedy)} \in O(\varepsilon^{-z})$.
Thus, standard concentration bounds give an additive error of $\varepsilon\cdot (\cost(\greedy)+\cost(\calS))$ with at most $O(\varepsilon^{-2-z})$ many samples for every group. 

To improve this to $O(\varepsilon^{-z})$, we use a different estimator defined as follows. For every cluster $C_i$, let $q_i$ be the
point of $C_i$ that is the closest to $\calS$. We then consider 
\begin{eqnarray}
\label{eq:estimator2}
& & \sum_{p\in C_i \cap \coreset} \left(\cost(p,\calS)-\cost(q_i,\calS)\right)\cdot \frac{\cost(G,\greedy)}{\cost(p,\greedy)\cdot \delta} \\
 \label{eq:estimator3}
& + & \sum_{p\in C_i \cap \coreset} \cost(q_i,\calS)\cdot\frac{\cost(G,\greedy)}{\cost(p,\greedy)\cdot \delta}
\end{eqnarray}

Conditioned on event $\calE$, the estimator in Equation~\ref{eq:estimator3} is always concentrated around its expectation, as $\cost(q_i,\calS)$ is fixed for $\calS$.
The first estimator in Equation~\ref{eq:estimator2} now has a reduced variance. Specifically, at the border cases of points at distance $\Theta(1/\varepsilon) \dist(p, \greedy)$ of $\calS$, the Estimator~\ref{eq:estimator2} has variance at most $O(1)\cdot\max(\varepsilon^{-2},\varepsilon^{-z})\cdot \cost(\greedy)\cdot \cost(\calS)$, which ultimately allows us to show that $O(\varepsilon^{-2}+\varepsilon^{-z})$ samples are enough to achieve an additive error of $\varepsilon\cdot \left(\cost(\calS)+\cost(\greedy)\right)$.
This technique is somewhat related to (and inspired by) chaining arguments (see e.g. Talagrand~\cite{talagrand1996majorizing} for more on chaining). The key difference is while chaining is generally applied to improve over basic union bounds, our estimator is designed to reduce the variance.

\paragraph{Preserving the Cost of Points not in Well-Structured Groups}

Unfortunately, it is not possible to decompose the entire point set into groups. Given an initial solution $\greedy$ and a cluster $C\in \greedy$, this is possible for all the points at distance at most $\varepsilon^{-O(z)}\cdot \frac{\cost(C,c)}{|C|}$. 
The remaining points are now both far from their respective center in $\greedy$ and, due to Markov's inequality, only a small fraction of the point set. In the following, let $P_{far}$ denote these points.

For any given subset of these far away points and a candidate solution $\calS$, now use that either the points pay at most what they do in $\greedy$, or an increase in their cost significantly increases the overall cost. In the former case, standard sensitivity sampling preserves the cost with a very small sample size. In the latter case, a significant cost by a point $p$ in $P_{far}$ also implies that all points close to the center $c$ serving $p$ in $\greedy$ have to significantly increase the cost.

\paragraph{A Union-bound to Preserve all Solutions } 
As pictured in the previous paragraphs, the cost of points with either very small or very large distance to $\calS$ is preserved for any solution $\calS$ with high probability. 

The guarantee we have for interesting points is weaker: their cost is preserved by the coreset with high probability for any \textit{fixed} solution $\calS$. Hence, for this to hold for any solution, we need to take a union-bound over the probability of failure for all possible solution $\calS$. However, the union-bound is necessary only for these interesting points 
: this explains the introduction of the approximate centroid set in \cref{def:centroid-set}. Assuming the existence of a set $\cand$ such as in \cref{def:centroid-set}, one can take a union-bound over the failure of the construction for all set of $k$ centers in $\cand^k$ to ensure that the cost of interesting points is preserved for all these solutions. To extend this result to \textit{any} solution $\calS$, one can take the set of $k$ points $\tilde \calS$ in $\cand^k$ that approximates best $\calS$, and relate the cost of interesting points in $\calS$ to their cost in $\tilde \calS$ with a tiny error. Since the cost of interesting points in $\tilde \calS$ is preserved in the coreset, the cost of these points in $\calS$ is preserved as well.

We briefly picture now how to get approximate centroid sets for specific metrics.
We are looking for a set $\cand$ with the following property: for every solution $\calS$, there exists a $k$-tuple $\tilde \calS \in \cand^k$ such that for every point $p$ with $\dist(p, \calS) \leq \eps^{-1} \dist(p, \greedy)$ in a given cluster $C$ of $\greedy$, $|\cost(p, \calS) - \cost(p, \tilde{\calS})| \leq \eps \left(\cost(p, \greedy) + \cost(p, \calS)\right)$. We call such points \textit{interesting}.

\subparagraph{Metrics with doubling dimension $d$:}
 $\cand$ is simply constructed taking \emph{nets} around each input point. A $\gamma$-net of a metric space is a set of points that are at least at distance $\gamma$ from each other, and such that each point of the metric is at distance at most $\gamma$ from the net. The existence of $\gamma$-nets of small size is one of the key properties of doubling metrics (see \cref{prop:doub:net}). 
 For every point $p$, $\cand$ contains an $\eps \cost(p, \greedy)$-net of the points at distance at most $\frac{8z}{\eps}\cdot \cost(p, \greedy)$ from $p$. If $p$ is an interesting point, there is therefore a center of $\cand$ close to its center in $\calS$. 
 
 However, this only shows that centers from the solution $\tilde \calS \in \cand^k$ are closer than those of $\calS$. Showing that none gets too close is a different ballgame. We will see two ways of achieving it. The first one, that we apply for the doubling, treewidth and planar case, is based on the following observation: if a center $s \in \calS$ is replaced by a center $\tilde s$, that is way closer to a point $p$ than $s$, then $s$ can be discarded in the first place and be replaced by the center serving $p$. This is formalized in \cref{lem:noproblem}. The other way of ensuring that no center from $\tilde \calS$ gets to close to a point in $p$ is based on guessing distances from points in $\calS$ to input points. It can be applied more broadly than \cref{lem:noproblem}, but yields larger centroid sets. We will use it only for minor-excluded graphs, for which \cref{lem:noproblem} cannot be applied.


\subparagraph{Graphs with treewidth $t$:} The construction of $\cand$ is not as easy in graph metrics: we use the existence of small-size \textit{separators}, building on ideas from Baker et al.~\cite{baker2020coresets}. Fix a solution $\calS$, and suppose that all interesting points are in a region $R$ of the graph, such that the boundary $B$ of $R$ is made of a constant number of vertices. Fix a center $c \in \calS$, and suppose $c$ is not in $R$. Then, to preserve the cost of interesting points, it is enough to have a center $c'$ at the same distance to all points in the boundary $B$ as $c$.

$\cand$ is therefore constructed as follows: for a point $p$, its distance tuple to $B = \{b_1, ..., b_{|B|}\}$ is the tuple $(d_1, ..., d_{|B|})$, where $d_i = \dist(p, b_i)$ is the distance to $b_i$. For every distance tuple to $B$, $\cand$ contains one point having approximately that distance tuple to $B$. 

Let $\tilde c$ be the point of $\cand$ having approximately the same distance tuple to $B$ as $c$: this ensures that $\forall p, ~\cost(p, c) \approx \cost(p, \tilde c)$.

It is however necessary to limit the size of $\cand$. For that, we approximate the distances to $B$. This can be done for interesting points $p$ as follows: since we have $\dist(p, c) \leq \eps^{-1}\dist(p, \greedy)$, rouding the distances to their closest multiple of $\eps \dist(p, \greedy)$ ensures that there are only $O(1/\eps^2)$ possibilities, and adds an error $\eps \cost(p, \greedy)$. We show in \cref{sec:centroid-treewidth} how to make this argument formal, 
and how to remove the assumption that all interesting points are in the same region.

\subparagraph{In minor-excluded graphs} this class of graphs, that includes planar graphs, admits as well small-size shortest-path separators. A construction similar in spirit to the one for treewidth is therefore possible, as presented in \cref{sec:centroid-minor}. This builds on the work of Braverman et al.~\cite{BravermanJKW21}.

However, due to the nature of the separator -- which are small sets of paths, and not simply small sets of vertices -- one cannot apply the idea of  \cref{lem:noproblem} to show that no center gets too close. Instead, we will guess the distance from input points to any point in $\calS$, allowing to construct $\tilde \calS$ with the same distances. Of course, this mere idea requires way too many guesses to have a small set $\cand$: we see in \cref{sec:centroid-minor} how to make it work properly.

We start the section by showing two preprocessing lemmas: the first one is \cref{lem:noproblem}, as described above. The second one allows to apply \cref{thm:main} in the case the input set is weighted, so that we can assume the input has only $\poly(k, \eps^{-1})$ many distinct points, by first computing a non-optimal coreset.

\subsection{Roadmap}
The paper is organized as follow: after defining the concepts used in the paper, we present formally the algorithm in \cref{sec:algo}. We then describe the construction of a coreset for a structured instance in \cref{sec:sample}, and the reduction to such an instance in \cref{sec:preprocess}. Finally, we show the existence of approximate centroid set in various metric spaces in \cref{sec:centroid}. We furthermore explain the dimension reduction technique leading to our result for Euclidean spaces in \cref{sec:dimension}, and the $O(k^2 \eps^{-2})$ construction in \cref{sec:ksquare}. A deeper description of related work is made in \cref{sec:related}.

\section{Related Work}
\label{sec:related}

We already surveyed most of the relevant bounds for coresets for $k$-means and $k$-median.
A complete overview over all of these bounds is given in Table~\ref{table:core}, further pointers to coreset literature can be found in surveys~\cite{MunteanuS18}.
For the remainder of the section, we highlight differences to previous techniques.

The early coreset results mainly considered input data embedded in constant dimensional Euclidean spaces~\cite{FrahlS2005,HaK07,HaM01}. 
These coresets relied on low-dimensional geometric decompositions inducing coresets of sizes typically of order at least $k\cdot \varepsilon^{-d}$. These techniques were replaced by \emph{importance sampling} schemes, initiated by the seminal work of Chen~\cite{Chen09}. The basic approach is to devise a non-uniform sampling distribution which picks points proportionately to their impact in a given constant factor approximation.  A significant advantage of importance sampling over other techniques is that it generalizes to non-Euclidean metrics. While the early coreset papers~\cite{HaK07,HaM04} were indeed heavily reliant on the structure of Euclidean spaces, Chen gave the first coreset of size $O(k^2\varepsilon^{-2} \log^2 n)$ for general $n$-point metrics.

\paragraph{Coresets via Bounded VC-Dimension}
The state of the art importance sampling techniques in Euclidean spaces are based on reducing the problem of constructing a coreset to constructing an $\varepsilon$-net in a range space of bounded VC-dimension\footnote{Strictly speaking, one has to use a generalization of VC-dimension known as the pseudo dimension. The interested reader is refereed to Pollard's book~\cite{pollard2012} for details.}. Li, Long and Srinivasan~\cite{LiLS01} showed that if the VC-dimension is bounded by $D$, an $\varepsilon$-approximation of size $O(\frac{D}{\varepsilon^{2}})$ exists. The remarkable aspect of these bounds is that they are independent of the number of input points.
To apply the reduction, we need a bound on the VC-dimension for the range space induced by the intersection of metric balls centered around $k$ points in a 
$d$-dimensional Euclidean space. For Euclidean $k$-means and $k$-median, an upper bound of $D\in O(kd\log k)$ is implicit in the work of~\cite{BlumerEHW89} and Eisenstat and Angluin~\cite{EisenstatA07}. This bound was recently shown to be tight by Csikos, Mustafa and Kupavskii~\cite{CsikosMK19}. The dependency on $d$ may be replaced with a dependency on $\log k$, as explained in more detail in Section~\ref{sec:dimension}. Thus $O(k\log^2 k)$ is a natural barrier for known techniques in Euclidean spaces.


\paragraph{VC-Dimension and Doubling Dimension}
A further complication arises when attempting to extend sampling techniques for bounded VC-dimension in range spaces of bounded doubling dimension $d$. While the two notions share certain similarities and are asymptotically identical for the range space induced by the intersection of balls in in Euclidean spaces, the two quantities are incomparable in general. For instance, Li and Long proved the existence of a range space with constant VC dimension and unbounded doubling dimension~\cite{LiL06a}. Conversely, \cite{HuangJLW18} also showed that a bound on the doubling dimension does not imply a bound on the VC-dimension. Nevertheless, by carefully distorting the metric they were able to prove that a related quantity known as the shattering dimension can be bounded, yielding the first coresets for bounded doubling dimension independent of $n$. Even so, their bound $\tilde{O}(k^3 d \varepsilon^{-2})$ is still far from what is currently achievable in Euclidean spaces.

Similarly, the construction from \cite{baker2020coresets} for graphs with bounded treewidth uses that a graph of treewidth $t$ has shattering dimension $O(t)$. They use this fact to get coreset for $k$-Median, of size $\tilde O(k^3t/\eps^2) $. For excluded-minor graphs, \cite{BravermanJKW21} proceeds similarly, but need an additional iterative procedure: they first show that in an excluded-minor graph, a subset $X$ of the vertices has coreset of size $O_{k, \eps}(\log |X|)$, using the shattering-dimension techniques.
They show then how to iterate this construction (using that "a coreset of a coreset is a coreset") to remove dependency in $|X|$. This iterative procedure is of independent interest, and we use it as well for bounded treewidth and excluded-minor settings. 

\paragraph{Further Related Work}
So far we only described works that aim at giving better coreset construction for unconstrained $k$-median and $k$-means in some metric space. Nevertheless, there is a rich literature on further related questions. As a tool for data compression, coresets feature heavily in streaming literature. Some papers consider a slightly weaker guarantee of summarizing the data set such that a $(1+\varepsilon)$ approximation can be maintained and extracted. Such notions are often referred to as \emph{weak coresets} or streaming coresets, see~\cite{FeldmanL11,FMS07}. Further papers focus on maintaining coresets with little overhead in various streaming and distributed models, see~\cite{BalcanEL13,BravermanFLR19,BravermanFLSY17,FrahlS2005,
FGSSS13}. 
Other related work considers generalizations of $k$-median and $k$-means by either adding capacity constraints~\cite{Cohen-AddadL19,HuangJV19,SSS19}, or considering more general objective functions~\cite{BachemLL18,BravermanJKW19}. Coresets have also been studied for many other problems: we cite non-comprehensively Determinant Maximization \cite{IndykMGR20}, Diversity Maximization \cite{ceccarello2018fast, IndykMMM14} logistic regression~\cite{huggins2016coresets,MunteanuSSW18}, dependency networks~\cite{MK18}, or low‐rank approximation~\cite{maalouf2019fast}.

\section{Preliminaries}
\label{sec:prelims}

\subsection{Problem Definitions}
\label{sec:problemdef}
Given an ambient metric space $(X, \dist)$, a set of points $P \subseteq X$ called \emph{clients}, and positive integers $k$ and $z$, the goal of the \emph{$(k,z)$-clustering problem} 
is to output a set $\calS$ of $k$ \emph{centers} (or \emph{facilities}) chosen in $X$ 
that minimizes 
\[\sum_{p \in P} \min_{c \in \calS}(\dist(p,c))^z\]


\begin{definition}
An $\eps$-coreset for the $(k,z)$-clustering problem in a metric space $(X, \dist)$ is a weighted subset $\coreset$ of $X$ with weights $w : \coreset \rightarrow \R_+$ such that, for any set $\calS \subset X$, $|\calS| = k$, 
\[|\sum_{p \in X} \cost(p, \calS) - \sum_{p \in \coreset} w(p)\cost(p, \calS)| \leq \eps \cdot \sum_{p \in X} \cost(p, \calS).\]
\end{definition}

Given a set of point $P$ with weights $w : P \rightarrow \mathbb{R}^+$ on a metric space $I = (X, \dist)$ and a solution $\calS$, we define $\cost(P, \calS) := \sum_{p \in P} w(p)\cost(p, \calS)$ and, in the case where $P$ contains all the points of the metric space, we define $\cost(\calS) := \cost(P, \calS)$. 

We will also make use of the following lemma, to have a weaker version of the triangle inequality for $k$-Means and more general distances. Proofs of this lemma (and variants thereof) can be found in~\cite{BecchettiBC0S19,Cohen-AddadS17,FeldmanSS20,MakarychevMR19,SohlerW18}. For completeness, we provide a proof in the appendix.

\begin{restatable}[Triangle Inequality for Powers]{lemma}{weaktri}
\label{lem:weaktri}
Let $a,b,c$ be an arbitrary set of points in a metric space with distance function $d$ and let $z$ be a positive integer. Then for any $\varepsilon>0$
\begin{align*}
d(a,b)^z &\leq (1+\varepsilon)^{z-1} d(a,c)^z + \left(\frac{1+\varepsilon}{\varepsilon}\right)^{z-1} d(b,c)^z\\
\left\vert d(a,S)^z - d(b, S)^z\right\vert &\leq \varepsilon \cdot d(a,S)^z + \left(\frac{2z+\varepsilon}{\varepsilon}\right)^{z-1} d(a, b)^z.
\end{align*}
\end{restatable}

\subsection{From Weighted to Unweighted Inputs}
We start by showing a simple reduction from weighted to unweighted inputs. Essentially, we convert a point with weight $w$ to $w$ copies of the point. 

\begin{corollary}\label{cor:weighted}
  Let $\eps, \pi > 0$.
  Let $(X, \dist)$ be a metric space, $P$ a set of clients with weights $w : P \rightarrow \R^+$
  and two positive integers $k$ and $z$. 
  Let also $\greedy$ be a constant-factor approximation for $(k,z)$-clustering on $P$ with weights.
  
  Suppose there exists a $\greedy$-approximate centroid set, denoted $\cand$.
  Then, there exists an algorithm running
  in time $O(|P|)$ that constructs with probability at least $1-\pi$
  a positively-weighted
  coreset of size 
\[O\left(\frac{2^{O(z\log z)}\cdot\log^4 1/\eps}{\min(\eps^3, \eps^z)}
  \left(k \log |\cand| + \log \log (1/\eps) + \log(1/\pi)\right)\right)\]  
   for the $(k,z)$-clustering problem
  on $P$ with weights.
\end{corollary}
\begin{proof}
  
  We start by making all weights integers: let $w_{min} = \min_{p \in P} w(p)$, and $\tilde w(p) = \left\lfloor 2 \frac{w(p)}{\eps w_{min}}\right\rfloor$.
 This definition ensures that
\[\forall p,~|w(p) - \frac{\eps w_{min} }{2}\cdot \tilde w(p)| \leq \frac{\eps}{2} w_{min} \leq \frac{\eps}{2} w(p).\]
We denote $\tilde P$ the set of points $P$ with weight $\tilde w$. First, we note that for any solution $\calS$,
\begin{align*}
\left| \cost(P, \calS) - \eps w_{min} \cost(\tilde P, \calS)\right| \leq \frac{\eps}{2} \cost(P, \calS).
\end{align*}
Hence, it is enough to find an $\eps/2$-coreset for $\tilde P$, and then scale the coreset weights of the coreset points by $\eps w_{min} / 2$.
We have that the weights in $\tilde P$ are integers: a weighted point can therefore be considered as multiple copies of the same points.

By the previous equation, $\greedy$ is a constant-factor approximation for $\tilde P$ as well. The definition of a centroid set does not depend on weights, so $\cand$ is a $\greedy$-centroid set for $\tilde P$ as well. Hence, we can apply \cref{thm:main} on $\tilde P$ and scale the resulting coreset by $\eps w_{min}/2$ to conclude the proof.
\end{proof}

\subsection{Partitioning an Instance into Groups: Definitions}\label{sec:def-groups}
As sketched, the algorithm partitions the input points into structured groups. We give here the useful definitions.

Fix a metric space $I = (X, \dist)$, positive integers $k, z$ and a set of clients $P$. For a solution $\calS$ of $(k,z)$-clustering on $P$ and a center $c \in \calS$, $c$'s cluster consists of all points closer to $c$ than to any other center of $\calS$. 

Fix as well some $\eps > 0$, and let $\greedy$ be any solution for $(k, z)$-clustering on $P$ with $k$ centers. 
Let $C_1, ..., C_{k}$ be the clusters induced by the centers of $\greedy$.

\begin{itemize}
\item the average cost of a cluster $C_i$ is $\Delta_{C_i} = \frac{\cost(C_i, \greedy)}{|C_i|}$
\item For all $i, j$, the \textit{ring} $R_{i,j}$ is the set of points $p \in C_i$ such that 
\[2^j \Delta_{C_i} \leq \cost(p, \greedy) \leq 2^{j+1} \Delta_{C_i}.\] 
\item The \textit{inner ring} $\inner(C_i) := \cup_{j \leq 2z\log(\eps/z)} R_{i,j}$ (resp. \textit{outer ring} $\out(C_i) := \cup_{j > 2z\log(z/\eps)}R_{i,j}$) of a cluster $C_i$ consists of the points of $C_i$ with cost at most $\left(\nicefrac \eps z\right)^{2z} \Delta_{C_i}$ (resp. at least $\left(\nicefrac z \eps \right)^{2z} \Delta_{C_i}$). The \textit{main ring} $\main(C_i)$ consists of all the other points of $C_i$. For a solution $\calS$, we let $\inner^\calS$ and $\out^\calS$ be the union of inner and outer rings of the clusters induced by $\calS$.
\item for each $j$, $R_j$ is defined to be $\cup_{i=1}^{k} R_{i,j}$.
\item For each $j$, the rings $R_{i,j}$ are gathered into \textit{groups} $G_{j,b}$ defined as follows:
\begin{eqnarray*}
G_{j, b} := \left\lbrace p \mid \exists i,~p\in R_{i,j} \text{ and } \left(\frac{\varepsilon}{4z}\right)^{z}\cdot    \frac{\cost(R_j,\greedy)}{k} \cdot 2^b  \leq \cost(R_{i,j},\greedy) \leq \left(\frac{\varepsilon}{4z}\right)^{z}  \cdot 2^{b+1}\cdot \frac{\cost(R_j,\greedy)}{k}\right\rbrace.
\end{eqnarray*}
\item  For any $j$,  let $G_{j, min} := \cup_{b \leq 0} G_{j, b}$ be the union of the cheapest groups, and $G_{j, max} := \cup_{b \geq z\log{\frac{4z}{\eps}}} G_{j, b}$ be the union of the most expensive ones. The set of interesting groups is made of $G_{j, min}, G_{j, max}$, and $G_{j,b}$ for all $0 < b < z\log{\frac{4z}{\eps}}$.
\item The set of outer rings is also partitioned into \emph{outer groups}:
\begin{align*}
\outergroup{}_{b} = \Big\lbrace p \mid \exists i,~ p\in C_i \text{ and }
\left(\frac{\varepsilon}{4z}\right)^{z}\cdot    &\frac{\cost(\out^\greedy,\greedy)}{k} \cdot 2^b  \leq \cost(\out(C_i),\greedy) \\
&\leq \left(\frac{\varepsilon}{4z}\right)^{z}  \cdot 2^{b+1}\cdot \frac{\cost(\out^\greedy,\greedy)}{k}\Big\rbrace.
\end{align*}
\item We let as well $\outergroup{}_{min} = \cup_{b \leq 0} \outergroup{}_b$ and 
$\outergroup{}_{max} = \cup_{b \geq z\log{\frac{4z}{\eps}}} \outergroup{}_b$.  The interesting outer groups are $\outergroup{}_{min}, \outergroup{}_{max}$ and all $\outergroup{}_{b}$ with $0 < b < z\log{\frac{4z}{\eps}}$.
\end{itemize}

Intuitively, grouping points by groups is helpful, as all points in the same ring can pay the same additive error. Since there are very few groups, it turns out possible to construct a coreset for each group, and then take the union of the group's coreset. This is essentially the algorithm we propose.

We note few facts about the partitioning:
\begin{fact}
\label{fact:krings}
There exist at most $O(z\log(z/\varepsilon))$ many non-empty $R_j$ that are not in some inner or outer ring, i.e., not in $\inner^\greedy$ nor in $\out^\greedy$.
\end{fact}

Hence, the number of different non-empty groups is bounded as well:
\begin{fact}
\label{fact:kgroups}
There exists at most $O(z^2\log^2(z/\varepsilon))$ many interesting $G_{j,b}$.
\end{fact}
This is simply due to the fact that $j$ can take only interesting values between $2z \log(\eps/z)$ and $2z \log(z/\eps)$, and interesting $b$ between $0$ and $z \log (4z/\eps)$.

By the definition of the outer groups, we have also that
\begin{fact}
\label{fact:outergroups}
There exists at most $O(z\log (z/\eps))$ many interesting outer groups.
\end{fact}

For simplicity, we will drop mention of "interesting" : when considering any group, it will implicitly be an interesting group.

\section{The Coreset Construction Algorithm, and Proof of Theorem~\ref{thm:main}}\label{sec:algo}

\subsection{The algorithm}
For an initial metric space $(X, \dist)$, set of clients $P$ and $\eps > 0$, our algorithm essentially consists of the following steps: given a solution $\greedy$, it processes the input in order to reduce the number of different
groups. Then, the algorithm computes a coreset of the points inside each group using the following \texttt{GroupSample} procedure. The final coreset is made of the union of the coresets for all groups.

The \texttt{GroupSample} procedure takes as input a group of points $G$ as defined in \cref{sec:def-groups}, a set of centers $\greedy$ inducing clusters $\tilde C_1, \tilde C_2, ..., \tilde C_k$ on $G$ and an integer $\delta$. Note importantly that the definition of clusters $\tilde C_i$ says that they are only made of points from the group $G$. The output of \texttt{GroupSample} is a set of weighted points, computed as follows: 
a point $p \in \tilde C_i$ is sampled with probability $\frac{\delta\cdot \cost(\tilde C_i,\greedy)}{|\tilde C_i|\cdot \cost(G,\greedy)}$, and the weight of any sampled point is rescaled by a factor $\frac{|\tilde C_i|\cdot \cost(G,\greedy)}{\delta \cost(\tilde C_i,\greedy)}$.\footnote{Note that this is essentially importance sampling, as each point in a cluster $\tilde C_i$ have cost roughly equal to the average. We chose this different distribution for simplicity in our proofs.} 

The properties of the \texttt{GroupSample} procedure are captured by the following lemma.
\begin{restatable}[]{lemma}{coresetreduc}
\label{lem:coreset-reduc}
Let $(X, \dist)$ be a metric space, $k, z$ be two positive integers and $G$ be a group of clients and $\greedy$ be a solution to $(k, z)$-clustering on $G$ with $k$ centers such that:
\begin{itemize} 
\item for every cluster $\tilde C$ induced by $\greedy$ on $G$, all points of $\tilde C$ have the same cost in $\greedy$, up to a factor $2$: $\forall p, q \in \tilde C,~ \cost(p, \greedy) \leq 2\cost(q, \greedy)$.
\item for all clusters $\tilde C$ induced by $\greedy$ on $G$, it holds that $\frac{\cost(G, \greedy)}{2k} \leq \cost(\tilde C, \greedy)$.
\end{itemize}

Let $\cand$ be a $\greedy$-approximate centroid set for $(k,z)$-clustering on $G$.

Then, there exists an algorithm \texttt{GroupSample}, running in time $O(|G|)$ that constructs a set $\coreset$ of size $\delta$ such that,
with probability $1-\exp\left(k\log |\cand| - 2^{O(z\log z)}\cdot \frac{\min(\eps^2,\eps^{z})}{\log^2 1/\eps}\cdot \delta\right)$ it holds that for all
set $\calS$ of $k$ centers:
\[|\cost(G, \calS) - \cost(\coreset,\calS)| = O(\eps)\left(\cost(G, \calS) + \cost(G, \greedy)\right).\qedhere\]
\end{restatable}

We further require the \texttt{SensitivitySample} procedure, which we will apply to some of the points not consider by the calls to \texttt{GroupSample}. From a group $G$, this procedure merely picks $\delta$ points $p$ with probability $\frac{\cost(p,\A)}{\cost(G, \A)}$. Each of the $\delta$ sampled points has a weight $\frac{\cost(G, \A)}{\delta\cdot \cost(p,\A)}$.

The key property of \texttt{SensitivitySample} is given in the following lemma.

\begin{restatable}[]{lemma}{coresetouter}
\mbox{}\label{lem:coreset-outer}
Let $(X, \dist)$ be a metric space, $k, z$ be two positive integers, $P$ be a set of clients
and $\greedy$ be a $c_{\A}$-approximate solution solution to $(k, z)$-clustering on $P$.

Let $G$ be either a group $\outergroup{}_{b}$ or $\outergroup{}_{\max}$. 
Suppose moreover that there is a $\greedy$-approximate centroid set $\cand$ for $(k,z)$-clustering on $G$ .

Then, there exists an algorithm \texttt{SensitivitySample} running in time $O(|G|)$ that constructs a set $\coreset$ of size $\delta$ such that it holds with probability 
$1-\exp\left(k\log |\cand| - 2^{O(z\log z)}\cdot \frac{\eps^2}{\log^2 1/\eps}\cdot \delta\right)$ that, for all sets $\calS$ of $k$ centers:
  $$|\cost(G, \calS) - \cost(\coreset,\calS)| =\frac{\eps}{z\log z/\varepsilon}\cdot \left(\cost(\calS) + \cost(\greedy)\right).\qedhere$$
\end{restatable}

An interesting feature of \cref{lem:coreset-outer} is that the probability does not depend on $\eps^{-z}$, as it does in \cref{lem:coreset-reduc}.

Using the two algorithms \texttt{GroupSample} and  \texttt{SensitivitySample}, we can formally present the whole algorithm:

\textbf{Input:} A metric space $(X, \dist)$, a set $P \subseteq X$, $k, z > 0$, a solution $\greedy$ to $(k, z)$-clustering on $P$,
and $\eps$ such that $0 < \eps < 1/3$.\\
\textbf{Output:} A coreset. Namely, a set of points $\coreset \subseteq P \cup \greedy$ and a weight function $w: \coreset \mapsto \R_+$
such that for any set of $k$ centers $\calS$, $\cost(P, \calS) = (1\pm\eps) \cost(\coreset, \calS)$. 
\begin{enumerate} 
\item Set the weights of all the centers of $\greedy$ to 0.
\item\label{step:ig} \textbf{Partition the remaining instance into groups}:
\begin{enumerate}
\item For each cluster $C$ of $\greedy$ with center $c$, remove $\inner(C)$ and increase the weight of $c$ by $|\inner(C)|$.
\item For each cluster $C$ with center $c$ in solution $\greedy$ , the algorithm discards also
  all of $C \cap \cup_{j} G_{j,min}$ and $\out(C) \cap \outergroup{}_{min}$,
  and increases the weight of $c$ by
  the number of points discarded in cluster $c$. 
\item Let $\discarded$ be the set of points discarded at those steps, and $\structured$ be the weighted set of centers that have positive weights.
 \end{enumerate}
\item \textbf{Sampling from well structured groups:} For every $j$ such that $z \log(\eps/z) \leq j \leq 2z \log(z/\eps)$ and every group $G_{j,b} \notin G_{j, min}$, %
  compute a coreset $\coreset_{j,b}$ of size 
  $$\delta = O\left(\valuedelta\right)$$ using the \texttt{GroupSample} procedure. 
  
\item \textbf{Sampling from the outer rings:} From each group $\outergroup{}_1, ..., \outergroup{}_{max}$,
  compute a coreset $\coreset^O_b$ of size
  $$\delta = O\left(\frac{2^{O(z \log z)} \cdot \log^2(1/\eps)}{\eps^2}(k \log |\cand| + \log \log(1/\eps) + \log(1/\pi))\right)$$ using the \texttt{SensitivitySample} procedure. 
  
\item \textbf{Output:}
  \begin{itemize}
  \item A coreset consisting of $\greedy \cup_{j,b} \coreset_{j,b} \cup_i \coreset^O_i$.
  \item Weights: weights for $\greedy$ defined throughout the algorithm,
    weights for $\coreset_{j,b}$ defined by the \texttt{GroupSample} procedure,
    weights for $\coreset_O$ defined by the\\ \texttt{SensitivitySample} procedure.
  \end{itemize}
\end{enumerate}

\begin{remark}
Instead of using the \texttt{GroupSample} procedure, one could use any coreset construction tailored for the well structured group. Improving on that step would improve the final coreset bound: if the size of the coreset produced for a group is $T$, then the total coreset has size 
\[\widetilde O\left(T + \frac{2^{O(z \log z)}}{\eps^2} \cdot k \log |\cand|\right)\qedhere\]
\end{remark}

\subsection{Proof of \cref{thm:main}}
As we prove in \cref{sec:preprocess}, the outcome of the partitioning step, $\discarded$ and $\structured$, satisfies the following lemma, that deals  with the inner ring, and the groups $G_{j, min}$ and $\outergroup{}_{min}$:
\begin{restatable}[]{lemma}{preprocess}
\label{lem:preprocess}
Let $(X, \dist)$ be a metric space with a set of clients $P$, $k, z$ be two positive integers, and $\eps \in \mathbb{R}^*_+$. 
For every solution $\calS$, it holds that 
$$|\cost(\discarded, \calS) - \cost(\structured, \calS)| = O(\eps) \cost(\calS),$$
where $\discarded$ and $\structured$ are defined in Step \ref{step:ig} of the algorithm.
\end{restatable}

Combining properties of the partitioning, \cref{lem:coreset-reduc}, \cref{lem:coreset-outer} and \cref{lem:preprocess} allows to prove \cref{thm:main}:


%
%
\begin{proof}[Proof of \cref{thm:main}]
Let $\coreset$ be the output of the algorithm described above, and $\delta = O\left(\valuedelta\right)$ as defined in step 3 of the algorithm.
Due to \cref{fact:kgroups} and \cref{fact:outergroups},
$\coreset$ has size $O(z^2 \log^2(z/\eps) \cdot \delta + |\greedy|)$, and non-negative weights by construction. {\setlength{\emergencystretch}{4.5em}\par}


We now turn to analysing the quality of the coreset.
Any group $G_{j,b}$ for $b > 0$ satisfies \cref{lem:coreset-reduc}:
the cost of any point $p \in G_{j, b} \cap C_i$ satisfies $2^j \Delta_{C_i} \leq \cost(p, \greedy) \leq 2^{j+1} \Delta_{C_i}$, and
\begin{itemize}
\item for $b \in \left\{0, ..., z\log \frac{4z}{\eps}\right\}$, the cost of all clusters induced by $\greedy$ on $G_{j, b}$ are equal up to a factor 2, hence for all $i$ $\frac{\cost(G_{j,b}, \greedy)}{2k} \leq \cost(C_i \cap G_{j,b}, \greedy)$
\item for $b=max$, it holds that $\frac{\cost(G_{j,max}, \greedy)}{2k} \leq \frac{\cost(R_j, \greedy)}{2k} \leq \cost(C_i \cap G_{j,max}, \greedy)$.
\end{itemize}


Hence, \cref{lem:coreset-reduc} ensures that, with probability $1-\exp\left(k\log |\cand| - 2^{O(z\log z)}\cdot \frac{\min(\eps^2,\eps^{z})}{\log^2 1/\eps}\cdot \delta\right)$, the coreset $\coreset_{j,b}$ constructed for $G_{j,b}$ satisfies for any solution $\calS$
$$|\cost(G_{j,b}, \calS) - \cost(\coreset_{j,b},\calS)| = O(\eps)\left(\cost(G_{j,b}, \calS) + \cost(G_{j,b}, \greedy) \right).$$ 

Similarly, \cref{lem:coreset-outer} ensures that, 
with probability $1-\exp\left(\log |\cand| - 2^{O(z\log z)} \cdot \frac{\eps^2}{\log^2 1/\eps}\cdot \delta\right)$, 
the coreset $\coreset^O_{b}$ constructed for $\outergroup{}_b$ satisfies for any solution $\calS$ 
$$|\cost(\outergroup{}_b, \calS) - \cost(\coreset^O_{b},\calS)| = \frac{\eps}{z\log(z/\eps)}\left(\cost(\calS) + \cost(\greedy) \right).$$

Taking a union-bound over the failure probability of  \cref{lem:coreset-outer} and of \cref{lem:coreset-reduc} applied to
all groups $G_{j, b}$ with $z\log(\eps/z) \leq j \leq 2z\log(z/\eps)$ and all $\outergroup{}_i$ implies that, 
with probability 
\begin{align*}
1-z^2\log^2(z/\eps)\exp\left(k\log |\cand| - 2^{O(z\log z)}\cdot \frac{\min(\eps^2,\eps^{z})}{\log^2 1/\eps}\cdot \delta\right) \\
- z\log(z/\eps)\exp\left(\log |\cand| - 2^{O(z\log z)} \frac{\eps^2}{\log^2 1/\eps}\cdot \delta\right)
\end{align*}
for any solution $\calS$,
\begin{align*}
&|\cost(\calS) - \cost(\coreset,~\calS)| \\
\leq ~&
|\cost(\discarded, \calS) - \cost(\structured, \calS)|  + \sum_{j
, b} |\cost(G_{j,b}, \calS) - \cost(G_{j,b}\cap \coreset, \calS)|\\
& \qquad \qquad+\sum_{i} |\cost(\outergroup{}_{b}, \calS) - \cost(\outergroup{}_b \cap \coreset, \calS)|\\
\leq ~ &O(\eps) \cost(\calS) + O(\eps)\cost(\greedy) \leq O(\eps)\cost(\calS)
\end{align*}
where the penultimate inequality uses \cref{lem:preprocess}, and the last one that $\greedy$ is a constant-factor approximation.

For $\delta = \valuedelta$, this probability can be simplified to
\begin{align*}
&1-\exp\Big(  2(\log z + \log\log(z/\eps)) + k\log |\cand| -2^{O(z\log z)}\cdot \frac{\min(\eps^2,\eps^{z})}{\log^2 1/\eps}\cdot \delta\Big) = 1 - \pi.
\end{align*}

The complexity of this algorithm is:
\begin{itemize}
\item $O(n)$ to compute the groups: given all distances from a client to its center, computing the average cost of all clusters costs $O(n)$, hence partitioning into $R_j$ cost $O(n)$ as well, and then decomposing $R_j$ into groups is also done in $O(n)$ time;
\item plus the cost to compute the coreset in the groups, which is $\sum_{j, b} O(|G_{j, b}|) + \sum_{i} O(|\outergroup{}_b|) = O(n)$
\end{itemize}
Hence, the total complexity is $O(n)$.
\end{proof}

\section{Sampling inside Groups: Proof of Lemma \ref{lem:coreset-reduc}}\label{sec:sample}

The goal of this section is to prove \cref{lem:coreset-reduc}:
\coresetreduc*

\subsection{Description of the GroupSample Algorithm}

The \texttt{GroupSample} merely consists of importance sampling in rounds, i.e. there are $\delta$ rounds in which one point of $G$ is sampled. Let $\tilde C_1, \tilde C_2, ...$ be the clusters induced by $\greedy$ on $G$:
the probability of sampling point $p\in \tilde C_i$ 
is $\frac{ \cost(\tilde C_i,\greedy)}{|\tilde C_i|\cdot \cost(G,\greedy)}$ -- recall that all clusters $\tilde C_i$ contain only points from the group $G$.
The weight of any sampled point is rescaled by a factor $\frac{|\tilde C_i|\cdot \cost(G,\greedy)}{\delta \cost(\tilde C_i,\greedy)}$.
If there are $m$ copies of a point, it is sampled in a round with probability $\frac{m \cdot \cost(\tilde C_i, \greedy)}{|\tilde C_i| \cdot \cost(G, \greedy)}$
(which is equivalent to sampling each copy with probability $\frac{\cost(\tilde C_i, \greedy)}{|\tilde C_i| \cdot \cost(G, \greedy)}$).
In what follows, each copies will be considered independently.

\begin{definition}
We denote $\weight(p) := \frac{|\tilde C_i|\cdot \cost(G,\greedy)}{\delta \cost(\tilde C_i,\greedy)}$ the scaling factor of the weight of a point $p \in \tilde C_i$.
\end{definition}

\subsection{Organization of the Proof}
To analyze the sampling procedure of \texttt{GroupSample}, we consider different cost ranges $I_{\ell,\calS}$ induced by a solution $\calS$ as follows. 
A point $p$ of $G$ is in $I_{\ell,\calS}$ if $2^{\ell}\cdot \cost(p,\greedy)\leq \cost(p, \calS) \leq 2^{\ell+1}\cdot \cost(p,\greedy)$.
We distinguish between the following cases.
\begin{itemize}
\item $\ell\leq \log \varepsilon/2$. We call all $I_{\ell,\calS}$ in this range \emph{tiny}. The union of all tiny $I_{\ell,\calS}$ is denoted by $I_{tiny, \calS}$.
\item $\log \varepsilon/2 \leq \ell \leq z\log(8z/\varepsilon)$. We call all $I_{\ell,\calS}$ in this range \emph{interesting}.
\item $\ell \geq z\log(4z/\varepsilon)$. We call all $I_{\ell,\calS}$ in this range \emph{huge}.
\end{itemize}

Note that interesting and huge ranges intersect. This is to give us some slack in the proof: for a solution $\calS$, we will deal with huge ranges before relating $\calS$ to its representative $\tilde \calS$ from $\cand^k$. Due to the approximation, some non-huge range for $\calS$ can become huge for $\tilde \calS$:  however, due to our definition, they stay in the interesting ranges.

A simple observation leads to the next fact.

\begin{fact}
\label{fact:interesting}
Given a solution $\calS$ , there are at most $O(z\log z/\varepsilon)$ interesting $I_{\ell,\calS}$.
\end{fact}

Bounding the difference in cost of $G\cap I_{\ell,\calS}$ requires different arguments depending on the type of $I_{\ell,\calS}$. 
The two easy cases are tiny and huge, so we will first proceed to prove those. Proving the interesting case is arguably both the main challenge and our main technical contribution.

For the proof, we will rely on Bernstein's concentration inequality:
\begin{theorem}[Bernstein's Inequality]
\label{thm:bernstein}
Let $X_1,\ldots X_{\delta}$ be non-negative independent random variables. Let $S=\sum_{i=1}^{\delta} X_i$.
If there exists an almost-sure upper bound $M\geq X_i$, then 
\begin{align*}
\mathbb{P}\left[\left\vert S - \mathbb{E}[S] \right\vert \geq t\right] \leq \exp\left(- \frac{t^2}{2\sum_{i=1}^{\delta} \left( \mathbb{E}[X_i^2] - \sum\mathbb{E}[X_i]^2\right) + \frac{2}{3}\cdot M\cdot t}\right).
\end{align*}
\end{theorem}
In this paper we will simply drop the $E[X_i]^2$ terms from the denominator, as the second moment will dominate in all important cases. 

In what follows, we fix $k$, $z$, $G$ and $\greedy$, as in the assumptions of \cref{lem:coreset-reduc}. Let $\tilde C_1, ..., \tilde C_{k}$ be the clusters induced by $\greedy$ on $G$. The assumptions imply the following fact:
\begin{fact}\label{fact:cost_rings}
For any $p \in \tilde C_i$, $\frac{\cost(\tilde C_i, \greedy)}{2|\tilde C_i|} \leq \cost(p, \greedy) \leq \frac{2\cost(\tilde C_i, \greedy)}{|\tilde C_i|}$.
\end{fact}
We will start with the tiny type, as it is mostly divorced from the others.

\begin{figure}
\centering
\includegraphics[scale=0.6]{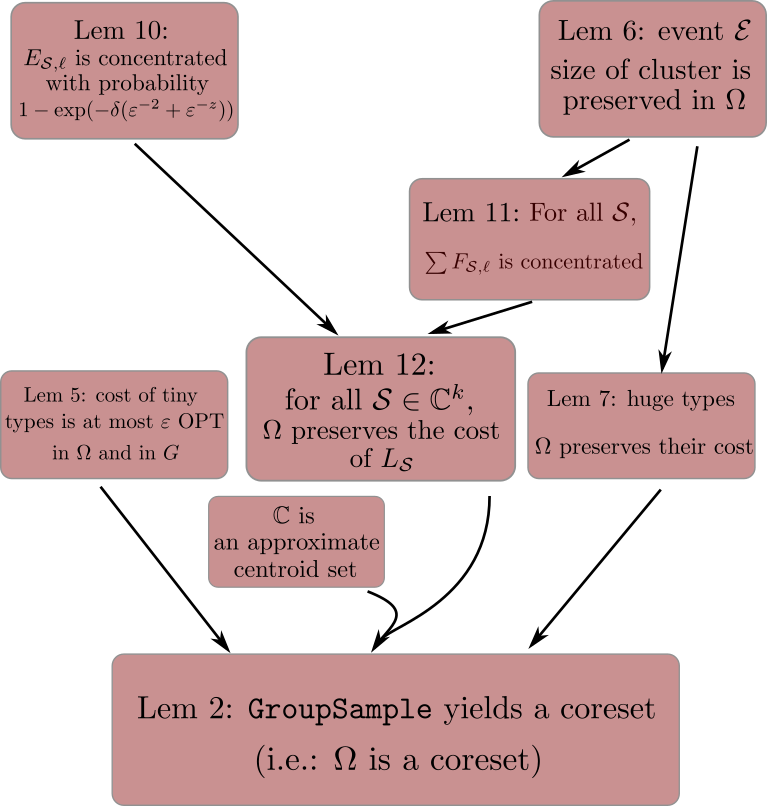}
\caption{Arrangement of Lemmas of \cref{sec:sample} to prove \cref{lem:coreset-reduc}. }\label{fig:dag-proof}
\end{figure}

\subsection{Dealing with Tiny Type}
\begin{lemma}
\label{lem:kepstiny}
It holds that, for any solution $\calS$,
    \[\max\left(\sum_{p \in I_{tiny,\calS}} \cost(p, \calS), ~
       \sum_{p \in  I_{tiny,\calS} \cap \coreset} 
             \weight(p) \cost(p,\calS)\right)     \le \varepsilon \cdot \cost(G,\alg).\qedhere\]   
\end{lemma}
\begin{proof}
 By definition of $I_{tiny,\calS}$, $\sum_{p \in I_{tiny,\calS}} \cost(p, \calS) \leq \sum\limits_{p \in I_{tiny,\calS}} \frac{\varepsilon}{2}\cdot  \cost(p, \greedy) \leq \frac{\varepsilon}{2} \cdot \cost(G, \greedy)$.
 Similarly, we have for the other term 
  \begin{align*}
   \sum_{p \in  I_{tiny,\calS} \cap \coreset} \weight(p) \cdot \cost(p,\calS) 
   & \leq  \sum_{p \in I_{tiny,\calS} \cap \coreset}  \weight(p) \frac{\varepsilon}{2}\cdot  \cost(p, \greedy) \\
   & \leq  \frac{\varepsilon}{2} \sum_{i=1}^{k} \sum_{p \in  \tilde C_i \cap I_{tiny,\calS} \cap \coreset}  \frac{|\tilde C_i| \cdot \cost(G,\greedy)}{\delta \cost(\tilde C_i,\greedy)} \cdot  \frac{2\cdot \cost(\tilde C_i,\greedy)}{|\tilde C_i|} \\
  & \leq  \varepsilon\cdot  \frac{|I_{tiny, \calS} \cap \coreset|}{\delta}\cost(G,\greedy)  \\  & \leq  \varepsilon\cdot  \cost(G,\greedy)  .
  \end{align*}
where the last inequality uses that $\coreset$ contains $\delta$ points.
\end{proof}

\subsection{Preserving the Weight of Clusters, and the Huge Type}
We now consider the huge ranges. For this, we first show that, given we sampled enough points, $|\tilde C_i|$ is well approximated for every cluster $\tilde C_i$. This lemma will also be used later for the interesting points.
We define event $\mathcal{E}$ to be:
For all cluster $\tilde C_i$ induced by $\greedy$ on $G$, 
\begin{equation*}
\sum_{p\in \tilde C_i \cap \coreset} \frac{|\tilde C_i|\cdot \cost(G,\greedy)}{\cost(\tilde C_i,\greedy)\cdot \delta} = (1\pm \varepsilon) \cdot |\tilde C_i|
\end{equation*}

\begin{lemma}
\label{lem:ksize}
We have that with probability at least $1 - k\cdot z^2\log^2(z/\varepsilon) \exp\left(-O(1)\frac{\varepsilon^2}{k}\delta\right)$, event $\calE$ happens.
\end{lemma}  
\begin{proof}
Consider any cluster $\tilde C_i$ induced by $\greedy$ on $G$. The expected number of points sampled from $\tilde C_i$ is then at least
\begin{align*}
\mu_{i}:= \sum_{p\in \tilde C_i} \frac{\delta\cost(\tilde C_i,\greedy)}{|\tilde C_i|\cdot \cost(G,\greedy)} &= \frac{\delta\cost(\tilde C_i,\greedy)}{\cost(G,\greedy)} \geq \frac{\delta}{2k},
\end{align*}
where the inequality holds by assumption on $G$.
Define the indicator variable of point $p$ from the sample being drawn from $\tilde C_i$ as $\mathcal{P}_i(p)$. Using Chernoff bounds, we therefore have
\begin{eqnarray}
\mathbb{P}\left[\left\vert \sum_{p\in G\cap \coreset} P_{i}(p) - \mu_{i} \right\vert \geq \varepsilon\cdot \mu_{i} \right] 
\label{eq:lemhuge1}
&\leq & \exp\left(-\frac{\varepsilon^2\cdot \mu_{i}}{3}\right) \leq \exp\left(-\frac{\eps^2\delta}{6k}\right).
\end{eqnarray}
Now, rescaling $P_{i}(p)$ by a factor $\frac{|\tilde C_i|\cdot \cost(G,\greedy)}{\delta\cost(\tilde C_i,\greedy)}$ implies that approximating $\mu_{i}$ up to a $(1\pm \varepsilon)$ factor also approximates $|\tilde C_i|$ up to a $(1\pm \varepsilon)$ factor.

The final result follows by applying a union bound for all clusters in all groups.
\end{proof}

We now show that for any cluster $\tilde C_i$ with a non-empty huge range, Lemma~\ref{lem:ksize} implies that the cost is well approximated -- without the need of going through the approximate solution $\tilde \calS	$.

\begin{lemma}
\label{lem:khuge}
Condition on event $\calE$. Then,
for any solution $\calS$, and any $i$
such that there exists $\ell \geq z\log(4z/\eps)$ and a point $p \in \tilde C_i$ with $\cost(p, \calS) \geq 2^\ell \cost(p, \greedy)$, we have:
\[\left\vert\cost(\tilde C_i,\calS) - \sum_{p\in \Omega \cap C_i} \weight(p) \cdot \cost(p,\calS) \right\vert \leq  7\eps\cdot \cost(\tilde C_i ,\calS).\qedhere\]
\end{lemma}
\begin{proof}
Let $p \in \tilde C_i$ as given in the statement. Using the structure of clusters in a group, this implies for any $q \in  \tilde C_i$: 
$\cost(p,q) \leq \left(\dist(p, \greedy) + \dist(q, \greedy)\right)^z \leq 3^z\cdot \cost(p,\greedy)\leq 3^z \cdot 2^{(\ell - z\log (4z/\eps))} \cost(p,\greedy) \leq (3\eps/4z)^{z}\cdot \cost(p,\calS)$. 
By Lemma~\ref{lem:weaktri}, we therefore have for any point $q\in \tilde C_i$
\begin{eqnarray*}
\cost(p,\calS) &\leq & \left(1+\eps/2z \right)^{z-1} \cost(q,\calS) + \left(1+2z/\eps\right)^{z-1}\cost(p,q) \\
&\leq & \left(1+\eps \right) \cost(q,\calS) + \varepsilon\cdot \cost(p,\calS) \\
\Rightarrow \cost(q,\calS) &\geq & \frac{1-\varepsilon}{1+\eps}\cost(p,S) \geq  (1-2\eps) \cost(p, \calS)
\end{eqnarray*}

By a similar calculation, we can also derive an upper bound of $\cost(q,\calS)\leq \cost(p,\calS)\cdot (1+2\eps)$. Hence, we have
\begin{eqnarray*}
\sum_{q\in \Omega \cap \tilde C_i} \frac{|\tilde C_i|\cdot \cost(G,\greedy)}{\cost(\tilde C_i,\greedy)\cdot \delta} \cdot \cost(q,\calS) 
&=& (1\pm 2\varepsilon)\cdot \cost(p,\calS)\cdot \sum_{q\in \Omega \cap \tilde C_i} \frac{|\tilde C_i|\cdot \cost(G,\greedy)}{\cost(\tilde C_i,\greedy)\cdot \delta} \\
&=& (1\pm 2\varepsilon)\cdot \cost(p,\calS)\cdot (1\pm \varepsilon)\cdot |\tilde C_i| ~~~(\text{Event }\calE) \\
&=& (1\pm 2\varepsilon)\cdot  (1\pm \varepsilon)\cdot (1\pm 2\varepsilon)\cdot \cost(\tilde C_i,\calS) \\
&=& (1\pm 7\varepsilon) \cdot \cost(\tilde C_i,\calS).
\end{eqnarray*}
\end{proof}

\subsection{Bounding Interesting $I_{\ell, \calS}$: a Simple but Suboptimal Analysis.}
\label{sec:interesting-subopt}
Now we move onto the most involved case, presenting first a suboptimal analysis of \texttt{GroupSample} for the interesting types.
As explained in the introduction, our main goal is to design a good estimator and apply Bernstein's inequality to it. 

Since the clusters intersecting a huge $I_{\ell, \calS}$ are dealt with by \cref{lem:khuge}, we only need to focus on the \textit{interesting clusters}, namely clusters $\tilde C$ that satisfy
\begin{equation}
\label{eq:interesting}
\nexists p\in \tilde C ~|~ \cost(p,\calS)\geq \left(\frac{8z}{\eps}\right)^{z}\cdot \cost(p,\greedy).
\end{equation}
In other words, a clustering is interesting only if it does not have any point in a huge $I_{\ell, \calS}$. This restriction will be crucial to our analysis. 
Let $L_\calS$ be a set of interesting clusters (possibly not all of them).\footnote{We define $L_\calS$ to contain only huge clusters but not all of them in order to relate the cost of solutions from the approximate centroid set $\cand$ to the cost of any solution, as it will become clear in \cref{sec:combining}.}  For simplicity, we will assimilate $L_\calS$ and the points contained in the clusters of $L_\calS$.

We present here a first attempt to show that the cost of interesting points is preserved. Although suboptimal, it serves as a good warm-up for our improved bound.

In this first attempt, we will use the simple estimator $E(L_\calS) := \sum_{p\in L_\calS \cap \Omega}\weight(p) \cost(p,\calS)$ 
  as an estimator of the cost for points in $L_\calS$. Note that by choice of the weights $\weight(p)$, this estimator is unbiased: 
$\E[E(L_\calS)] = \sum_{p \in L_\calS} \cost(p, \calS)$, precisely the quantity we seek to estimate.

To show concentration, we rely on Bernstein's inequality from \cref{thm:bernstein}. Hence, the key part of our proof is to bound the variance of the estimator.

\begin{lemma}
\mbox{}\label{lem:interesting-easy}
Let $G$ be a group of points, and $\greedy$ be a solution. 
Let $\cand$ be an $\greedy$-approximate centroid set, as in \cref{def:centroid-set}.
It holds with probability 
$$ 1-\exp\left(k\log |\cand| - \frac{\eps^{2+z}}{2^{O(z\log z)}\log^2 1/\eps}\cdot \delta\right)$$
 that, for all solution $\tilde \calS \in \cand^k$ and any set of interesting clusters $L_{\tilde \calS}$ induced by $\greedy$ on $G$:
%
\[\left|E(L_\calS) - \E\left[E(L_\calS)\right]\right| \leq \frac{\eps}{z \log z/\eps}\cdot\cost(G, \greedy)\]
\end{lemma}
\begin{proof}
First, we fix some solution $\calS$ and some set of interesting clusters $L_\calS$, verifying \cref{eq:interesting}.
We express $E(L_\calS)$ as a sum of i.i.d variables : $E(L_\calS) = \sum_{j=1}^\delta X_j$, where $X_j = f(\coreset_j) \cost(\coreset_j, \calS)$ when the $j$-th sampled point is $\coreset_j \in L_\calS$, $X_j = 0$ otherwise.
Recall that,  due to \cref{fact:cost_rings}, the probability that the $j$-th sampled point is $p$ from some cluster $\tilde C$ satisfies $\mathbb{P}[\coreset_j = p] = \frac{\cost(\tilde C,\greedy)}{|\tilde C|\cdot \cost(G,\greedy)} \leq \frac{2\cost(p,\greedy)}{\cost(G,\greedy)}$. From the same fact, $\weight(p) \leq \frac{2\cost(G,\greedy)}{\delta\cost(p,\greedy)}$.

We will rely on Bernstein's inequality (\cref{thm:bernstein}). To do this, we need an upper bound on the variance of $E(L_\calS)$, as well as an almost sure upper bound $M$ on every sample. 
We first bound $\E[X_i^2]$:
\begin{eqnarray*}
\nonumber
 \E[X_i^2] &=&  
 	\E\left[ \left(\weight(\coreset_i)\cost(\coreset_i,\calS) \right)^2\right] \\
&=& \sum_{p \in  L_\calS} 
		\left(\weight(p) \cost(p,\calS) \right)^2 \Pr\left[\coreset_i = p\right] \\ 
&\leq & \sum_{p \in L_\calS} 
		\cost(p, \calS)\cdot \left(\frac{4z}{\eps}\right)^z \cdot \cost(p,\greedy)\cdot 
			\left(\frac{2\cost(G,\greedy)}{\delta \cost(p,\greedy)}\right)^2 \frac{2\cost(p,\greedy)}{\cost(G,\greedy)}\\
&\leq &	\left(\frac{4z}{\eps}\right)^{z} \cdot 
			\frac{\cost(G,\greedy)}{\delta^2} \sum_{p \in L_\calS} \cost(p, \calS)\\
&\leq &	\left(\frac{4z}{\eps}\right)^{z} \cdot 
			\frac{\cost(G,\greedy) \cost(G, \calS)}{\delta^2} \\
&\leq &	\left(\frac{4z}{\eps}\right)^{z} \cdot 
			\frac{(\cost(G,\greedy) + \cost(G, \calS))^2}{\delta^2} 
\end{eqnarray*}
Where, in the third line, we upper bounded only one of the $\cost(p, \calS)$ by $(4z/\eps)^z \cost(p, \greedy)$. Hence, it holds that $\sum_{i=1}^\delta \E[X_i^2] \leq \left(\frac{4z}{\eps}\right)^{z} \cdot \frac{(\cost(G,\greedy) + \cost(G, \calS))^2}{\delta}$. 

To apply Bernstein's inequality, we also need an upper-bound on the value of $X_i$: using $\cost(p, \calS) \leq \left(\frac{4z}{\eps}\right)^z \cost(p, \greedy)$ and $\weight(p) \leq \frac{2\cost(G, \greedy)}{\delta \cost(p, \greedy)}$ we get
\begin{eqnarray}
\notag
X_i ~~\leq ~~M & := & 2^{O(z\log z)}\cdot \eps^{-z}\frac{(\cost(G,\greedy) + \cost(G, \calS))}{\delta}
\end{eqnarray}

Applying Bernstein's inequality with those bounds on the variance and the value of the $X_i$, we then have:
\begin{eqnarray}
\nonumber
& &\mathbb{P}\left[\vert E(L_\calS) - \mathbb{E}[E(L_\calS)]\vert >  
		\frac{\eps}{z\log z/\eps}\cdot \left(\cost(G,\greedy) + \cost(G, \calS)\right)\right] \\
\nonumber
&\leq & \exp\left(- \frac{\frac{\eps^2}{z^2\log^2 z/\eps}
					\cdot \left(\cost(G,\greedy) + \cost(G, \calS)\right)^2}
	{2\sum_{i=1}^\delta \text{Var}[X_i] + \frac{1}{3}M\cdot 
	  \frac{\eps}{z\log z/\eps}\cdot(\cost(G,\greedy) + \cost(G, \calS))}\right)\\
\nonumber
&\leq & \exp\left(- \frac{\eps^{2+z}}{2^{O(z\log z)}\log^2 1/\eps}\cdot \delta\right)
\end{eqnarray}

Hence, for a fixed solution $\calS$ and a fixed set of interesting clusters $L_\calS$, it holds with probability $1- \exp\left(- \frac{\eps^{2+z}}{2^{O(z\log z)}\log^2 1/\eps}\cdot \delta\right)$ that $\vert E(L_\calS) - \mathbb{E}[E(L_\calS)]\vert >  
		\frac{\eps}{z\log z/\eps}\cdot \left(\cost(G,\greedy) + \cost(G, \calS)\right)$.
		
Doing a union-bound over the $\cand^k$ many solutions $\calS$ and the $2^k$ many sets of interesting clusters concludes the lemma:
it holds with probability  $1- \exp\left(k\log \cand- \frac{\eps^{2+z}}{2^{O(z\log z)}\log^2 1/\eps}\cdot \delta\right)$ that, for any solution $\calS \in \cand^k$ and any set of interesting clusters $L_\calS$, $\vert E(L_\calS) - \mathbb{E}[E(L_\calS)]\vert >  
		\frac{\eps}{z\log z/\eps}\cdot \left(\cost(G,\greedy) + \cost(G, \calS)\right)$.
\end{proof}

In order to apply \cref{lem:interesting-easy}, note that the quantity $\left\vert E(L_\calS) - \mathbb{E}[E(L_\calS)]\right\vert$ is equal to $\left|\cost(L_\calS \cap \coreset, \calS) - \cost(L_\calS, \calS)\right|$, namely the difference between the cost in the full input and the cost in the coreset of points in $L_\calS$.{\setlength{\emergencystretch}{4.5em}\par}

This lemma is enough to conclude that the outcome of \texttt{GroupSample} is a coreset, once combined with Lemmas \ref{lem:kepstiny} and \ref{lem:khuge}. To see the end of the proof, one can jump directly to the proof of \cref{lem:coreset-reduc} (in \cref{sec:combining}) and use use \cref{lem:interesting-easy} instead of \cref{lem:union-bound}. This would give a coreset of size $\tilde O\left(k \eps^{-2-z}\right)$, instead of $\tilde O\left(k\eps^{-\max(2,z)}\right)$.

\subsection{Bounding Interesting $I_{\ell, \calS}$: Improved Analysis}
\label{sec:interesting}
The shortcoming of the previous estimator is its huge variance, with dependency in $\eps^{-z}$. We present an alternate estimator with small variance, allowing  in turn to increase the success probability of the algorithm.

As for the previous estimator, we only need to focus on some interesting clusters $L_{\calS}$, namely clusters that do not have any point in a huge $I_{\ell, \calS}$ and satisfy \cref{eq:interesting}, important enough to be recalled here: all clusters in $L_\calS$ verify
\begin{equation}
\nexists p\in \tilde C ~|~ \cost(p,\calS)\geq \left(\frac{8z}{\eps}\right)^{z}\cdot \cost(p,\greedy).
\end{equation}

\subsubsection{Designing a Good Estimator: Reducing the Variance}

Our first observation is that we can estimate the cost of points in $I_{\ell, \calS} \cap L_\calS$, for each $\ell$ independently, instead of estimating directly the cost of $L_\calS$ as in previous section. 
For them, we will use the following estimator:
  
  \begin{definition}
Let $G$ be a group of points, and $\tilde C_i$ be the clusters induced by a solution $\greedy$ on $G$.  For a given set of interesting clusters $L_\calS$, we let
\begin{equation}
\label{eq:estimator_def}
E_{\ell,\calS}(L_\calS) := \sum_{\tilde C_i \in L_\calS} \sum_{p \in \tilde C_i \cap I_{\ell,\calS}\cap \Omega} \weight(p) (\cost(p,\calS) - \cost(q_{i, \calS},\calS)),
\end{equation}
where $q_{i, \calS}=\underset{p\in \tilde C_i}{\text{ argmin }}\cost(p,\calS)$.
\end{definition}
 $E_{\ell, \calS}(L_\calS)$ can be expressed differently: 

\begin{align}
E_{\ell,\calS}(L_\calS) &= \sum_{\tilde C_i \in L_\calS} \sum_{p \in \tilde C_i \cap I_{\ell,\calS}\cap \Omega} \weight(p)(\cost(p,\calS) - \cost(q_{i, \calS},\calS)) \notag\\
&= \sum_{p \in I_{\ell,\calS}\cap L_\calS \cap \Omega} \weight(p) \cost(p,\calS)  - F_{\ell,\calS}(L_\calS), \\
\notag
\text{with } F_{\ell,\calS}(L_\calS) &:= \sum_{\tilde C_i \in L_\calS} \sum_{p \in \tilde C_i \cap I_{\ell,\calS}\cap \Omega}  \weight(p) \cost(q_{i, \calS},\calS)
\label{eq:estimator}
\end{align}
$F_{\ell, \calS}(L_\calS)$ is a
random variable whose value depends on the randomly sampled points $\Omega$ (we will discuss $F_{\ell,\calS}(L_\calS)$ in more detail later).

Note that the expectation of $E_{\ell,\calS}(L_\calS)$ is 
\begin{align*}
  \mathbb{E}\left[E_{\ell,\calS}(L_\calS)\right] &=  \sum_{p \in I_{\ell,\calS} \cap L_\calS} \frac{\delta \cost(\tilde C_i, \mathcal{G})}{|\tilde C_i|\cost(G, \mathcal{G})}\cdot \weight(p) \cost(p,S)  - \mathbb{E}[F_{\ell,\calS}(L_\calS)]\\
  &=  \sum_{p \in I_{\ell,\calS}\cap L_\calS} \frac{\delta \cost(\tilde C_i, \mathcal{G})}{|\tilde C_i|\cost(G, \mathcal{G})} \cdot \frac{|\tilde C_i|\cost(G, \mathcal{G})}{\delta \cost(\tilde C_i, \mathcal{G})} \cdot \cost(p,S)  - \mathbb{E}[F_{\ell,\calS}(L_\calS)]\\
  &=  \cost(I_{\ell,\calS} \cap L_\calS ,S) - \mathbb{E}[F_{\ell,\calS}(L_\calS)],
\end{align*}

Now instead of attempting to show directly concentration of all $\cost(I_{\ell,\calS} \cap L_\calS \cap \coreset, \calS)$, we will instead show that:
\begin{enumerate}
\item $E_{\ell,\calS}(L_\calS)$ is concentrated for all $\calS$, and
\item $\sum_\ell F_{\ell,\calS}(L_\calS)$ is concentrated around its expectation. 
\end{enumerate}

The reason for decoupling the two arguments is that $E_{\ell,\calS}(L_\calS)$ has a very small variance, for which few samples are sufficient: each term of the sum has magnitude $\cost(p, \calS) - \cost(q_{i, \calS}, \calS)$ instead of simply $\cost(p, \calS)$. This difference is crucial to our analysis.
Furthermore,  event $\calE$ from Lemma~\ref{lem:ksize} easily leads to a concentration bound on 
$F_{\calS}(L_\calS)=\sum_{\ell} F_{\ell,\calS}(L_\calS)$.

To establish the gain in variance obtained by subtracting $\cost(q_{i,\calS}, \calS)$, we have the following lemma.
\begin{lemma}
\label{lem:kbinom}
Let $G$ be a group of points, and $\calS$ be an arbitrary solution and $\tilde C_i$ be a cluster induced $\greedy$ on $G$
where all points have same cost, up to a factor 2.
Denote by $q_{i,\calS}=\underset{p\in \tilde C_i}{\text{ argmin }}\cost(p,\calS)$. Then for every interesting range with $\ell \geq \log \eps/2$ and every point $p\in \tilde C_i\cap I_{\ell, \calS}$, 
\[w_p := \frac{\cost(p,\calS) - \cost(q_{i, \calS},\calS)}{cost(q_{i, \calS},\greedy)} \in \left[0,2^{\ell(1-1/z)}\cdot 2^{O(z\log z)}\right]\qedhere\]
\end{lemma}
\begin{proof}
Let $w_p = \frac{\cost(p,\calS) - \cost(q_{i, \calS},\calS)}{\cost(q_{i, \calS},\greedy)}$. By choice of $q_{i, \calS}$, $w_p \geq 0$, so we consider the upper bound.

We first show useful inequalities, relating the different solutions.
Since $p \in I_{\ell, \calS}$, we have: 
\begin{align*}
\cost(q_{i,\calS}, \calS) &\leq \cost(p,  \calS) \leq 2^{\ell+1} \cost(p, \greedy) \\
&\leq 2^{\ell+2} \cost(q_{i, \calS}, \greedy),
\end{align*}
where the last inequality holds since $p$ and $q_{i, \calS}$ are in the same cluster and have up to a factor 2 the same cost.
 We also have that $\cost(p, q_{i, \calS}) \leq 2^{z-1}(\cost(p, \greedy) + \cost(q_{i, \calS}, \greedy)) \leq 3\cdot 2^{z-1}\cost(q_{i, \calS}, \greedy)$.

Now, using Lemma~\ref{lem:weaktri}, for any $\alpha \leq 1$, 
\begin{eqnarray*}
& &\cost(p,\calS) \leq  (1+\alpha/z)^{z-1}\cost(q_{i,\calS},\calS) + \left(1+\frac{z}{\alpha}\right)^{z-1} \cost(p,q_{i,\calS}) 
\end{eqnarray*}
which after rearranging implies
\begin{eqnarray*}
\cost(p,\calS) - \cost(q_{i,\calS},\calS) 
&\leq & 2\alpha\cdot  \cost(q_{i,\calS},\calS)  + 
		\left(\frac{2z}{\alpha}\right)^{z-1} \cost(p,q_{i,\calS})\\
 &\leq & \alpha\cdot 2^{\ell+3}\cdot \cost(q_{i, \calS},\greedy) + 
 		\left(\frac{2z}{\alpha}\right)^{z-1} \cdot 3 \cdot 2^{z-1}\cost(q_{i,\calS}, \greedy) \\
 &\leq & 2^{z+1}\cdot \left(\alpha\cdot 2^{\ell+3} + \left(\frac{2z}{\alpha}\right)^{z-1}\right)\cdot \cost(q_{i, \calS},\greedy).
\end{eqnarray*}
We optimize the final term with respect to
$\alpha$, which leads to $\alpha = 2^{-\frac{\ell}{z}}$ (ignoring constants that depend on $z$) and hence an upper bound of
$$\cost(p,\calS) - \cost(q_{i, \calS},\calS) 
  \leq 2^{O(z\log z)} 2^{\ell(1-1/z)} \cdot\cost(q_{i, \calS},\greedy).$$
\end{proof}

\subsubsection{Concentration of the Estimator $E_{\ell, \calS}(L_\calS)$}
First, we show that every estimator $E_{\ell, \calS(L_\calS)}$ is tightly concentrated. This follows the lines of the proof of \cref{lem:interesting-easy}, incorporating carefully the result of \cref{lem:kbinom}.

\begin{lemma}
\label{lem:kepsmagic}
Let $G$ be a group of points, and $\greedy$ be a solution. Consider an arbitrary solution $\calS$. Then for any set of interesting clusters $L_\calS$ induced by $\greedy$ on $G$, and any estimator $E_{\ell,\calS}(L_\calS)$ with $\ell\leq z\log 4z/\eps$,
it holds that: 
$$\vert E_{\ell,\calS}(L_\calS) - \mathbb{E}[E_{\ell,\calS}(L_\calS)]\vert \leq  \frac{\eps}{z\log z/\eps}\cdot \left(\cost(G,\greedy) + \cost(I_{\ell,\calS},\calS)\right) ,$$
with probability at least 
$$ 1-\exp\left(-2^{O(z\log z)}\cdot \frac{\min(\eps^2,\eps^{z})}{\log^2 1/\eps}\cdot \delta\right).\qedhere$$ 
\end{lemma}
\begin{proof}
In order to simplify the notations, we drop mention of $L_\calS$ and define $E_{\ell, \calS} = E_{\ell, \calS}(L_\calS)$.

Lemma~\ref{lem:kbinom} allows to write slightly differently $E_{\ell, \calS}$:
\[E_{\ell, \calS} = \sum_{\tilde C_i \in L_\calS} \sum_{p \in \tilde C_i \cap I_{\ell,\calS}\cap \Omega} \weight(p) \cdot w_p \cost(q_{i, \calS},\calS),\]
with all the weights $w_{p}$ are in $[0,2^{\ell(1-1/z)}\cdot 2^{O(z\log z)}]$.

We can also write $E_{\ell,\calS}$ as a sum of independent random variables: $E_{\ell,\calS} = \sum\limits_{j = 1}^\delta X_j$, where $X_j = \weight(\coreset_j)\cdot w_{\coreset_j}\cost(q_{i, \calS},\greedy)$ 
when the $j$-th sampled point of $G$ is $\coreset_j \in \tilde C_i \cap I_{\ell,\calS} \cap L_\calS$ and $X_j = 0$ when $\coreset_i \notin I_{\ell,\calS} \cap L_\calS$. 
Recall that,  due to \cref{fact:cost_rings}, the probability that the $j$-th sampled point is $p$, where $p \in \tilde C_i$ satisfies $\mathbb{P}[\coreset_j = p] = \frac{\cost(\tilde C_i,\greedy)}{|\tilde C_i|\cdot \cost(G,\greedy)} \leq \frac{2\cost(p,\greedy)}{\cost(G,\greedy)}$. From the same fact, $\weight(p) \leq \frac{2\cost(G,\greedy)}{\delta \cost(p,\greedy)}$.

We will rely on Bernstein's inequality (\cref{thm:bernstein}). To do this, we need an upper bound on the variance of $E_{\ell,\calS}$, as well as an almost sure upper bound $M$ on every sample. 
We first bound $\E[X_j^2]$: in the second line, we use that $\Omega_j$ consists of a single point to move the square inside the sum.
\begin{eqnarray*}
\nonumber
 \E[X_j^2] &=&  
 	\E\left[\left(
 		\weight(\coreset_j)\cost(\coreset_j,\greedy)\cdot w_{\coreset_j,\calS} \right)^2\right] \\
 		& = & \sum_{p \in I_{\ell,\calS} \cap L_\calS} 
		\left(\weight(p) \cost(p,\greedy)\cdot w_{p,\calS} \right)^2 \cdot \Pr\left[\coreset_i = p\right] \\ 
&\leq & \sum_{p \in I_{\ell,\calS}} 
		\left(\frac{2\cost(G,\greedy)}{\delta \cost(p,\greedy)}\cdot \cost(p,\greedy)\cdot w_{p,\calS} 
			\right)^2\cdot \Pr\left[\coreset_i = p\right] \\
&\leq &  \sum_{p \in I_{\ell,\calS}}   
		2^{2\ell(1-1/z)}\cdot 2^{O(z\log z)}\cdot  \frac{\cost^2(G,\greedy)}{\delta^2} \cdot \frac{\cost(p,\greedy)}{\cost(G,\greedy)}\\
&\leq &  \sum_{p\in I_{\ell,\calS}}   
		2^{2\ell(1-1/z)}\cdot 2^{O(z\log z)} \cdot \frac{\cost(G,\greedy)}{\delta^2}\cdot \cost(p,\greedy),
\end{eqnarray*}
where the fourth line follows from  using \cref{lem:kbinom} to replace the value of $w_{p,\calS}$.

To bound $ \sum_{p\in I_{\ell,\calS}} \cost(p,\greedy)$, we need to deal with the cases $z=1$ (i.e. $k$-median) and $z\geq 2$ ($k$-means and higher powers) separately.
For the former, we have $2^{2\ell(1-1/1)}=1$, so we can use $\sum_{p\in I_{\ell,\calS}} \cost(p,\greedy)\leq\cost(G,\greedy)$ as an upper bound. For the latter, we use $\sum_{p\in I_{\ell,\calS}} 2^{\ell}\cdot \cost(p,\greedy)\leq \cost(I_{\ell,\calS} ,\calS)$ as an upper bound. Combining this with $\text{Var}[X_i]  \leq \text{E}[X_i^2]$, we obtain for $z=1$:
\begin{eqnarray}
\text{Var}[X_i]  &\leq & 
 \frac{\cost(G,\greedy)}{\delta^2} \cdot 2^{O(z\log z)} \cdot \cost(G,\greedy),
\label{eq:kepsmagic1}
\end{eqnarray}
and for $z > 1$:
\begin{eqnarray}
\text{Var}[X_i]  &\leq & 
 \frac{\cost(G,\greedy)}{\delta^2}\cdot 2^{O(z\log z)} 2^{\ell(1-2/z)} \cost(I_{\ell,\calS},\calS).
\label{eq:kepsmagic1b}
\end{eqnarray}

The almost sure upper bound (for which no case distinction is required) can be derived similarly , using $X_i \leq \sup \frac{2\cost(G,\greedy)}{\delta \cost(p,\greedy)} \cdot \cost(p, \greedy) \cdot w_{p,\calS}$:
\begin{eqnarray}
\notag
X_i ~~\leq ~~M & := & 2^{\ell(1-1/z)}\cdot 2^{O(z\log z)}\cdot \frac{\cost(G,\greedy)}{\delta}\\
\label{eq:kepsmagic2}
& \leq & \frac{z}{\eps}\cdot 2^{\ell(1-2/z)}\cdot 2^{O(z\log z)}\cdot \frac{\cost(G,\greedy)}{\delta}, 
\end{eqnarray}
where the inequality holds due to $\ell \leq z\log(4z/\eps)$.
Applying Bernstein's inequality with Equations~\ref{eq:kepsmagic1},~\ref{eq:kepsmagic1b}, and~\ref{eq:kepsmagic2}, we then have
\begin{eqnarray}
\nonumber
& &\mathbb{P}\left[\vert E_{\ell,\calS} - \mathbb{E}[E_{\ell,\calS}]\vert \leq  
		\frac{\eps}{z\log z/\eps}\cdot
		\left(\cost(G,\greedy) + \cost(I_{\ell,\calS},\calS)\right) \right] \\
\nonumber
&\leq & \exp\left(- \frac{\frac{\eps^2}{z^2\log^2 z/\eps}
					\cdot \left(\cost(G,\greedy) + \cost(I_{\ell,\calS},\calS)\right)^2}
	{2\sum_{i=1}^\delta \text{Var}[X_i] + \frac{1}{3}M\cdot 
	  \frac{\eps}{z\log z/\eps}\cdot \left(\cost(G,\greedy) + \cost(I_{\ell,\calS},\calS)\right) }\right)\\
\nonumber
&\leq & \exp\left(- \frac{\frac{\eps^2}{z^2\log^2 z/\eps}\cdot \delta}
				{2^{O(z\log z)} \cdot 
				\begin{cases}
				    1 & \text{if } z=1 \\ 
					2^{\ell(1-2/z)} &\text{if }z\geq 2
				\end{cases}}\right)
\end{eqnarray}

For $z=1$ this becomes $\exp\left(-\frac{\eps^2\cdot \delta}{2^{O(z\log z)} \log^2 1/\eps}\right)$. For $z=2$, we have $2^{\ell(1-2/z)}=1$, so the same bound as for $z=1$. For $z>2$, we use $\ell \leq z\log 4z/\eps$, which implies $\eps^2 \cdot 2^{-\ell(1-2/z)}\geq \eps^{2+z-z2/z}\cdot 2^{-O(z\log z)} = \eps^z\cdot  2^{-O(z\log z)}$. This yields our final desired bound of
$$\exp\left(- \frac{\min(\eps^2,\eps^{z})}{2^{O(z\log z)}\log^2 1/\eps}\cdot \delta\right).$$
\end{proof}

\subsubsection{Concentration of $F_{\ell, \calS}(L_\calS)$}

We now turn our attention to bounding the random variable $F_{\ell,\calS}(L_\calS)$. It turns out that bounding 
$$F_{\ell, \calS}(L_\calS)=\sum_{\tilde C_i\in L_\calS}\sum_{p\in \tilde C_i\cap \Omega \cap I_{\ell,\calS}} \cost(q_{i, \calS},\calS)\cdot \frac{|\tilde C_i|\cdot \cost(G,\greedy)}{\delta\cost(\tilde C_i,\greedy)}$$ 
is rather hard, and in fact no easier than bounding $\cost(I_{\ell,\calS}\cap \coreset,\calS)$. Fortunately, this is not necessary, as it turns out that we can merely bound the sum of $F_{\ell, \calS}(L_\calS)$. 
We consider the random variable defined as follows:
\begin{align*}
F_{\calS}(L_\calS) &= \sum_{\ell  \leq z \log(4z/\eps)} F_{\ell,\calS} (L_\calS)
\end{align*}
with expectation
\begin{align*}
\mathbb{E}[F_{\calS}(L_\calS)] = \sum_{\tilde C_i\in L_{\calS}}\sum_{p\in \tilde C_i\cap \Omega} \cost(q_{i, \calS},\calS)\cdot \frac{|\tilde C_i|\cdot \cost(G,\greedy)}{\delta\cost(\tilde C_i,\greedy)}
\end{align*}
Showing that $F_{\calS}(L_\calS)$ is concentrated is now an almost direct consequence of event $\calE$ from Lemma~\ref{lem:ksize}, which says that $\sum_{p \in \tilde C_i \cap \coreset} \frac{|\tilde C_i|\cdot \cost(G,\greedy)}{\delta\cost(\tilde C_i,\greedy)} = (1\pm \eps) |\tilde C_i|$.

\begin{lemma}
\mbox{}\label{lem:kepsF}
Let $G$ be a group of points, and $\greedy$ be a solution.
Conditioned on event $\calE$, we have for all solutions $\calS$ and all sets of interesting clusters $L_\calS$ induced by $\greedy$ on $G$:
$$|F_{\calS}(L_\calS) - \mathbb{E}[F_{\calS}(L_\calS)]| \leq \eps \cdot \cost(G,\calS).\qedhere$$
\end{lemma}
\begin{proof}
Given a solution $\calS$ and any set of interesting clusters $L_\calS$ induced by $\greedy$ on $G$, we have
\begin{align*}
  \mathbb{E}[F_{\calS}(L_\calS) ] &=  \sum_{\tilde C_i\in L_{\calS}}\sum_{p\in \tilde C_i} \cost(q_{i, \calS},\calS)\cdot \frac{|\tilde C_i|\cdot \cost(G,\greedy)}{\delta\cost(\tilde C_i,\greedy)} \Pr[p \in \Omega] = \sum_{\tilde C_i\in L_{\calS}} |\tilde C_i| \cdot \cost(q_{i, \calS},\calS).
\end{align*}

  Event $\calE$ ensures that the mass of each cluster is preserved in the coreset, i.e., that $\sum_{p \in \tilde C_i\cap \coreset} \frac{|\tilde C_i|\cdot \cost(G,\greedy)}{\delta\cost(\tilde C_i,\greedy)}  = (1 \pm\eps)\cdot |\tilde C_i|$, for every cluster $\tilde C_i\in L_{\calS}$. Hence 
\[F_{\calS}(L_\calS)=\sum_{\tilde C_i\in L_{\calS}}\sum_{p \in \tilde C_i\cap \Omega}  \cost(q_{i, 	\calS},\calS) \cdot \frac{|\tilde C_i|\cdot \cost(G,\greedy)}{\delta \cost(\tilde C_i,\greedy)}  = (1\pm \eps) \cdot \mathbb{E}[F_\calS(L_\calS)].\]
Now finally observe that since $q_{i, \calS}$ was always 
  the point of $\tilde C_i$ whose cost in $\calS$ is the smallest, 
   we have $\mathbb{E}[F_S(L_\calS)] \leq \cost(L_{\calS},\calS) \leq \cost(G,\calS)$.
\end{proof}

\subsection{Combining Them All} \label{sec:combining}

We can now show that the sample $\coreset$ indeed verifies \cref{lem:coreset-reduc}. To do that, we naturally follow the structure of previous lemmas, and decompose
$$\left\vert \cost(G,\calS) - \sum_{p\in\coreset}\weight(p) \cdot \cost(p,\calS)\right\vert$$
into terms for which we can apply Lemmas~\ref{lem:kepstiny},~\ref{lem:khuge},~\ref{lem:kepsmagic}, and~\ref{lem:kepsF}.

First, we note that the probability of success of \cref{lem:kepsmagic} is too small to take a union-bound over its success for all $\calS$. To cope with that issue, we use the approximate centroid set, in order to relate $E_{\ell, \calS}(L_\calS)$ to $E_{\ell, \tilde \calS}(L_\calS)$, where $\tilde \calS$ comes from a small set on which union-bounding is possible.

\begin{lemma}
\mbox{}\label{lem:union-bound}
Let $G$ be a group of points, and $\greedy$ be a solution. Let $\cand$ be an $\greedy$-approximate centroid set, as in \cref{def:centroid-set}.
It holds with probability 
$$ 1-\exp\left(k\log |\cand| - 2^{O(z\log z)}\cdot \frac{\min(\eps^2,\eps^{z})}{\log^2 1/\eps}\cdot \delta\right)$$ 

 that, for all solution $\tilde \calS \in \cand^k$ and any set of interesting clusters $L_{\tilde \calS}$ induced by $\greedy$ on $G$:
\begin{align*}
\left| \cost(L_{\tilde \calS}, \tilde \calS) - \cost(\coreset \cap L_{\tilde \calS}, \tilde \calS)\right| \leq \eps\left(\cost(G, \greedy) + \cost(L_{\tilde \calS}, \tilde \calS)\right)\qedhere.
\end{align*} 
\end{lemma}
\begin{proof}
Taking a union-bound over the success of \cref{lem:kepsmagic} for all possible $\tilde \calS \in \cand^k$, all choice of interesting clusters $L_{\tilde \calS}$ and all $\ell$ such that $\log(\eps/2) \leq \ell \leq z\log(4z/\eps)$, it holds with probability  $1-\exp(k \log |\cand|)\exp\left(-2^{O(z\log z)}\cdot \frac{\min(\eps^2,\eps^{z})}{\log^2 1/\eps}\cdot \delta\right)$ that, for every $\tilde \calS \in \cand^k, L_{\tilde \calS}$ and $\ell$,

\begin{equation}
\label{eq:unionbound-tilde}
\vert E_{\ell,\tilde \calS}(L_{\tilde \calS}) - \mathbb{E}[E_{\ell,\tilde \calS}(L_{\tilde \calS})]\vert \leq
  \frac{\eps}{z\log z/\eps}\cdot \left(\cost(G,\greedy) + \cost(I_{\ell,\tilde \calS},\tilde \calS)\right)
\end{equation}

For simplicity, we drop again the mention of $L_{\tilde \calS}$ and write $ E_{\ell,\tilde \calS} =  E_{\ell,\tilde \calS}(L_{\tilde \calS})$, $F_{\tilde \calS} =  F_{\tilde \calS}(L_{\tilde \calS})$.
We now condition on that event, together with event $\calE$. We write:

\begin{eqnarray}
\notag
& &\left\vert \sum_{p \in  L_{\tilde \calS}} \cost(p, \tilde \calS) \right.  -  \left.
\sum_{p \in L_{\tilde \calS} \cap \coreset} \weight(p)\cdot \cost(p, \tilde \calS)\right\vert \\
\notag
&=& \left\vert \sum_{p \in L_{\tilde \calS}} \cost(p, \tilde \calS) - 
\mathbb{E}[F_{\tilde \calS}] + \mathbb{E}[F_{\tilde \calS}] -F_{\tilde \calS}  + F_{\tilde \calS} 
- \sum_{p \in L_{\tilde \calS} \cap \coreset} \weight(p) \cdot \cost(p, \tilde \calS)\right\vert \\
\notag
&\leq & \left\vert \sum_{p \in L_{\tilde \calS}} \cost(p, \tilde \calS) - 
	\mathbb{E}[F_{\tilde \calS}]  + F_{\tilde \calS} - 
	\sum_{p \in L_{\tilde \calS} \cap \coreset} \weight(p) \cdot \cost(p, \tilde \calS)\right\vert 	+ |\mathbb{E}[F_{\tilde \calS}] -F_{\tilde \calS}| \\
\label{eq:tiny}
&\leq &
\sum_{\ell <\log \eps/2}
\left\vert \sum_{p \in I_{\ell,\tilde \calS} \cap L_{\tilde \calS}} \cost(p,\tilde  \calS) - \mathbb{E}[F_{\ell,\tilde \calS}]  + 
		F_{\ell,\tilde \calS} - \sum_{\mathclap{p \in I_{\ell, \tilde \calS} \cap L_{\tilde \calS} \cap \coreset}} \weight(p)\cdot \cost(p, \tilde \calS)\right\vert \\ 
\label{eq:interest}
& & +\sum_{\ell=\log \eps/2}^{z\log z/4\eps} 
\left\vert \sum_{p \in I_{\ell,\tilde \calS}\cap L_{\tilde \calS}} \cost(p,  \tilde \calS) - \mathbb{E}[F_{\ell,\tilde \calS}]  + 
		F_{\ell,\tilde \calS} - \sum_{\mathclap{p \in I_{\ell, \tilde \calS} \cap L_{\tilde \calS} \cap \coreset}} \weight(p)\cdot \cost(p, \tilde \calS)\right\vert \\
\notag
& & + |\mathbb{E}[F_{\tilde \calS}] -F_{\tilde \calS}|
\end{eqnarray}

We note that Equation~\ref{eq:interest} is $\sum_{\ell=\log \eps/2}^{z\log z/4\eps} |E_{\ell, \tilde \calS} - \E[E_{\ell, \tilde \calS}]|$ and can be directly bounded using Equation~\ref{eq:unionbound-tilde}. To bound tiny points of Equation~\ref{eq:tiny}, we combine \cref{lem:kepstiny} with the observation that 
$F_{\ell, \tilde \calS} \leq \sum_{p \in I_{\ell, \tilde \calS} \cap \coreset} \weight(p)\cost(p, \tilde \calS)$. This gives: 
\begin{eqnarray*}
& & \sum_{\ell <\log \eps/2}
\left\vert \sum_{p \in I_{\ell,\tilde \calS} \cap L_{\tilde \calS}} \cost(p,  \tilde \calS) - \mathbb{E}[F_{\ell,\tilde \calS}]  + 
		F_{\ell,\tilde \calS} - \sum_{p \in I_{\ell, \tilde \calS} \cap \coreset} \weight(p)\cdot \cost(p, \tilde \calS)\right\vert\\
&\leq & \sum_{\ell <\log \eps/2} \left(
\sum_{p \in I_{\ell,\tilde \calS}} \cost(p, \tilde \calS) + \mathbb{E}[F_{\ell,\tilde \calS}]  + 
		F_{\ell,\tilde \calS} + \sum_{p \in I_{\ell, \tilde \calS} \cap \coreset} \weight(p)\cdot \cost(p, \tilde \calS)\right)\\
&\leq & 2 \sum_{\ell <\log \eps/2} \left(
\sum_{p \in I_{\ell,\tilde \calS}} \cost(p, \tilde \calS) + \sum_{p \in I_{\ell, \tilde \calS} \cap \coreset} \weight(p)\cdot \cost(p, \tilde \calS)\right)\\
&\leq & 4\eps \cost(G, \greedy),
\end{eqnarray*}
where the last equation uses \cref{lem:kepstiny}. Plugging this result into the previous inequality, we have:{\setlength{\emergencystretch}{2.5em}\par}

\begin{eqnarray*}
& & \left\vert  \sum_{p \in L_{\tilde \calS}} \cost(p, \tilde \calS)  - 
\sum_{p \in L_{\tilde \calS} \cap \coreset} \weight(p)\cdot \cost(p, \tilde \calS)\right\vert\\
& \leq &  4\eps \cost(G,\greedy) + \sum_{\ell=\log \eps/2}^{z\log z/4\eps} 
\left\vert \mathbb{E}[E_{\ell, \tilde \calS}] - E_{\ell, \tilde \calS}\right\vert  
+ |\mathbb{E}[F_{\tilde \calS}] -F_{\tilde \calS}| \\
& \leq &  4\eps \cost(G, \greedy) + \sum_{\ell=\log \eps/2}^{z\log z/4\eps} 
\frac{\eps}{z\log z/\eps} \cdot \left(\cost(G, \greedy) + \cost(I_{\ell,\tilde \calS}, \tilde \calS)\right)
+ |\mathbb{E}[F_{\tilde \calS}] -F_{\tilde \calS}| \\
&\leq &  4\eps \cost(G, \greedy) +   \left(z\log (z/4\eps) - \log \eps /2 \right)\cdot \frac{\eps}{z \log z/\eps} \cdot \left(\cost(G, \greedy) + \cost(L_{\tilde \calS}, \tilde \calS) \right)\\
& & + \eps\cdot \cost(G, \tilde \calS) \\
&\leq & O(\eps) \cdot (\cost(G, \greedy) + \cost(L_{\tilde \calS}, \tilde \calS)),
\end{eqnarray*}

where the second to last inequality used Lemma~\ref{lem:kepsF}.

\end{proof}

\paragraph{From the approximate centroid set to any solution.}
We can now finally turn to the proof of \cref{lem:coreset-reduc}: it combines the result we show previously for the huge type, and the use of approximate centroid set with the \cref{lem:union-bound} for the interesting and tiny types.

\begin{proof}[Proof of \cref{lem:coreset-reduc}]
Let $X, k, z, G$ and $\greedy$ as in the lemma statement. 
We condition on event $\calE$ happening. Let $\calS$ be a set of $k$ points, and $\tilde \calS \in \cand^k$ that approximates best $\calS$, as given by the definition of $\cand$ (see \cref{def:centroid-set}). This ensures that 
for all points $p$ with $\dist(p, \calS) \leq \frac{8z}{\eps}\cdot \dist(p, \greedy)$ or $\dist(p, \tilde \calS) \leq \frac{8z}{\eps}\cdot \dist(p, \greedy)$ , we have $|\cost(p, \calS) - \cost(p, \tilde \calS)|\leq \frac{\eps}{z \log (z/\eps)}(\cost(p, \calS) + \cost(p, \greedy))$.

Our first step is to deal with points that have $\dist(p, \calS) > \frac{4z}{\eps}\cdot \dist(p, \greedy)$, using \cref{lem:khuge}.
None of the remaining points is huge with respect to $\tilde \calS$: hence, they all are in interesting clusters with respect to $\tilde \calS$. Let $L_{\tilde \calS}$ be this set of cluster: it can be handled with \cref{lem:union-bound}. The remaining of the proof formalizes the argument.

Let $H_{\calS}$ be the set of all clusters that are intersecting with some $I_{\ell,\calS}$ with $\ell > z \log(4z/\eps)$. We also denote $H_\calS$ the points contained in those clusters.
We decompose the cost difference as follows:

\begin{align}
\label{eq:kepsdecompose3}
\left\vert \cost(G,\calS) - \sum_{\mathclap{p\in\Omega\cap G}}\weight(p)\cdot \cost(p,\calS)\right\vert 
\leq & \left\vert \sum_{p \in G \setminus H_\calS} \cost(p, \calS) - \sum_{\mathclap{p \in (G \setminus H_\calS) \cap \coreset}} \weight(p)\cdot \cost(p,\calS)\right\vert \\
\label{eq:kepsdecompose4}
& + \left\vert \sum_{p \in H_{\calS}} \cost(p, \calS) - \sum_{\mathclap{p \in H_{\calS} \cap \coreset}} \weight(p)\cdot \cost(p,\calS)\right\vert 
\end{align}
Since we condition on event $\calE$, the  term~\ref{eq:kepsdecompose4} is $O(\eps)\cdot (\cost(G,\greedy)+\cost(G,\calS))$, using Lemma~\ref{lem:khuge}. Now we take a closer look at term~\ref{eq:kepsdecompose3}. By definition of $\tilde \calS$, it holds for all points $p \in G \setminus H_\calS$ that $|\cost(p, \calS) - \cost(p, \tilde \calS)|\leq \eps(\cost(p, \calS) + \cost(p, \greedy))$. Therefore:
\begin{eqnarray}
\notag
\left\vert \sum_{p \in G \setminus H_\calS} \cost(p, \calS) \right.  & - &  \left.
\sum_{p \in (G \setminus H_\calS) \cap \coreset} \weight(p)\cdot \cost(p, \calS)\right\vert \\
\notag
& \leq & \left\vert \sum_{p \in G \setminus H_\calS} \cost(p, \tilde \calS) \right.  -  \left.
\sum_{p \in (G \setminus H_\calS) \cap \coreset} \weight(p)\cdot \cost(p, \tilde \calS)\right\vert \\
\notag
& + & \eps \left(\cost(G, \calS) + \cost(G, \greedy) + \cost(\coreset, \calS) + \cost(\coreset, \greedy)\right).
\end{eqnarray}

This allows us to focus on bounding the cost difference to solution $\tilde \calS$ instead of $\calS$. 

For the remaining points in $G \setminus H_\calS$, we aim at using \cref{lem:union-bound}: for that, we show that $L_{\tilde \calS} := G \setminus H_\calS$ contains only interesting clusters with respect to $\tilde \calS$. 
Indeed, for any $p \in L_{\tilde \calS}$, we have $|\cost(p, \calS) - \cost(p, \tilde \calS)|\leq \frac{\varepsilon}{z\log z/\varepsilon}(\cost(p, \calS) + \cost(p, \greedy))$ by definition of $\tilde \calS$. Hence,
\begin{align*}
\cost(p, \tilde \calS) &\leq \cost(p, \calS) + \frac{\varepsilon}{z\log z/\varepsilon}(\cost(p, \calS) + \cost(p, \greedy)) \\
&\leq \left((1+\varepsilon)\left(\frac{4\eps}{z}\right)^z +\eps\right)\cost(p, \greedy) \\
&\leq \left(\frac{8\eps}{z}\right)^z\cost(p, \greedy),
\end{align*}
and $p$ is indeed not huge with respect to $\tilde \calS$. Therefore, we can apply \cref{lem:union-bound} to get:
\begin{align*}
\left\vert \sum_{p \in G \setminus H_\calS}  \cost(p, \tilde \calS)  -
\sum_{\mathclap{p \in (G \setminus H_\calS) \cap \coreset}} \weight(p)\cdot \cost(p, \tilde \calS)\right\vert 
 &= \left\vert \sum_{p \in L_{\tilde \calS}} \cost(p, \tilde \calS)  - 
\sum_{p \in L_{\tilde \calS} \cap \coreset} \weight(p)\cdot \cost(p, \tilde \calS)\right\vert \\
&\leq \eps(\cost(G, \greedy) + \cost(L_{\tilde \calS}, \tilde \calS))\\
&= O(\eps)(\cost(L_{\tilde \calS}, \calS) + \cost(G, \greedy))
\end{align*}

Combining all the equations yields
\begin{eqnarray*}
\left\vert \cost(G,\calS) - \cost(\coreset, \calS)\right\vert \leq   O(\eps)\cdot \left(\cost(G,\greedy) + \cost(G,\calS) + \cost(\coreset,\greedy) + \cost(\coreset,\calS)\right).
\end{eqnarray*}

To conclude the proof, it only remains to remove the term $\cost(\coreset, \greedy) + \cost(\coreset, \calS)$ from the right-hand-side. Applying this inequality for $\calS = \greedy$ and using $\cost(\coreset, \greedy) \leq \cost(G, \greedy) + \left\vert \cost(G,\greedy) - \cost(\coreset, \greedy)\right\vert$ yields first 
\begin{eqnarray*}
\cost(\coreset, \greedy) =  O(1)\cdot \cost(G,\greedy).
\end{eqnarray*}
Similarly, we can use $\cost(\coreset, \calS) \leq \cost(G, \calS) + \left\vert \cost(G,\calS) - \cost(\coreset, \calS)\right\vert$ to get
\begin{eqnarray*}
\cost(\coreset, \calS) =  O(1)\cdot \big(\cost(G,\calS) + \cost(G, \greedy)\big).
\end{eqnarray*}

Hence, we finally conclude:
\begin{eqnarray*}
\left\vert \cost(G,\calS) - \cost(\coreset, \calS)\right\vert \leq  O(\eps)\cdot \left(\cost(G,\greedy) + \cost(G,\calS)\right).
\end{eqnarray*}

\bigskip 
The probability now follows from taking a union-bound over the failure probability of \cref{lem:ksize} and \cref{lem:union-bound}. Specifically
\begin{align*}
&1-\exp\left(k\log |\cand| - 2^{O(z\log z)}\cdot \frac{\min(\eps^2,\eps^{z})}{\log^2 1/\eps}\cdot \delta\right) -  k\cdot z^2\log^2(z/\varepsilon) \exp\left(-O(1)\frac{\varepsilon^2}{k}\delta\right) \\
\end{align*}

In a given cluster $\tilde C_i$ induced by $\greedy$ on $G$, the complexity of the algorithm is $O(|\tilde C_i|)$: it is both the cost of computing the scaling factor $\weight(p)$ for all $p \in \tilde C_i$, and the cost of sampling $\delta$ points using reservoir sampling~\cite{Vitter85}. Hence, the cost of this algorithm for all clusters is $O(|G|)$.
\end{proof}

\section{Sampling from Outer Rings}\label{sec:sampleout}

In this section we prove \cref{lem:coreset-outer}:

\coresetouter*

Recall that the \texttt{SensitivitySample} procedure merely picks $\delta$ points $p$ with probability $\frac{\cost(p,\greedy)}{\cost(G, \greedy)}$. Each of the $\delta$ sampled points has a weight $\frac{\cost(G, \greedy)}{\delta\cdot \cost(p,\greedy)}$. The procedure runs in time $O(|G|)$.

The main steps of the proof are as follows. 
\begin{itemize}
\item First, we consider the cost of the points in $G$ such that $\cost(p,\calS)$ is at most $4^z\cdot\cost(p,\greedy)$.
For this case, we can (almost) directly apply Bernstein's inequality as in the previous section. 
\item Second, we consider the cost of the points in $G$ such that $\cost(p,\calS)> 4^z\cdot\cost(p,\greedy)$. Denote this set by $G_{far, \calS}$.	  
  For these points, we can afford to replace their cost in $\calS$ with the distance to the closest center $c \in \greedy$ plus the distance from $c$ to
  the closest center in $\calS$. The latter part can be charged to the remaining points of the cluster from the original dataset (i.e., not restricted to group $G$) which are in much larger number and already
  paying a similar value in $\calS$. 
\end{itemize}

We first analyse the points not in $G_{far,\calS}$. For that, we will go through the approximate centroid set $\cand$ to afford a union-bound: we show the following lemma. 
\begin{lemma}\label{lem:outerclose}
Let $\tilde \calS \in \cand^k$, and define $ G_{close, \tilde \calS}$ to be the set of points of $G$ such that $\cost(p, \tilde \calS)\leq 5^z\cdot\cost(p,\greedy)$. It holds with probability 
\[1-\exp\left( - 2^{-O(z)} \left(\frac{\varepsilon}{\log 1/\varepsilon}\right)^2 \delta\right)\]
that
\[|\cost(G_{close, \tilde \calS}, \tilde \calS) - \cost(\coreset \cap G_{close, \tilde \calS}, \tilde \calS)| \leq \frac{\varepsilon}{z\log z/\varepsilon} \left(\cost(G,\greedy) + \cost(G_{close, \tilde \calS}, \tilde \calS)\right)\]
\end{lemma}
\begin{proof}
We aim to use Bernstein's Inequality. Let $E_{close, \tilde \calS}= \sum_{i=1}^{\delta} X_i$, where $X_i = \frac{\cost(G,\greedy)}{\delta\cdot \cost(p,\greedy)} \cdot \cost(p, \tilde \calS)$ if the $i$-th sampled point is $p\in G_{close, \tilde \calS}$ and $X_i=0$ the $i$-th sampled point is $p\notin G_{close, \tilde \calS}$. Recall that the probability that $p$ is the $i$-th sampled point is $\frac{\cost(p,\greedy)}{\cost(G,\greedy)}$.
We consider the second moment $\text{E}[X_i^2]$:

\begin{eqnarray*}
E[X_i^2]  &=& \sum_{p\in G_{close,\tilde \calS}} \left(\frac{\cost(G,\greedy)}{\delta\cdot \cost(p,\greedy)} \cdot \cost(p,\tilde \calS)\right)^2 \cdot \mathbb{P}[p \in \coreset] \\
&=& \cost(G,\greedy)\cdot\sum_{p\in G_{close,\tilde \calS}} \frac{\cost(p, \tilde \calS)}{\delta^2\cdot \cost(p,\greedy)} \cdot \cost(p,\tilde \calS) \\
&\leq & \cost(G,\greedy)\cdot\sum_{p\in G_{close,\tilde \calS}} \frac{5^z}{\delta^2} \cdot \cost(p, \tilde \calS) \\
&\leq & \frac{5^z}{\delta^2} \cdot \cost(G,\greedy) \cdot \cost(G_{close, \tilde \calS}, \tilde \calS)
\end{eqnarray*}
	
Furthermore, we have the following upper bound for the maximum value any of the $X_i$:
\begin{equation}
X_i\leq M:= \max_{p\in G_{close, \tilde \calS}} \frac{\cost(G,\greedy)}{\delta\cdot \cost(p,\greedy)} \cdot \cost(p, \tilde \calS) \leq \frac{5^z}{\delta}\cdot \cost(G,\greedy).
\end{equation}

Combining both bounds with Bernstein's inequality now yields
\begin{eqnarray*}
& &\mathbb{P}[|E_{close, \tilde \calS} - \mathbb{E}[E_{close, \tilde \calS}]| \leq \frac{\varepsilon}{z\log z/\varepsilon}\cdot \left(\cost(G,\greedy) + \cost(G_{close, \tilde \calS},\tilde \calS)\right)] \\
&\leq & \exp\left(-\frac{\left(\frac{\varepsilon}{z\log z/\varepsilon}\right)^2 \cdot \left(\cost(G,\greedy) + \cost(G_{close, \tilde \calS},\tilde \calS)\right)^2}{2 \sum_{i=1}^\delta Var[X_i] + \frac{1}{3} M \cdot \varepsilon\cdot \left(\cost(G,\greedy) + \cost(G_{close, \tilde \calS},\tilde \calS) \right) }\right) \\
&\leq &\exp\left(-\frac{\left(\frac{\varepsilon}{z\log z/\varepsilon}\right)^2 \cdot \delta \cdot \left(\cost(G,\greedy) + \cost(G_{close, \tilde \calS}, \tilde \calS)\right)^2}{2 4^z  \cdot \cost(G,\greedy) \cdot \cost(G_{close, \tilde \calS},\tilde \calS) + 4^z\cdot \cost(G,\greedy) \cdot \varepsilon\cdot \left(\cost(G,\greedy) + \cost(G_{close, \tilde \calS},\tilde \calS) \right) }\right) \\
&\leq & \exp\left(-2^{-O(z)}\cdot \left(\frac{\varepsilon}{z\log z/\varepsilon}\right)^2 \cdot \delta \right) 
\end{eqnarray*}

Noting that $\cost(\coreset \cap G_{close, \tilde \calS}, \tilde \calS) = \mathbb{E}[E_{close, \tilde \calS}]$, concludes: we have with probability

$1-\exp\left(-2^{-O(z)}\cdot \left(\frac{\varepsilon}{\log 1/\varepsilon}\right)^2 \cdot \delta \right)$ that: 
\begin{equation*}
|\cost(G_{close, \tilde \calS}, \tilde \calS) - \cost(\coreset \cap G_{close, \tilde \calS}, \tilde \calS)| \leq \frac{\varepsilon}{z\log z/\varepsilon}\cdot \left(\cost(G,\greedy) + \cost(G_{close, \tilde \calS}, \tilde \calS)\right)
\end{equation*}
\end{proof}

Now we turn our attention to $G_{far, \calS}$. For this, we analyse the following event  $\mathcal{E}_{far}$, similar to $\calE$:
For all cluster $C$ of solution $\greedy$ such that $C \cap G\neq \emptyset$  
\begin{equation*}
\sum_{p\in C\cap G\cap \coreset} \frac{\cost(G,\greedy)}{\delta\cdot\cost(p,\greedy)} \cost(p,\greedy) = (1\pm \varepsilon)\cdot \cost(C\cap G,\greedy)
\end{equation*}

\begin{lemma}
\label{lem:eventEFar}
Event $\calE_{far}$ happens with probability at least 
\[1 - k\exp\left(\frac{\eps^2}{6k}\cdot \delta\right).\qedhere\]
\end{lemma}
\begin{proof}
We aim to use Bernstein's Inequality. Let $E_{C}= \sum_{i=1}^{\delta} X_i$, where $X_i = \frac{\cost(G,\greedy)}{\delta\cdot \cost(p,\greedy)} \cdot \cost(p,\greedy)$ if the $i$-th sampled point $p\in C$ and $X_i=0$ the $i$-th sampled point $p\notin C$. Recall that the probability that the $i$-th sampled point is $p$ is $\frac{\cost(p,\greedy)}{\cost(G,\greedy)}$.
We consider the second moment $\text{E}[X_i^2]$:
\begin{eqnarray*}
E[X_i^2]  &=& \sum_{p\in C\cap G} \left(\frac{\cost(G,\greedy)}{\delta\cdot \cost(p,\greedy)} \cdot \cost(p,\greedy)\right)^2 \cdot \mathbb{P}[p \text{ is the } i\text{-th sampled point}] \\
&=& \frac{\cost(G,\greedy)}{\delta^2} \cdot\sum_{p\in C\cap G}\cost(p,\greedy)  \\
&=& \frac{\cost(G,\greedy)}{\delta^2} \cost(C\cap G,\greedy)  \\
&\leq & \frac{2k}{\delta^2}\cdot \cost^2(C\cap G,\greedy)
\end{eqnarray*}
where the final inequality follows since every cluster has cost at least half the average. Indeed, either the group considered is $ \outergroup{}_{\max}$, and then any cluster verifies $\cost(C \cap G) \geq \frac{1}{k}\cost(\out^\greedy, \greedy) \geq \frac{1}{k}\cost(\outergroup{}_{\max},\greedy)$, or all the clusters in $\outergroup{}_{b}$ have an equal cost, up to a factor of $2$ -- hence none cost less than half of the average.{\setlength{\emergencystretch}{2.5em}\par}

Furthermore, we have by the same argument the following upper bound for the maximum value any of the $X_i$:
\begin{equation*}
X_i\leq M:= \max_{p\in C\cap G} \frac{\cost(G,\greedy)}{\delta\cdot \cost(p,\greedy)} \cdot \cost(p,\greedy) \leq \frac{2k}{ \delta}\cdot \cost(C \cap G,\greedy).
\end{equation*}

Combining both bounds with Bernstein's inequality now yields 
\begin{eqnarray*}
& &\mathbb{P}[|\cost(C\cap G\cap \coreset,\greedy) - \cost(C\cap G,\greedy)| \leq \varepsilon\cdot \cost(C\cap G,\greedy)] \\
&\leq & \exp\left(-\frac{\varepsilon^2 \cdot \cost^2(C\cap G,\greedy)}{2 \sum_{i=1}^\delta Var[X_i] + \frac{1}{3} M \cdot \varepsilon\cdot \cost(C\cap G,\greedy) }\right) \leq \exp\left(-\frac{\varepsilon^2}{6k} \cdot \delta\right)
\end{eqnarray*}

Reformulating, we now have  

\begin{equation*}
\sum_{p\in C\cap G\cap \coreset} \frac{\cost(G,\greedy)}{\delta\cdot\cost(p,\greedy)} \cost(p,\greedy) = (1\pm \varepsilon)\cdot \cost(C\cap G,\greedy)
\end{equation*}
\end{proof}

\begin{lemma}
\label{lem:outerfar}
Let $(X, \dist)$ be a metric space, $k, z$ be two positive integers. Suppose $G$ is either a group $\outergroup{}_{b}$ or $\outergroup{\greedy, \seeded}_{\max}$. Let $G_{far, \calS}\subset G$ be the set of all clients such that $\cost(p,\calS)> 4^z \cdot \cost(p,\greedy)$. Condition on event $\calE_{far}$.

Then, the set $\coreset$ of size $\delta$ constructed by \texttt{SensitivitySample} verifies the following.
  It holds for all sets $\calS$ of $k$ centers that:
\[\cost(G_{far, \calS}, \calS) + \cost(\coreset\cap G_{far, \calS},\calS) \leq \frac{2\eps}{z \log z/\eps} \cdot \cost(\calS).\qedhere\]
\end{lemma}
\begin{proof}
Our aim will be to show that $\max\left(\cost(G_{far, \calS}, \calS), \cost(\coreset\cap G_{far, \calS},\calS)\right) \leq \frac{\eps}{z \log z/\eps} \cdot \cost(\calS)$. It is key here that we compare to the cost of the full input in $\calS$, and not simply the cost of the group $G$.{\setlength{\emergencystretch}{2.5em}\par}

First, we fix a cluster $C \in \greedy$, and show that the total contribution of points of $C\cap G_{far, \calS}$ is very cheap 
compared to $\cost(C,\calS)$, i.e. that $\cost(G_{far, \calS}\cap C,\calS) \leq \frac{\eps}{z \log z/\eps}\cdot \cost(C, \calS)$.

For this, fix a point $p \in G_{far, \calS} \cap C$, and let $c$ be the center of cluster $C$.

Let $C_{close}$ be the points of $C$ with  cost at most $\left(\frac{z}{\varepsilon}\right)^z \cdot \frac{\cost(C,\greedy)}{|C|}$. Due to Markov's inequality, most of $C$'s points are in $C_{close}$: $|C_{close}| \geq  (1-\varepsilon/z)\cdot |C|$.

Using that the point $p$ is both in the outer ring of $C$ and in $G_{far, \calS}$, we can lower bound the distance from $c$ to $\calS$ as follows. Triangle inequality and $\cost(p,\calS) > 4^z \cdot \cost(p,c)$, yield $\dist(c,\calS) \geq \dist(p,\calS)  - \dist(p,c) \geq 4 \dist(p,c) - \dist(p,c) \geq 3\dist(p,c)$. Since $p$ is from an outer group, it verifies $\cost(p, c)\geq \left(\frac{z}{\varepsilon}\right)^{2z} \cdot \frac{\cost(C,c)}{|C|}$. Combining those two observations yields:
$\cost(c,\calS) \geq 3^z \cost(p, \greedy) \geq 3^z \cdot \left(\frac{z}{\varepsilon}\right)^{2z} \cdot \frac{\cost(C,c)}{|C|}$.

 Using this and \cref{lem:weaktri}, we now have for any $q\in C_{close}$:
\begin{eqnarray}
\nonumber
\cost(c,\calS) &\leq & (1+\varepsilon/(2z))^{z-1}\cdot \cost(q,\calS) + \left(\frac{2z+\varepsilon}{\varepsilon}\right)^{z-1}\cdot \cost(q,c) \\
\nonumber
&\leq & (1+\varepsilon) \cost(q,\calS) + \left(\frac{2z+\varepsilon}{\varepsilon}\right)^{z-1} \cdot \left(\frac{z}{\varepsilon}\right)^{z} \cdot \frac{\cost(C,c)}{|C|}  \\
\nonumber
&\leq & (1+\varepsilon) \cost(q,\calS) + 3^{z-1}\cdot \left(\frac{z}{\varepsilon}\right)^{2z-1} \cdot \frac{\cost(C,c)}{|C|} \\
\nonumber
&\leq & (1+\varepsilon) \cost(q,\calS) + \frac{\eps}{3z} \cdot \cost(c,\calS) \\
\nonumber 
\Rightarrow \cost(q,\calS) & \geq &  \frac{1-\varepsilon}{1+\varepsilon} \cdot \cost(c,\calS) \\
\label{eq:outerfar2}
\Rightarrow \cost(C,\calS) &\geq & \cost(C_{close},\calS) \geq  |C_{close}|\cdot \frac{1-\varepsilon}{1+\varepsilon} \cdot cost(c,\calS).
\end{eqnarray}

Using additionally that $|C_{close}| \geq (1-\frac{\eps}{z})\cdot |C|$ and $\cost(c,\calS) \geq 3^z \cdot \left(\frac{z}{\varepsilon}\right)^{2z} \cdot \frac{\cost(C,c)}{|C|}$, we get:
\begin{equation}
\label{eq:outerfar3}
\cost(C,\calS) \geq |C_{close}| \cdot \frac{1-\varepsilon}{1+\varepsilon} \cdot 3^z \cdot \left(\frac{z}{\varepsilon}\right)^{2z}  \cdot \frac{\cost(C,\greedy)}{|C|} \geq 3^z \cdot \left(\frac{z}{\varepsilon}\right)^{2z-1}  \cdot \cost(C,\greedy).
\end{equation}

We are now equipped to show the first part of the lemma, namely $\cost(G_{far, \calS},\calS) \leq \frac{\eps}{z \log z/\eps}\cdot \cost(\calS)$.

Since $G \cap C$ contains only points from the outer ring of $C$, with distance at least $(z/\eps)^2$ times the average, Markov's inequality implies that $|G \cap C|\leq \left(\frac{\eps}{z}\right)^2 \cdot |C|$. Hence, $|G_{far, \calS}\cap C| \leq  \frac{1}{1-\eps/z} \cdot \left(\frac{\varepsilon}{z}\right)^{2} \cdot |C_{close}|$.
This yields

\begin{eqnarray}
\nonumber
& &\cost(G_{far, \calS}\cap C,\calS) = \sum_{p\in G_{far, \calS}\cap C} \cost(p,\calS) \\
\nonumber
(Lem. \ref{lem:weaktri}) &\leq & \sum_{p\in G_{far, \calS}\cap C} (1+\varepsilon/2z)^{z-1} \cost(c,\calS) + \left(\frac{2z+\varepsilon}{\varepsilon}\right)^{z-1}\cdot \cost(p,c)  \\
\nonumber
&\leq & |G_{far, \calS}\cap C| \cdot (1+\varepsilon)\cdot  \cost(c,\calS)  \\
\label{eq:outerfar4}
& & +  \left(\frac{2z+\varepsilon}{\varepsilon}\right)^{z-1}  \cdot \cost(G_{far, \calS}\cap C,\greedy) \\
\nonumber
&\leq & \frac{1+\eps}{1-\eps/z} \cdot  \left(\frac{\varepsilon}{z}\right)^{2} \cdot |C_{close}|  \cost(c,\calS)  +  \left(\frac{2z+\varepsilon}{\varepsilon}\right)^{z-1}   \cost(G_{far, \calS}\cap C,\greedy) \\
\nonumber
(Eq.~\ref{eq:outerfar2})&\leq & \frac{(1+\varepsilon)^2}{(1-\varepsilon)^2}\cdot \left(\frac{\varepsilon}{z}\right)^{2}\cdot \cost(C,\calS)  + \left(\frac{2z+\varepsilon}{\varepsilon}\right)^{z-1}  \cdot \cost(G_{far, \calS}\cap C,\greedy) \\
\nonumber
(Eq.~\ref{eq:outerfar3})&\leq & \frac{(1+\varepsilon)^2}{(1-\varepsilon)^2}\cdot \left(\frac{\varepsilon}{z}\right)^{2}\cdot \cost(C,\calS)  \\
& & +  \left(\frac{2z+\varepsilon}{\varepsilon}\right)^{z-1}  \cdot   \frac{1}{3^z}\cdot \left(\frac{\varepsilon}{z}\right)^{2z-1}  \cdot \cost(G_{far, \calS}\cap C,\calS) \\
\label{eq:outerfar5}
&\leq & \frac{\varepsilon}{z\log z/\varepsilon} \cdot \cost(C,\calS)  
\end{eqnarray}

Summing this up over all clusters $C$, we therefore have

\begin{equation}
\cost(G_{far, \calS},\calS) \leq \frac{\eps}{z \log z/\eps}\cdot \cost(\calS)
\end{equation}

What is left to show is that, in the coreset, the weighted cost of the points in $G_{far, \calS} \cap \coreset$ can be bounded similarly.
For that, we use event $\calE_{far}$ to show that $\sum_{p\in G_{far, \calS}\cap C\cap \coreset} \frac{\cost(G,\greedy_0)}{\cost(p,\greedy_0)} \approx |G_{far, \calS}\cap C|$ 

In particular, event $\calE_{far}$ implies that with probability $1-k'\cdot \exp\left(-O(1)\cdot \frac{\varepsilon^2}{k'}\cdot \delta\right)$ for all clusters $C$ induced by $\greedy$

\begin{eqnarray}
\nonumber
\sum_{p\in C\cap G\cap \coreset} \frac{\cost(G,\greedy)}{\delta \cdot \cost(p,\greedy)} \cdot \left(\frac{2z}{\varepsilon}\right)^{2z} \cdot \frac{\cost(C,\greedy)}{|C|} &\leq & \sum_{p\in C\cap G\cap \coreset} \frac{\cost(G,\greedy)}{\delta\cdot\cost(p,\greedy)} \cost(p,\greedy) \\
\nonumber
& \leq &(1+\varepsilon) \cdot \cost(C\cap G,\greedy) 
\end{eqnarray}
\begin{eqnarray}
\Rightarrow \sum_{p\in C\cap G\cap\coreset} \frac{\cost(G,\greedy)}{\delta\cdot\cost(p,\greedy)} &\leq & (1+\varepsilon)\cdot \left(\frac{\varepsilon}{2z}\right)^{2z} \cdot |C| \frac{\cost(C\cap G,\greedy)}{\cost(C,\greedy)} \\
\label{eq:outerfar7}
& \leq  & (1+\varepsilon)\cdot \left(\frac{\varepsilon}{2z}\right)^{2z} \cdot |C|
\end{eqnarray}

Therefore, we have

\begin{eqnarray}
\nonumber
& &\cost(G_{far, \calS}\cap \coreset\cap C,\calS) \\
\nonumber
& =& \sum_{p\in G_{far, \calS}\cap C} \frac{\cost(G,\greedy)}{\delta\cdot\cost(p,\greedy)}\cdot \cost(p,\calS) \\
\nonumber
(Lem.~\ref{lem:weaktri}) &\leq & \sum_{p\in G_{far, \calS}\cap \coreset\cap C } \frac{\cost(G,\greedy)}{\delta\cdot\cost(p,\greedy)}\cdot \left(\left(1+\frac{\eps}{2z}\right)^{z-1} \cost(c,\calS) \right.\\
\nonumber
& & \left. \phantom{xxxxxxxxxxxxxxxxxxxxx}+ \left(\frac{2z+\varepsilon}{\varepsilon}\right)^{z-1}\cdot \cost(p,c)\right)  \\
\nonumber
&\leq & (1+\varepsilon)\cdot  \cost(c,\calS) \cdot \sum_{p\in G_{far, \calS}\cap \coreset \cap C} \frac{\cost(G,\greedy)}{\delta\cdot\cost(p,\greedy)} \\
\nonumber
(\calE_{far})& &+  \left(\frac{2z+\varepsilon}{\varepsilon}\right)^{z-1}  \cdot (1+\varepsilon) \cdot \cost(C\cap G,\greedy) \\
\nonumber
(Eq.~\ref{eq:outerfar7})&\leq & (1+\varepsilon)^2  \cost(c,\calS) \cdot \left(\frac{\varepsilon}{2z}\right)^{2z} \cdot |C| +  \left(\frac{2z+\varepsilon}{\varepsilon}\right)^{z-1}  \cdot (1+\varepsilon) \cdot \cost(C\cap G,\greedy) \\
\label{eq:outerfar8}
&\leq & (1+\varepsilon)^2\cdot \left(\frac{\varepsilon}{2z}\right)^{2z} \cdot |C| \cdot \cost(c,\calS)  +  \left(\frac{2z+\varepsilon}{\varepsilon}\right)^{z-1}  \cdot \cost(C,\greedy) \\
\nonumber
&\leq & \frac{\varepsilon}{z\log z/\varepsilon}\cdot \cost(C,\calS)  
\end{eqnarray}

where the steps following Equation~\ref{eq:outerfar8} are identical to those used to derive Equation~\ref{eq:outerfar5} from Equation~\ref{eq:outerfar4}.
Again, summing over all clusters now yields
$$ \cost(G_{far, \calS} \cap \coreset,\calS) \leq \frac{\eps}{z \log z/\eps}\cdot \cost(\calS),$$
which yields the claim.
\end{proof}

\paragraph{Combining far and close to show \cref{lem:coreset-outer}}

The overall proof follows from those lemmas.
\begin{proof}[Proof of \cref{lem:coreset-outer}]
First, we condition on event $\calE_{far}$, and on the success of \cref{lem:outerclose} for all solution in $\cand^k$. This happens with probability 
\[1 - k\exp\left(\frac{\eps^2}{k}\cdot \delta\right) - \exp\left(k\log|\cand| - 2^{-O(z)} \left(\frac{\varepsilon}{\log 1/\varepsilon}\right)^2 \delta\right).\]

Let $\calS$ be a solution, and $\tilde \calS$ its corresponding solution in $\cand^k$. We break the cost of $\calS$ into two parts: points with $\cost(p, tilde \calS) \leq 5^z \cdot \cost(p, \greedy)$, on which we can apply \cref{lem:outerclose}, on the others, on which we will apply \cref{lem:outerfar}.

From \cref{lem:outerclose}, we directly get
\[|\cost(G_{close, \tilde \calS}, \tilde \calS) - \cost(\coreset \cap G_{close, \tilde \calS}, \tilde \calS)| \leq \frac{\varepsilon}{z\log z/\varepsilon}\cdot \left(\cost(G,\greedy) + \cost(G, \tilde \calS)\right).\]

Since any point in $G_{close, \tilde \calS}$ verifies $|\cost(p, \calS) - \cost(p, \tilde \calS)| \leq \frac{\varepsilon}{z\log z/\varepsilon}(\cost(p, \calS) + \cost(p, \greedy))$ we can relate this to $\cost(G_{close, \tilde \calS}, \calS)$ as follows. First, this implies $\cost(G_{close, \tilde \calS}, \tilde \calS) \leq (1+\eps)\cost(G_{close, \tilde \calS}, \calS) + \eps \cost(G, \greedy)$. Hence:{\setlength{\emergencystretch}{2.5em}\par}

\begin{align*}
&|\cost(G_{close, \tilde \calS}, \calS) - \cost(\coreset \cap G_{close, \tilde \calS}, \calS)| \\
\leq~&
|\cost(G_{close, \tilde \calS}, \tilde \calS) - \cost(\coreset \cap G_{close, \tilde \calS}, \tilde \calS)| \\
& \qquad + \frac{\varepsilon}{z\log z/\varepsilon} (\cost(G_{close, \tilde \calS}, \calS) + \cost(G_{close, \tilde \calS}, \greedy))\\
&\qquad + \frac{\varepsilon}{z\log z/\varepsilon}(\cost(G_{close, \tilde \calS} \cap \coreset, \calS) + \cost(G_{close, \tilde \calS} \cap \coreset, \greedy))\\
\leq ~& \frac{O(\varepsilon)}{z\log z/\varepsilon}\left(\cost(G, \greedy) + \cost(G, \calS)\right)\\
&\qquad + \frac{\varepsilon}{z\log z/\varepsilon}(\cost(G_{close, \tilde \calS} \cap \coreset, \calS) + \cost(G_{close, \tilde \calS} \cap \coreset, \greedy))\\
\end{align*}

We now deal with the other far points. For this, note that $G \setminus G_{close, \tilde \calS} \subseteq G_{far, \calS}$. Indeed, any point $p \in G \setminus G_{far, \calS}$ has its cost preserved by $\tilde \calS$, and therefore verifies 
\begin{align*}
\cost(p, \tilde \calS) &\leq (1+\frac{\varepsilon}{z\log z/\varepsilon})\cost(p, \calS) + \frac{\varepsilon}{z\log z/\varepsilon} \cost(p, \greedy)\\
&\leq (1+\eps)\cdot 4^z \cost(p,\greedy) + \eps \cost(p, \greedy) \leq 5^z \cost(p, \greedy).
\end{align*}

Consequently, $G \setminus G_{far, \calS} \subseteq G_{close, \tilde \calS}$, which implies $G \setminus G_{close, \tilde \calS} \subseteq G_{far, \calS}$. Hence, we can use \cref{lem:outerfar}:
\begin{align*}
&|\cost(G\setminus G_{close, \tilde \calS}, \calS) - \cost(\coreset \cap (G \setminus G_{close, \tilde \calS}), \calS)| \\
\leq~& \cost(G\setminus G_{close, \tilde \calS}, \calS) + \cost(\coreset \cap (G \setminus G_{close, \tilde \calS}), \calS)\\
\leq ~ &\cost(G_{far, \calS}, \calS) + \cost(\coreset \cap G_{far, \calS}, \calS)\\
\leq ~&\frac{\eps}{z \log z/\eps} \cdot \cost(\calS).
\end{align*}

Hence, adding the two inequalities gives that
\begin{align*}
&|\cost(G, \calS) - \cost(\coreset \cap G, \calS)| \\
\leq & \frac{O(\eps)}{z \log z/\eps} \cdot \cost(\calS) +\frac{\varepsilon}{z\log z/\varepsilon}(\cost(G \cap \coreset, \calS) + \cost(G \cap \coreset, \greedy)).
\end{align*}

To remove the terms depending on $\coreset$ from the right hand side, one can proceed as in the end of  \cref{lem:coreset-reduc}, applying grossly the previous inequality to get $\cost(G \cap \coreset, \calS) = O(1) \cost(\calS)$ and $\cost(G \cap \coreset, \greedy) = O(1) \cost(\greedy)$. This concludes the theorem:
\begin{align*}
|\cost(G, \calS) - \cost(\coreset \cap G, \calS)| 
\leq  \frac{O(\eps)}{z \log z/\eps} \cdot (\cost(G, \calS) + \cost(G, \greedy)).
\end{align*}
\end{proof}

This concludes the coreset construction for the outer groups.

\section{Partitioning into Well Structured Groups}\label{sec:preprocess}

In this section, we show that the outcome of the partitioning step satisfies  \cref{lem:preprocess}, that we restate for convenience.

\preprocess*

Recall that the \textit{inner ring} $\inner(C)$ (resp. \textit{outer ring} $\out(C)$) of a cluster $C$ consists of the points of $C$ with cost at most $\left(\nicefrac \eps z\right)^{2z} \Delta_C$ (resp. at least $\left(\nicefrac z \eps \right)^{2z} \Delta_C$). The \textit{main ring} $\main(C)$ consist of all the other points of $C$.

Recall also that $\discarded$ contains all points that are either in some inner ring, in some group $G_{j, min}$ or in $\outergroup{}_{min}$.
$\structured$ contains center of $\greedy$ weighted by the number of points from $\discarded$ in their clusters.

To prove \cref{lem:preprocess}, we treat separately the inner ring and the groups $G_{j, min}$ and $\outergroup{}_{min}$ in the next two lemmas. Their proof are deferred to  next sections. For all those lemmas, we fix a metric space $I$ a set of clients $P$, two positive integers $k$ and $z$, and $\eps \in \mathbb{R}^*_+$. We also fix $\greedy$, a solution to $(k, z)$-clustering on $P$ with cost $\cost(\greedy) \leq c_\greedy \cost(\opt)$.

\begin{restatable}[]{lemma}{preprocessinner}
\label{lem:preprocess-inner}
For any solution $\calS$ and any cluster $C$ with center $c$ of $\greedy$, 
\begin{eqnarray*}
& &\left|\cost(\inner(C) ,\calS) -  |\inner(C)|\cdot \cost(c,\calS)\right| \leq \eps(\cost(C, \greedy) + \cost(\inner(C), \calS)) .
\end{eqnarray*}

\end{restatable}

%

%
%
%
\begin{restatable}[]{lemma}{ksmallgroups}
\label{lem:ksmallgroups}
For any solution $\calS$ and any $j$, 
\begin{eqnarray*}
& &\left|\cost(G_{j, min} ,\calS) - \sum_{i=1}^{k} |C_i\cap G_{j,min}|\cdot \cost(c_i,\calS)\right| \leq  \varepsilon\cdot \cost(R_j,\calS)+ \varepsilon \cdot \cost(R_j,\greedy) .
\end{eqnarray*}

Moreover, for any solution $\calS$,
\begin{eqnarray*}
& &\left|\cost(\outergroup{\greedy, \seeded}_{min} ,\calS) - \sum_{i=1}^{k} |C_i\cap \outergroup{\greedy, \seeded}_{min}|\cdot \cost(c_i,\calS)\right| \leq  \varepsilon\cdot \cost(\calS)+ \varepsilon \cdot \cost(\greedy) .
\end{eqnarray*}

\end{restatable}

The proof of Lemma \ref{lem:preprocess} combines those lemmas.
\begin{proof}[Proof of Lemma \ref{lem:preprocess}]
We decompose $|\cost(\discarded, \calS) -\cost(\structured, \calS)|$ into terms corresponding to the previous lemmas:

\begin{align*}
 |\cost(\discarded, \calS) - \cost(\structured, \calS)| 
&\leq 
\sum_{i=1}^k \left\vert \cost(\inner(C_i), \calS) -  |\inner(C_i)| \cost(c_i, \calS)\right\vert \\
&\quad + \sum_{j=2z\log(\eps/z)}^{2z\log(z/\eps)} \left|\cost(G_{j, min} ,\calS) - \sum_{i=1}^k |C_i\cap G_{j, min}|\cost(c_i,\calS)\right|\\
&\quad + \left|\cost(\outergroup{}_{min} ,\calS) -\sum_{i=1}^k |C_i\cap \outergroup{}_{min}|\cost(c_i,\calS)\right|\\
&\leq \sum_{i=1}^k \eps(\cost(C_i, \greedy) + \cost(\inner(C_i), \calS))\\ 
& \quad +2\eps \cost(\calS) + 2\eps \cost(\greedy) + \eps(\cost(\calS) + \cost(\greedy))\\ 
&\leq 8\eps c_{\greedy} \cost(\calS), 
%
\end{align*}
where the second inequality uses Lemmas~\ref{lem:preprocess-inner} and ~\ref{lem:ksmallgroups}.

\end{proof}

\subsection{The Inner Ring: Proof of Lemma \ref{lem:preprocess-inner}}
\preprocessinner*
%
\begin{proof}
Let $C$ be a cluster induced by $\greedy$, and $p$ be a point in the inner ring $\inner(C)$.
We start by bounding $|\cost(p, \calS) - \cost(c, \calS)|$. Let $\calS(p)$ (resp. $\calS(c)$) be the closest point from $\calS$ to $p$ (resp. $c$).

%

Using \cref{lem:weaktri}, we get 
\[|\cost(p, \calS) - \cost(c, \calS)| \leq \eps \cdot \cost(p, \calS) +(1+2z/\eps)^{z-1} \cdot \cost(c, p).\]
Since $p$ is from the inner ring of its cluster, $\cost(c, p) \leq \left(\frac{\eps}{z}\right)^{2z} \Delta_C$, hence $(1+2z/\eps)^{z-1} \cost(c, p) \leq (2+\eps)^{z-1} \cdot (\eps/z)^{z+1} \cdot \Delta_C \leq \eps \Delta_C$, for small enough $\eps$.

 Summing this over all points of the inner ring yields
 \begin{align*}
\left\vert\cost(\inner(C), \calS)-  |\inner(C)|\cdot \cost(c,\calS)\right| 
&\leq \sum_{p\in \inner(C)} |\cost(p, \calS) - \cost(c, \calS)|\\
&\leq  \sum_{p\in \inner(C)} \eps \cost(p, \calS) + \eps \Delta_C \\
&\leq \eps \cost(\inner(C), \calS) + \eps |\inner(C)| \Delta_C\\
&\leq  \eps \cost(\inner(C), \calS) +\eps \cost(C, \greedy)
\end{align*}

This implies 
\[\left\vert\cost(\inner(C), \calS)-  |\inner(C)|\cdot \cost(c,\calS)\right| \leq \eps(\cost(C, \greedy) + \cost(\inner(C), \calS)).\]
\end{proof}

\subsection{The Cheap Groups: Proof of \cref{lem:ksmallgroups}}
\ksmallgroups*
\begin{proof}
Using Lemma~\ref{lem:weaktri}, for a point $p$ in cluster $C_i$
\[\left\vert \cost(c_i,\calS) - \cost(p, \calS)\right\vert \leq \eps \cost(p,\calS) + \left(1+\frac{2z}{\varepsilon}\right)^{z-1} \cost(p,c_i).
\]

Let $G$ be a group, either $G_{j, min}$ or $\outergroup{\greedy, \seeded}_{min}$. 
Summing for all cluster $C_i$ and all $p\in G \cap C_i$, 
we now get
\begin{eqnarray*}
& &\left\vert\sum_{i=1}^{k} |C_i\cap G|\cdot \cost(c_i,\calS) - \cost(G,\calS)\right\vert\\
 &\leq & \varepsilon\cdot \cost(G,\calS) +\sum_{i=1}^{k} \sum_{p\in G\cap C_i}\left(1+\frac{2z}{\varepsilon}\right)^{z-1} \cost(p,\greedy) \\
    &\leq & \varepsilon\cdot \cost(G,\calS) +\sum_{i=1}^{k}  \left(\frac{3z}{\varepsilon}\right)^{z-1}\cost(C_i\cap G,\greedy) \\
&\leq & \varepsilon\cdot \cost(G,\calS) + \left(\frac{3z}{\varepsilon}\right)^{z-1}\cost(G,\greedy)  \\
\end{eqnarray*}
Now, either $G = G_{j, min}$ for some $j$, and $\cost(G, \greedy) \leq \left(\frac{\varepsilon}{4z}\right)^{z} \cdot \cost(R_j,\greedy)$; or $G = \outergroup{\greedy, \seeded}_{min}$, and $\cost(G, \greedy) \leq \left(\frac{\varepsilon}{4z}\right)^{z} \cdot \cost(\out(\greedy), \greedy)\leq \left(\frac{\varepsilon}{4z}\right)^{z}\cdot \cost(\greedy)$. 

In both cases, the lemma follows.
\end{proof}

\section{Application of the Framework: New Coreset Bounds for Various Metric Spaces}\label{sec:centroid}

In this section, we apply the coreset framework to specifics metric spaces. For each of them, we show the existence of a small approximate centroid set, and apply \cref{thm:main} to prove the existence of small coresets.

We recall the definition of a centroid set (\cref{def:centroid-set}): given an instance of $(k, z)$-clustering and a set of centers $\greedy$, an $\greedy$-approximate centroid set $\cand$ is a set that satisfies the following: for every solution $\calS$, there exists $\tilde \calS \in \cand^k$ such that for all points $p$ that verifies $\cost(p, \calS) \leq \left(\frac{8z}{\eps}\right)^z \cost(p, \greedy)$ or $\cost(p, \tilde \calS) \leq \left(\frac{8z}{\eps}\right)^z \cost(p, \greedy)$, it holds $|\cost(p, \calS) - \cost(p, \tilde{\calS})| \leq \frac{\eps}{z\log(z/\eps)}\left(\cost(p, \calS) + \cost(p, \greedy)\right)$.

\cref{thm:main} states that in case there is an $\greedy$-approximate centroid set $\cand$, then there is a linear-time algorithm that constructs with probability $1-\pi$ a coreset of size 
\[O\left(\frac{2^{O(z\log z)}\cdot\log^4 1/\eps}{\min(\eps^2, \eps^z)}\left(k \log |\cand| + \log \log (1/\eps) + \log(1/\pi)\right)\right)\]

\subsection{Structural Property on Solutions}
We also show a structural property on solutions, that we will use in order to show the existence of small approximate centroid sets. Essentially, when replacing a center $s$ by a center in $\cand$ we will make an error $\eps \cost(q, \greedy)$ for some $q$ that we can choose: it is necessary to ensure this error is tiny compared to any $\cost(p, s) + \cost(p, \greedy)$.

Given a point $q$ and a center $s$, we say that a point $p$ is \emph{problematic} with respect to $q$ and $s$ when $\dist(p, \greedy) + \dist(p, s) \leq \frac{\eps^2}{8z^2} (\dist(q, \greedy) + \dist(q, s))$. In that case, we cannot bound the error $\dist(q, \greedy) + \dist(q, s)$ by some quantity depending on $\cost(p, s) + \cost(p, \greedy)$. However, we show the following:
\begin{lemma}\label{lem:noproblem}
Let $\calS$ be a solution, such that any input point $p$ verifies $\dist(p, \calS) \leq \frac{8z}{\eps}\cdot \dist(p, \greedy)$. There exists a solution $\calS ' \subseteq \calS$ such that 
\begin{itemize}
\item for all $p$, it holds that $|\cost(p, \calS) - \cost(p, \calS')| \leq \frac{\eps}{z \log z/\eps}(\cost(p, \calS) + \cost(p, \greedy))$, and
\item for any center $s \in \calS'$, let $q = \argmin_{p : \dist(p, s) \leq \frac{10z}{\eps}\dist(p, \greedy)} \dist(p, \greedy) + \dist(p, s)$. There is no problematic point with respect to $q$ and $s$.
\end{itemize}
\end{lemma}
\begin{proof}
First, we show that in case there is a problematic point $p$ with respect to some $s$ and $q$, then we can serve the whole cluster of $s$ by $\calS(p)$, the point that serves $p$ in $\calS$. 
We work in this proof with particular solutions, where points are not necessarily assigned to their closest center. This simplifies the proof, but needs particular care at some moments. In particular, we will ensure that $\dist(p, \calS(p)) \leq \frac{10z}{\eps}\cdot \dist(p, \greedy)$ is always verified.
We will then remove inductively centers with problematic points to construct $\calS'$. 

\textbf{Removing a center that has a problematic point.} 
\\ Let $s \in \calS$, and $q = \argmin_{p : \dist(p, s) \leq \frac{10z}{\eps}\dist(p, \greedy)} \dist(p, \greedy) + \dist(p, s)$ as in the statement. Let $p$ be a problematic point with respect $s$ and $q$, and $\calS(p)$ its the center serving $p$ in $\calS$. First, note that since $p$ is problematic, it must be that  $\dist(p, \calS(p)) \leq \dist(s, p)$: otherwise, $p$ would verify $\dist(p, s) \leq \frac{10z}{\eps}\dist(p, \greedy)$, and the minimality of $q$ would ensure that $p$ is not problematic.
 Thus, it holds that:
\begin{align*}
\dist(s, \calS(p)) &\leq \dist(s, p) + \dist(p, \calS(p)) \leq 2\dist(s,p)\\
&\leq 2(\dist(s, p) + \dist(p, \greedy)) \leq \frac{\eps^2}{4z^2} (\dist(q, \greedy) + \dist(q, s)).
\end{align*}

Now, let $p'$ be served by $s$. Using the triangle inequality, we immediately get
\begin{align*}
\dist(p', \calS(p)) \leq \dist(p', s) + \dist(s, \calS(p)) \leq \dist(p', s) + \frac{\eps^2}{4z^2} (\dist(q, \greedy) + \dist(q, s)).
\end{align*}

Additionally, it holds that 
\begin{align}
\notag
\min_{q' : \dist(q', \calS(p)) \leq \frac{10z}{\eps} \dist(q', \greedy)} \dist(q', \greedy) + \dist(q', \calS(p)) &\leq \dist(p, \calS(p)) + \dist(p, \greedy)\\
\notag
&\leq \dist(s, p) + \dist(p, \greedy)\\
\label{eq:noAccumulation}
&\leq \frac{\eps^2}{8z^2} (\dist(q, \greedy) + \dist(q, s))
\end{align}
Hence, if $\calS(p)$ is removed as well, the error for points served by $s$ will be an $\frac{\eps^2}{8z^2}$-fraction of the initial error. This implies that the total error will not accumulate, as we will now see.

\textbf{Constructing $\calS'$.} To construct $\calS'$, we proceed iteratively: start with $\calS' = \calS$, and as long as there exists a center $s$ that have a problematic point $p$ with respect to it, remove $s$ and reassign the whole cluster of $s$ to $\calS'(p)$, the closest point to $p$ in the current solution. This process must end, as there is no problematic point when there is a single center.

For a point $p$, let $s_1,..., s_j$ be the successive cluster it is reassigned to, with corresponding $q_1,..., q_j$. Using \cref{eq:noAccumulation}, it holds that 
$\dist(q_{i+1}, \greedy) + \dist(q_{i+1}, s_{i+1}) \leq \frac{\eps^2}{8z^2} (\dist(q_{i}, \greedy) + \dist(q_{i}, s_{i}))$. Hence, the distance increase for $p$ is geometric: using that $\dist(q_{1}, \greedy) + \dist(q_{1}, s_{1}) \leq \dist(p, \greedy) + \dist(p, \calS)$ (which holds by minimality of $q_1$), we get that at any given step $i$ it holds that 
\begin{align*}
\dist(p, s_i) &\leq \dist(p, \calS) + (\dist(p, \greedy) + \dist(p, \calS)) \sum_{j=1}^i  \left(\frac{\eps^2}{8z^2}\right)^i\\
&\leq \dist(p, \calS) + \frac{\eps^2}{4z^2}(\dist(p, \greedy) + \dist(p, \calS))\\
&\leq \frac{8z}{\eps}\cdot \dist(p, \greedy) + \frac{\eps^2}{4z^2}\cdot(1+\frac{8z}{\eps})\dist(p, \greedy)\\
&\leq \frac{10z}{\eps}\cdot \dist(p, \greedy),
\end{align*}
as promised in order to remove centers.

Last, we show that the first bullet of the lemma holds. We consider now the standard assignment:  $p$ is assigned to its closest center of $\calS'$, instead of $s_j$. First, since we only removed centers, it holds that $\cost(p, \calS) \leq \cost(p, \calS')$.
Second, using \cref{lem:weaktri}, we have for any $ \eps' = \frac{\eps}{z \log z/\eps}$:
\begin{align*}
\cost(p, \calS') &\leq \cost(p, s_j) \\
&\leq (1+ \eps') \cost(p, \calS) + \left(\frac{4z}{\eps'}\right)^{z-1} \cdot \left(\frac{\eps^2}{4z^2}\right)^z \cdot (\dist(p, \greedy) + \dist(p, \calS))^z \\
&\leq (1+\eps')\cost(p, \calS) + \left(\frac{\eps}{z}\right)^z \cdot (2\log z/\eps)^{z-1}\cdot (\cost(p, \greedy) + \cost(p, \calS))\\
&\leq (1+\eps')\cost(p, \calS) + \eps'(\cost(p, \greedy) + \cost(p, \calS))
\end{align*}

Hence, we conclude that 
\[\left|\cost(p, \calS') - \cost(p, \calS)\right| \leq  \frac{\eps}{z \log z/\eps}\left(\cost(p, \calS) + \cost(p, \greedy)\right),\]
which concludes the lemma.
\end{proof}

\subsection{In Metrics with Bounded Doubling Dimension}

We start by defining the Doubling Dimension of a metric space, and stating a key lemma.

Consider a metric space $(X,\dist)$.
For a point $p \in X$ and an integer $r \ge 0$, we let
$\beta(p,r) = \{x\in X \mid \dist(p,x) \le r \}$ be the
\emph{ball} around $p$ with radius $r$. 

\begin{definition}
The \emph{doubling dimension} of a metric is the smallest integer $d$ such that any
ball of radius $2r$ can be covered by $2^d$ balls of radius~$r$. 
\end{definition}

Notably, the Euclidean space $\mathbb{R}^d$ has doubling dimension $\theta(d)$.

A $\gamma$-\emph{net} of $V$ is a set of points $X\subseteq V$ such that for all
$v \in V$ there is an $x \in X$ such that $\dist(v, x) \leq \gamma$, and 
for all $x, y \in X$ we have $\dist(x, y) > \gamma$. A net is therefore a set 
of points not too close to each other, such that
every point of the metric is close to a net point. 
The following lemma bounds the cardinality of a net in doubling metrics.\\

\begin{lemma}[from Gupta et. al \cite{GuptaKL03}]\label{prop:doub:net}
Let $(V, \dist)$ be a metric space with doubling dimension $d$ and, diameter 
$D$, and let $X$ be a $\gamma$-net of $V$. Then $|X| \leq 2^{d \cdot 
\lceil \log_2 (D/\gamma)\rceil}$.
\end{lemma}

The goal of this section is to prove the following lemma. Combined with \cref{thm:main}, it ensures the existence of small coreset in graphs with small doubling dimension.

\begin{lemma}
\label{lem:k2unionbound2}
Let $M = (X, \dist)$ be a metric space with doubling dimension $d$, let $P\subset X$, let $k$ and $z$ be positive integers and let $\varepsilon>0$. 
Further, let $\greedy$ be a $c_{\greedy}$-approximate solution with at most $k$ centers.
There exists an $\greedy$-approximate centroid set for $P$ of size
$$|P| \cdot \left(\frac{z}{\eps}\right)^{O(d)} \qedhere$$
\end{lemma}

A direct corollary of that lemma is the existence of a coreset in Doubling Metrics, as it is enough to show the mere existence of a small centroid set for applying \cref{cor:weighted}.

\begin{corollary}\label{cor:coreset-doubling}
  Let $M = (X, \dist)$ be a metric space with doubling dimension $d$, and two positive integers $k$ and $z$. 
  
  There exists an algorithm with running time $\tilde O(nk)$ that constructs an $\eps$-coreset for $(k, z)$-clustering on $P \subseteq X$ with size 
  \[O\left(\frac{\log^5 1/\eps}{2^{O(z\log z)}\min(\eps^2, \eps^z)}\left(kd + \log 1/\pi\right)\right)\]
\end{corollary}
\begin{proof}
We first compute a coreset of size $\tilde{O}(k^3d \varepsilon^{-2})$ \cite{HuangJLW18}. 
Then, combining \cref{thm:main} and \cref{lem:k2unionbound2} yields an algorithm constructing a coreset of size
$$O\left(\frac{\log^4 1/\eps}{2^{O(z\log z)}\min(\eps^2, \eps^z)}\left(kd \log 1/\eps + k \log kd/\varepsilon + \log 1/\pi\right)\right).$$
If $\log k > d$ then $O(\log kd) = O(\log k)$. If $d> \log k$ then $O(kd + k\log kd) = O(kd)$, hence the claimed bound follows. 
\end{proof}

\begin{proof}[Proof of \cref{lem:k2unionbound2}]
For each point $p\in P$, let $c$ be the center to which $p$ was assigned in $\greedy$. Let $B\left(p,\left(\frac{8z}{\varepsilon}\right)\dist(p,c)\right)$ be the metric ball centered around $p$ with radius $\left(\frac{8z}{\varepsilon}\right)\cdot \dist(p,c)$, and let $N_{p}$ be an $\left(\frac{\varepsilon}{4z}\right)\cdot \dist(p,\greedy)$-net of that ball. 

Due to Lemma~\ref{prop:doub:net}, $N_p$ has size $(\varepsilon/z)^{-O(d)}$. Additionally, let $s_f$ be a point not in any $B(p,\left(\frac{10z}{\varepsilon}\right)\dist(p,\greedy))$, if such a point exist.

Let $\mathcal{N}:=s_f \bigcup_{p\in Y} N_p$.
We claim that $\mathcal{N}$ is the desired approximate centroid set.

For a candidate solution $\calS$, apply first \cref{lem:noproblem}, so that we can assume that for any center $s \in \calS$, and $q = \argmin_{p : \dist(p, s) \leq \frac{10z}{\eps}\dist(p, \greedy)} \dist(p, \greedy) + \dist(p, s)$, there is no problematic point with respect to $q$ and $s$.

let $\tilde \calS$ be the solution obtained by replacing every center $s \in \calS$ by $\tilde s \in \cand$ as follows: let $q = \argmin_{p : \dist(p, s) \leq \frac{10z}{\eps}\dist(p, \greedy)} \dist(p, \greedy) + \dist(p, s)$. Pick $\tilde s$ to be the closest point to $s$ in $N_{q}$. If such a $q$ does not exist, pick $\tilde s = s_f$.

Now, let $p$ be a point such that $\cost(p,\calS) \leq \left(\frac{8z}{\varepsilon}\right)^z\cdot \cost(p,\greedy)$, let $s$ be any center in $\calS$ and $q$ defined as previously.
Then, by construction of the $\tilde \calS$, there is a center $\tilde s$ with $\dist(s, \tilde s) \leq \left(\frac{\varepsilon}{4z}\right) \dist(q,\greedy)$ and therefore, using that $p$ is no problematic:
\begin{align}
\notag
\cost(p, \tilde \calS) &\leq \cost(p, \tilde s) \leq (1+\eps)\cost(p, s) + (1+z/\eps)^{z-1}\cost(s, \tilde s)\\
\notag
&\leq (1+\eps) \cost(p, \calS) + (2z/\eps)^{z-1}\left(\frac{\varepsilon}{2z}\right)^z \cost(q,\greedy)\\
&\leq (1+\eps) \cost(p, \calS) + \eps \cost(q, \greedy)\\
\label{eq:s-to-tildes}
&\leq (1+\eps) \cost(p, \calS) + \eps \cost(p, \greedy).
\end{align}

To show the other direction, for any point in $\tilde \calS$ there is a center $s$ with $\dist(s, \tilde s) \leq \left(\frac{\varepsilon}{4z}\right) \dist(q,\greedy)$. Hence the previous equations apply as well, and we can conclude:
for a point $p$ such that $\cost(p,\calS) \leq \left(\frac{8z}{\varepsilon}\right)^z\cdot \cost(p,\greedy)$, 
\[|\cost(p, \calS) - \cost(p, \tilde \calS)|\leq \eps(\cost(p, \calS) + \cost(p, \greedy)).\]

Rescaling $\eps$ concludes the lemma: there is an $\greedy$-approximate centroid set with size 
$|P|\left(\frac{z^2\log z/\eps}{\eps}\right)^{O(d)} = |P|\left(\frac{z}{\eps}\right)^{O(d)}$.

\end{proof}

\section{Graphs with Bounded Treewidth} \label{sec:centroid-treewidth}
In this section, we show that for graphs with treewidth $t$, there exists a small approximate centroid set. Hence, the main framework provides an algorithm computing a small coreset. We first define the treewidth of a graph:

\begin{definition}
A tree decomposition of a graph $G = (V, E)$ is a tree $\calT$ where each node $b$ (call a \textit{bag}) is a subset of $V$ and the following conditions hold:
\begin{itemize}
\item The union of bags is $V$,
\item $\forall v \in V$, the nodes containing $v$ in $\calT$ form a connected subtree of $\calT$, and
\item for all edge $(u, v) \in E$, there is one bag containing $u$ and $v$.
\end{itemize}

The treewidth of a graph $G$ is the smallest integer $t$ such that their exists a tree decomposition with maximum size bag $t+1$.
\end{definition}

\begin{lemma}\label{lem:centroid-treewidth}
Let $G = (V, E)$ be a graph with treewidth $t$, $X \subseteq V$ and $k, z > 0$. Furthermore, let $\greedy$ be solution to $(k, z)$-clustering for $X$. 
Then, there exists an $\greedy$-approximate centroid set for $(k,z)$-clustering on $V$ of size $\poly(|X|)\left(\frac{z^2\log z/\eps}{\eps}\right)^{O(t)}$.
\end{lemma}

Applying this lemma with $X$ yields the direct corollary:
\begin{corollary}\label{cor:coreset-treewidth}
Let $G = (V, E)$ be a graph with treewidth $t$, $X \subseteq V$, $k$ and $z > 0$.

There exists an algorithm running time $\tilde O(nk)$ that constructs an $\eps$-coreset for $(k, z)$-clustering on $X$, with size 
\[O\left(\frac{\log^5 1/\eps}{2^{O(z\log z)}\min(\eps^2, \eps^z)}\left(k \log k + kt + \log(1/\pi)\right)\right).\qedhere\]
\end{corollary}
\begin{proof}
Let $X \subseteq V$. 
We start by computing a $(k, \eps)$-coreset $X_1$ of size $O(\poly(k, 1/\eps, t))$, using the algorithm from \cite{baker2020coresets}

We now apply our framework to $X_1$. Computing an approximation on $X_1$ takes time $\tilde O(|X_1|k)$, using the algorithm from Mettu and Plaxton~\cite{MettuP04}.

\cref{lem:centroid-treewidth} ensure the existence of an approximate centroid set for $X_1$ with size $\poly(|X_1|)\left(\frac{z}{\eps}\right)^{O(t)}$. Hence,  \cref{cor:weighted} and the framework developed in the previous sections gives an algorithm that computes an $\eps$-coreset of $X$ with size 
\[O\left(\frac{\log^4{1/\eps}}{2^{O(z\log z)}\min(\eps^2, \eps^z)}\left(k \log |X_1| + kt\log 1/\eps + \log(1/\pi)\right)\right).\]

Using that $|X_1| = O(\poly(k, \eps, t))$ yields a coreset of size 
\[O\left(\frac{\log^5 1/\eps}{2^{O(z\log z)}\min(\eps^2, \eps^z)}\left(k \log k + kt + \log(1/\pi)\right)\right).\]

Instead of using \cite{baker2020coresets}, one could apply our algorithm repeatedly as in Theorem 3.1 of \cite{BravermanJKW21}, to reduce iteratively the number of distinct point consider and to eventually get the same coreset size. 
%
The number of repetition needed to achieve that size bound is $O(\log^* n)$, where $\log^*(x)$ is the number of times $\log$ is applied to $x$ before the result is at most $1$; formally $\log^*(x) = 0$ for $x \leq 1$, and $\log^*(x) = 1 + \log^* \log x$ for $x > 1$. The complexity of this repetition is therefore $\tilde O(nk)$, and the success probability $1 - \pi$, as proven in \cite{BravermanJKW21}.
\end{proof}

For the proof of \cref{lem:centroid-treewidth}, we rely on the following structural lemma:\footnote{In the statement of \cite{baker2020coresets}, the third item is slightly different. To recover our statement from theirs, take $P_A = A$ when $|A| = O(t)$.}.

\begin{lemma}[Lemma 3.7 of \cite{baker2020coresets}]
Given a graph $G = (V, E)$ of treewidth $t$, and $X \subseteq V$, there exists a collection $\calT$ of subsets of $V$ such that:
\begin{enumerate}
\item $\cup_{A\in \calT} A = V$,
\item $|\calT| = \poly(|X|)$,
\item For each $A \in \calT$,  $|A\cap X| = O(t)$, and there exists $P_A \subseteq V$ with $|P_A| = O(t)$ 
such that there is no edge between $A\setminus P_A$ and $V \setminus (A \cup P_A)$.
\end{enumerate}
\end{lemma}

Our construction relies on the following simple observation. Let $s$ be a possible center, and $p$ be a vertex such that $\cost(p, s) \leq \left(\frac{4z}{\eps}\right)^z \cost(p, \greedy)$. Let $A \in \calT$ such that $p \in A$. Then, either $s \in A$, or the path connecting $p$ to $s$ has to go through $P_A$. 

We use this observation as follows: it would be enough to replace a center $s$ from solution $\calS$ by one that has approximately the same distance of all points of $P_A$.  The main question is : how should we round the distances to $P_A$? The goal is to classify the potential centers into few classes, such that taking one representative per class gives an approximate centroid set. The previous observation indicates that classifying the centers according to their distances to points of $P_A$ is enough. However, there are too many different classes: instead, we round those distances.

Ideally, this rounding would ensure that for any point $p$ and any center $s$, all centers in $s$'s class have same distance to $p$, up to an additive error $\eps(\cost(p, s) + \cost(p, \greedy))$. This would mean rounding the distance from $s$ to any point in $P_A$ by that amount -- for instance, rounding to the closest multiple of $\eps(\cost(p, s) + \cost(p, \greedy))$. Nonetheless, this way of rounding depends on each point $p$: a rounding according to $p$ may not be suited for another point $q$. To cope with that, we will quite naturally round distances according to the point $p$ that minimizes $\cost(p, s) + \cost(p, \greedy)$. Additionally, to ensure that the number of classes stays bounded, it is not enough to round to the closest multiple of $\eps(\cost(p, s) + \cost(p, \greedy))$: we also show that distances bigger than $\frac{1}{\eps}(\cost(p, s) + \cost(p, \greedy))$ can be trimmed down to $\frac{1}{\eps}(\cost(p, s) + \cost(p, \greedy))$. That way, for each point of $P_A$ there are only $1/\eps^2$ many possible rounded distances.

Hence, a class is defined by a certain point $p$, a part $A$ and by $|P_A| = t$ many rounded distances: in total, that makes $\poly(|X|) \eps^{-O(t)}$ many classes. The approximate centroid set contains one representative of each class: this would prove \cref{lem:centroid-treewidth}. We now make the argument formal, in particular to show that the error incurred by the trimming is affordable.

\begin{proof}[Proof of \cref{lem:centroid-treewidth}]
Given a point $s \in V$ and a set $A \in \calT$, we call a \textit{distance tuple to $A$} 
$\mathbf{d}_{A}(s) := \left(\dist(s, x) ~|~ \forall x \in X \cap A\right) +
 \left(\dist(s, x) ~|~ \forall x \in P_A\right)$. 
 Let $q \in X$: the rounded distance tuple of $s$  with respect to $q$ is $\widetilde {\mathbf{d}_{A, q}}(s)$ defined as follows:
\begin{enumerate}
\item For $x \in X \cap A$ or $x = q$, $\widetilde  d(s, x)$ is the multiple of 
$\frac{\eps}{z} \cdot \dist(x, \greedy)$ smaller than $\frac{10z}{\eps} \dist(x, \greedy)$ closest to $\dist(s, x)$. 
\item For $y \in P_A$, $\widetilde  d(s, y)$  is the multiple of 
$\frac{\eps^3}{8z^3}\cdot \dist(q, \greedy)$ smaller than $\frac{200z^3}{\eps^3}\dist(q, \greedy)$ closest to $\dist(s, y)$.
\end{enumerate}

Now, for every $A\in \calT$, $q\in X$ and every rounded distance tuple $T$ to $A$ with respect to $q$ such that $\exists s: T = \widetilde  {\mathbf{d}_A}(s)$, $\cand$ contains one point $s \in A$ having that rounded distance tuple.

\paragraph*{Bounding the size of $\cand$.}
Fix some $A \in \calT$, and $q \in X$. A rounded distance tuple to $A$ is made of $O(t)$ many distances. Each of them takes its value among $\poly(z/\eps)$ possible numbers, due to the rounding. Hence, there are at most $\left(\frac{z}{\eps}\right)^{O(t)}$ possible rounded distance tuple to $A$, and so at most that many points in $\cand$. Since there are $\poly(|X|)$ different choices for $A$ and $q$, the total size of $\cand$ is $\poly(|X|)\left(\frac{z}{\eps}\right)^{O(t)}$.

\paragraph*{Bounding the error.} 
We now bound the error induced by approximating a solution $\calS$ by a solution $\tilde \calS \subseteq \cand$.

First, by applying \cref{lem:noproblem}, we can assume that  for any center $s \in \calS$, and $q = \argmin_{p : \dist(p, s) \leq \frac{10z}{\eps}\dist(p, \greedy)} \dist(p, \greedy) + \dist(p, s)$, there is no problematic point with respect to $q$ and $s$.
{\setlength{\emergencystretch}{5em}\par}

Let $A \in \calT$ such that $s \in A$, and $q = \argmin_{\substack{p : \dist(p, s) \leq \frac{10z}{\eps}\dist(p, \greedy)}} \dist(p, \greedy) + \dist(p, s)$. $\tilde s$ is chosen to have the same rounded distance tuple to $A$ with respect to $q$ as $s$. $\tilde \calS$ is the solution made of all such $\tilde s$, for $s \in \calS$.

As in the proof of \cref{lem:k2unionbound2}, we first show that points close to $s$ have cost preserved in $\tilde s$. We will later show that points with large distance to $s$ have also large distance to $\tilde s$, to ensure that their distance to $\tilde \calS$ does not decrease.

Let $p \in X$ be an input point. By \cref{lem:noproblem}, $p$ is not problematic with respect to $s$ and $q$.
Note that $s$ is not necessarily the closest center to $p$. We aim at showing that $|\cost(p, s) - \cost(p, \tilde s)| \leq \eps(\cost(p, s) + \cost(p, \greedy))$.

 First, when $p \notin X \cap A$, we distinguish two subcases:
\begin{itemize}
\item either $\dist(p, s) \leq \frac{200z^3}{\eps^3}\dist(q, \greedy)$: in that case, let $x \in p_A$ that is on the shortest path between $p$ and $s$. We have $\dist(s, x) \leq \frac{200z^3}{\eps^3}\dist(q, \greedy)$, and so $s$ and $\tilde s$ have the same rounded distance to $x$. Hence,
\begin{align*}
\dist(p, \tilde s) &\leq \dist(p, x) + \dist(x, \tilde s) \leq \dist(p, x) + \dist(x, s) + \frac{\eps^3}{8z^3} \dist(q, \greedy)\\
&\leq \dist(p, s) + \frac{\eps^3}{8z^3} \cdot \left(\frac{8z^2}{\eps^2}\right)(\dist(p, \greedy) + \dist(p, s))\\
&\leq \left(1+\frac{\eps}{z}\right)\dist(p, s) + \frac{\eps}{z} \dist(p, \greedy),
\end{align*}
The first line implies that $\dist(p, \tilde s) \leq  \frac{200z^3}{\eps^3}\dist(q, \greedy)$ as well: we can therefore repeat the argument, choosing $x$ to be on the shortest path between $p$ and $\tilde s$ instead, to show that $\dist(p, s) \leq  \left(1+\frac{\eps}{z}\right)\dist(p, \tilde s) + \frac{\eps}{z} \dist(p, \greedy)$.
This implies, using \cref{lem:weaktri}, that $|\cost(p, \tilde s) - \cost(p, s)| \leq \frac{\eps}{z}\cdot (\cost(p, s) + \cost(p, \greedy))$.

\item Otherwise, $\dist(p, s) > \frac{200z^3}{\eps^3}\dist(q, \greedy)$. In that case, we can argue that $\dist(s, \tilde s)$ is negligible compared to $\dist(p, s)$. Recall that $\dist(q, s) \leq \frac{10z}{\eps}\dist(q, \greedy)$.

The rounding ensures that the distance to $q$ is preserved: $\dist(q, \tilde s) \leq \dist(q,s) +\frac{\eps}{z} \dist(q, \greedy)$. Hence, we get that 

%
\begin{align*}
\dist(s, \tilde s) &\leq 2\dist(q, s) + \frac{\eps}{z}\dist(q, \greedy)\\
&\leq \left(\frac{100z^2}{\eps^2} + \frac{\eps}{z}\right)\cdot \dist(q, \greedy)\\
&\leq \frac{200z^2}{\eps^2} \cdot \frac{\eps^3}{200z^3}\cdot \dist(p, s) \leq \frac{\eps}{z}\cdot \dist(p, s).
\end{align*}
Finally, using Lemma~\ref{lem:weaktri}, we conclude again that $|\cost(p, \tilde s) - \cost(p, s)| \leq \eps\cost(p, s) + \eps\cost(p, \greedy)$.
\end{itemize}

Now, in the other case where $p \in X \cap A$, if $\dist(p, s) \leq \frac{10z}{\eps}\dist(p, \greedy)$, then
the choice of $\tilde s$ ensures that 
 $|\dist(\tilde s, p) - \dist(s, p)| \leq \frac{\eps}{z} \dist(x, \greedy)$ 
 and therefore $|\cost(p, s) - \cost(p, \tilde s)| \leq \eps \cost(p, s) + (1+z/\eps)^{z-1}\cost(s, \tilde s) 
 \leq \eps\cost(p, \calS) + \eps \cost(p, \greedy)$. 
 In the last case when $\dist(p, s) > \frac{10z}{\eps}\dist(p, \greedy)$, then the rounding enforces $\dist(p, \tilde s) = \frac{10z}{\eps}\dist(p, \greedy)$.

Hence, in all possible cases, it holds that either $\dist(p, s)$ and $\dist(p, \tilde s)$ are bigger than $\frac{10z}{\eps}\dist(p, \greedy)$, or: 
\begin{equation}\label{eq:centroid-treewidth}
|\cost(p, \tilde s) - \cost(p, s)| \leq \eps\cost(p, s) + \eps\cost(p, \greedy).
\end{equation}

To extend that result to the full solutions $\calS$ and $\tilde \calS$ instead of a particular center, we note that since $p$ is interesting, $\dist(p, \calS) \leq \frac{8z}{\eps}\dist(p, \greedy)$. Hence, we can apply \cref{eq:centroid-treewidth} with $s$ being the closest point to $p$ in $\calS$:  
$\cost(p, \tilde \calS) \leq (1+\eps)\cost(p, \calS) + \eps\cost(p, \greedy)$.

In particular, this implies that $\dist(p, \tilde s) \leq \frac{10z}{\eps}\dist(p, \greedy)$.
Choose now $\tilde s$ to be the closest point to $p$ in $\tilde \calS$ and $s$ its corresponding center in $\calS$. Using \cref{eq:centroid-treewidth} therefore gives:
\begin{align*}
\cost(p, s) &\leq \cost(p, \tilde \calS) + \eps(\cost(p, \greedy) + \cost(p, s)) \\
\implies \cost(p, s) &\leq \frac{1}{1-\eps}\cost(p, \tilde \calS) + \frac{\eps}{1-\eps} \cost(p, \greedy)\\
&\leq (1+2\eps)\cost(p, \tilde \calS) + 2\eps \cost(p, \greedy)\\
\implies \cost(p, \calS) &\leq (1+2\eps)\cost(p, \tilde \calS) + 3\eps \cost(p, \greedy).
\end{align*}

Hence, combining those two inequalities yields
\[|\cost(p, \calS) - \cost(p, \tilde \calS)| \leq \eps\cost(p, \calS) +2\eps\cost(p, \tilde \calS) + 4\eps\cost(p, \greedy).\]
To remove the dependency in $\cost(p, \tilde \calS)$ from the right hand side, one can upper bound it with $\cost(p, \calS) + |\cost(p, \calS) - \cost(p, \tilde \calS)|$, which yields the following:
\begin{align*}
|\cost(p, \calS) - \cost(p, \tilde \calS)| 
&\leq \eps\cost(p, \calS) +2\eps\left(\cost(p, \calS) +|\cost(p, \calS) - \cost(p, \tilde \calS)|\right)\\
&\qquad + 4\eps\cost(p, \greedy)\\
 \iff |\cost(p, \calS) - \cost(p, \tilde \calS)| &\leq \frac{1}{1-2\eps}\left(3\eps\cost(p, \calS) +  4\eps\cost(p, \greedy)\right)\\
 &\leq 9\eps\cost(p, \calS) +  12\eps\cost(p, \greedy)
\end{align*}
Finally, rescaling $\eps$ concludes.

\end{proof}

\section{Planar Graphs}
\label{sec:planar}
The goal of this section is to prove the existence of small centroid sets for planar graph, analogously to the treewidth case. This is the following lemma:

\begin{lemma}\label{thm:planar-centroid}
  Let $G = (V, E)$ be an edge-weighted planar graph, a set $X \subseteq V$ and two positive integers $k$ and $z$. Furthermore, let $\greedy$ be a solution of $(k, z)$-clustering of $X$. 
  
  Then, there exists an $\greedy$-approximate centroid set for $(k,z)$-clustering on $V$ of size $\poly(|X|)\cdot \exp\big( O(z^3\eps^{-3} \log z/\eps)\big)$.

\end{lemma}

As for treewidth, this lemma implies the following corollary:

\begin{corollary}
\label{cor:coreset-planar}
  Let $G = (V, E)$ be an edge-weighted planar graph, a set $X \subseteq V$, and two positive integers $k$ and $z$ . 
  
  There exists an algorithm with running time $\tilde O(nk)$ that constructs an $\eps$-coreset for $(k, z)$-clustering on $X$ with size 
  \[O\left(\frac{\log^5 1/\eps}{2^{O(z\log z)}\min(\eps^2, \eps^z)}\left(k \log^2 k + \frac{k \log k}{\eps^3} + \log 1/\pi\right)\right)\qedhere\]
\end{corollary}

The big picture is the same as for treewidth. As in the treewidth case, planar graph can be broken into $\poly(X)$ pieces, each containing at most $2$ vertices of $X$.
The main difference is in the nature of the separators: while treewidth admit small vertex separators, the region in the planar decomposition are bounded by a few number of shortest path instead. This makes the previous argument void: we cannot round distances to all vertices in the boundary of a region. We show how to bypass this, using the fact that separators are shortest paths: it is enough to round distances to a well-chosen subset of the paths, as we will argue in the proof.

Formally, the decomposition is as follows:

\begin{lemma}[Lemma 4.5 of~\cite{BravermanJKW21}, see also~\cite{EisenstatKM14}]
\mbox{}\label{lem:planar-decomp}
For every edge-weighted planar graph $G = (V, E)$ and subset $X \subseteq V$, there exists a collection of subsets of $V$ $\Pi := \{V_i\}$ with $|\Pi| = \poly(|X|)$ and $\cup V_i = V$ such that, for every $V_i \in \Pi$:
\begin{itemize}
\item $|V_i \cap X| = O(1)$, and
\item there exists a collection of shortest paths $\calP_i$ with $|\calP_i| = O(1)$ such removing the vertices of all paths of $\calP_i$ disconnects $V_i$ from $V\setminus V_i$.
\end{itemize}
\end{lemma}

As for treewidth, we proceed as follows: given the decomposition of \cref{lem:planar-decomp}, for any center $s \in V_i$, we identify a point $q$ and round distances from $s$ to $\calP_i$ according to $\dist(q, \greedy)$. $\cand$ contains one point $\tilde s$ with the same rounded distances as $s$, and we will argue that $\tilde s$ can replace $s$. 
As mentioned, we cannot round distances to the whole shortest-paths $\calP_i$.
Instead, we show that it is enough to round distances from $s$ to points on the boundary of $V_i$ that are close to $q$: since the boundary consists of shortest path, it is possible to discretize that set.

\begin{proof}[Proof of \cref{thm:planar-centroid}]
Let $\Pi = \{V_i\}$ be the decomposition given by \cref{lem:planar-decomp}. 
For any $V_i$ and any $q \in X$, we define a set of \emph{landmarks} $\calL_{i,q}$ as follows: for any $P \in \calP_i$, let $\calL_{i, q, P}$ be a $\frac{\eps}{z}\cdot \dist(q, \greedy)$-net of $P \cap B\left(q,\frac{90z^2}{\eps^2 }\cdot \dist(q, \greedy)\right)$. Note that since $P$ is a shortest path, the total length of $P \cap B\left(q,\frac{90z^2}{\eps^2}\cdot \dist(q, \greedy)\right)$ is at most $\frac{180z^2}{\eps^2}\cdot \dist(q, \greedy)$, and so the net has size at most $\frac{180z^3}{\eps^3}$.
 We define $\calL_{i,q} = \left(V_i \cap X\right)\cup_{P \in \calP_i} \calL_{i,q,P}$. 

\textbf{Rounding the distances to $\calL_{i,q}$}
We now describe how we round distances to landmarks, and define $\cand$ such that for each possible distance tuple, $\cand$ contains a point having that distance tuple.
Formally, given a point $s \in V_i$ and a point $q \in X$, the distance tuple $\mathbf{d}_q(s)$ of $s$ is defined as $\mathbf{d}_q(s) = (\dist(s, x)~ | ~\forall~x\in X \cap V_i) + \left(\dist(s, y)~ |~ \forall y \in \calL_{i,q}, \forall i\right)$.
The \emph{rounded distance tuple} $\tilde{\mathbf{d}}_q(s)$ of $s$ is defined as follows
:
\begin{itemize}
\item For $x \in X \cap V_i$ or $x = q$, $\tilde d(s, x)$ is the multiple of 
$\frac{\eps}{z} \dist(x, \greedy)$ smaller than $\frac{10z}{\eps} \dist(x, \greedy)$ closest to $\dist(s, x)$. 
\item For $y \in \calL_{i,q}$, $\tilde d(s, y)$  is the multiple of 
$\frac{\eps}{z}\cdot \dist(q, \greedy)$ smaller than $\frac{90z^2}{\eps^2}\dist(q, \greedy)$ closest to $\dist(s, y)$. 
\end{itemize}

The set $\cand$ is constructed as follows: for every $V_i$ and every $q$, for every rounded distance tuple $\{\tilde{\mathbf{d}}_q(p)\}$, add to $\cand$ a point that realizes this rounded distance tuple (if such a point exists).

It remains to show both that $\cand$ has size $\poly(|X|)\exp\big( O(z^3\eps^{-3} \log z/\eps)\big)$, and that $\cand$ contains good approximation of each center of any given solution.

\paragraph*{Size analysis.}
For any given $V_i$ and $q$, there are $\left(\frac{90z^3}{\eps^3}\right)^{|\calL_{i,q}|}$ possible rounded distances. 
As explained previously, $|\calL_{i,q}| = O(z^3/\eps^{3})$. 

There are $|V|$ choices of $q$, and \cref{lem:planar-decomp} ensures that there are $\poly(|X|)$ choices for $V_i$.

Hence, the total size of $\cand$ is at most
$\poly(|X|)\cdot \exp\big( O(z^3\eps^{-3} \log z/\eps)\big)$.

\paragraph*{Error analysis.}
We now show that for all solution $\calS$, every center can be approximated by a point of $\cand$.
First, by applying \cref{lem:noproblem}, we can assume that  for any center $s \in \calS$, and $q = \argmin_{p : \dist(p, s) \leq \frac{10z}{\eps}\dist(p, \greedy)} \dist(p, \greedy) + \dist(p, s)$, there is no problematic point with respect to $q$ and $s$.

 Let $S$ be some cluster of $\calS$, with center $s$. As in \cref{lem:k2unionbound2} and \ref{lem:centroid-treewidth}, we aim at showing how to find $\tilde s \in \cand$ such that, for every $p \in X \cap S$ with $\dist(p, \calS) \leq \frac{10z}{\eps}\cdot \dist(p, \greedy)$, we have $|\cost(p, s) - \cost(p, \tilde s)| \leq 3\eps \left(\cost(p, s) + \cost(p, \greedy)\right)$.

For this, let $V_i$ be a part of $\Pi$ containing $s$, and $\calP_i$ be the paths given by \cref{lem:planar-decomp}.  
We let $q := \argmin_{p \in X: \dist(p, s) \leq \frac{10z}{\eps} \dist(p, \greedy)} \dist(p, s) + \dist(p, \greedy)$.
We define $\tilde s$ to be the point of $\cand$ that has the same rounded distance tuple to $\calL_{i,q}$ as $s$. Let $\tilde \calS$ be the solution constructed from $\calS$ that way.
 We show now that $\tilde \calS$ has the required properties. 



First, if $p \notin V_i$, then we show how to use that $s$ and $\tilde s$ have the same rounded distances to $\calL_{i,q}$.
\begin{itemize}
\item If $\dist(p, s) > \frac{21z^2}{\eps^2}\cdot \dist(q, \greedy)$, we argue that $d(s, \tilde s)$ is negligible. The argument is exactly alike the one from \cref{lem:centroid-treewidth}, we repeat it for completeness.

The rounding ensures that the distance to $q$ is preserved: $\dist(q, \tilde s) \leq \dist(q, s) + \frac{\eps}{z} \dist(q, \greedy)$, and therefore: 


\begin{align*}
\dist(s, \tilde s) &\leq 2\dist(q, s) + \frac{\eps}{z} \cdot \dist(q, \greedy)\\
&\leq \left(\frac{20z}{\eps} + \frac{\eps}{z}\right)\cdot \dist(q, \greedy)\\
&\leq \frac{21z}{\eps} \cdot \frac{\eps^2}{21z^2}\cdot \dist(p, s) \leq \frac{\eps}{z}\cdot \dist(p, s).
\end{align*}

Hence, using the modified triangle inequality \cref{lem:weaktri}, we can conclude:
 $|\cost(p, \tilde s) - \cost(p, s)| \leq \eps\cost(p, s) + \eps\cost(p, \greedy)$.
 
\item Otherwise, $\dist(p, s) \leq \frac{21z^2}{\eps^2}\cdot \dist(q, \greedy)$ and we can make use of the landmarks. 
Since $p \notin V_i$ the shortest-path $p \leadsto s$ and crosses $\calP_i$ at some vertex $x$. 

First, it holds that $\dist(x, q) \leq \dist(x, s) + \dist(s, q) \leq \dist(p, s) + \dist(s, q) \leq (\frac{10z}{\eps}+\frac{8z^2}{\eps^2})\dist(q, \greedy)$, hence $x$ is in $P \cap B(q, \frac{90z^2}{\eps^2} \dist(q, \greedy))$. 
By choice of landmarks, this implies that there is $\ell \in \calL_{i,q}$, with $\dist(x, \ell) \leq \frac{\epsilon}{z} \dist(q, \greedy)$. 
To show that $s$ and $\tilde s$ have the same distance to $\ell$, it is necessary to show that $s$ is not too far away from $\ell$: 
\begin{align*}
\dist(s, \ell) &\leq \dist(s, x) + \frac{\epsilon}{z} \cdot \dist(q, \greedy)  \leq \dist(p, s) + \frac{\epsilon}{z} \dist(q, \greedy) \\
&\leq \frac{21z^2}{\eps^2}\cdot \dist(q, \greedy) + \frac{\eps}{z}\cdot\dist(q, \greedy)
\end{align*}
Hence, $s$ is close enough to $\ell$ to ensure that $\tilde s$ has the same rounded distance to $\ell$ as $s$, and we get: 
\begin{align*}
\dist(p, \tilde s) & \leq \dist(p, \ell) + \dist(\ell, \tilde s) \\
&\leq \dist(p, \ell) + \dist(\ell, s) + \frac{\epsilon}{z}\cdot \dist(q, \greedy) \\
&\leq \dist(p, x) + \dist(x, s) + 2\dist(x, \ell)+\frac{\epsilon}{z}\cdot \dist(q, \greedy) \\
&= \dist(p, s) + \frac{3\epsilon}{z}\cdot\dist(q, \greedy)
\end{align*}
First, this ensures that $\dist(p, \tilde s) \leq \frac{8z^2}{\eps^2} \cdot \dist(q, \greedy)$, and so we can repeat the argument switching roles of $s$ and $\tilde s$, to get $|\dist(p, \tilde s) - \dist(p, s)| \leq \frac{3\epsilon}{z}\cdot\dist(q, \greedy)$
Using that $p$ is not problematic with respect to $q$ and $s$, we can conclude that 
\[|\dist(p, \tilde s) - \dist(p, s)| \leq\frac{3\epsilon}{z}\cdot(\dist(p, \greedy) + \dist(p, s)).\] 
In turn, using \cref{lem:weaktri}, we conclude:
\[|\cost(p, \tilde s) - \cost(p, s)| \leq \eps \cdot(\cost(p, \greedy) + \cost(p, s)).\] 
\end{itemize}

Finally, in the case where $p \in V_i$, then we get either $|\dist(p, \tilde s) - \dist(p, s)| \leq \frac{\eps}{z} \dist(p, \greedy)$ and we are done, or both $\dist(p, \tilde s)$ and $\dist(p, s)$ are bigger than $\frac{8z}{\eps} \dist(p, \greedy)$.

We can now conclude, exactly as in the treewidth case:
in all possible cases, it holds that either $\dist(p, s)$ and $\dist(p, \tilde s)$ are bigger than $\frac{10z}{\eps}\dist(p, \greedy)$, or: 
\begin{equation}\label{eq:centroid-planar}
|\cost(p, \tilde s) - \cost(p, s)| \leq \eps\cost(p, s) + \eps\cost(p, \greedy).
\end{equation}

To extend that result to the full solutions $\calS$ and $\tilde \calS$ instead of a particular center, we note that since $p$ is interesting, $\dist(p, \calS) \leq \frac{8z}{\eps}\dist(p, \greedy)$. Hence, we can apply \cref{eq:centroid-planar} with $s$ being the closest point to $p$ in $\calS$:  
$\cost(p, \tilde \calS) \leq (1+\eps)\cost(p, \calS) + \eps\cost(p, \greedy)$.

In particular, this implies that $\dist(p, \tilde s) \leq \frac{10z}{\eps}\dist(p, \greedy)$.
Chose now $\tilde s$ to be the closest point to $p$ in $\tilde \calS$ and $s$ its corresponding center in $\calS$. Using \cref{eq:centroid-planar} therefore gives:
\begin{align*}
\cost(p, \calS) &\leq (1+\eps)\cost(p, \tilde \calS) + \eps\cost(p, \greedy) \\
&\leq (1+\eps)^2 \cost(p, \calS) + \eps(2+\eps) \cost(p, \greedy)\\
&\leq (1+3\eps)\cost(p, \calS) + 3\eps \cost(p, \greedy).
\end{align*}

Rescaling $\eps$ and combining the two inequality concludes.
\end{proof}
\newcommand{\bq}{\textbf{q}}

\section{Minor-Excluded Graphs} \label{sec:centroid-minor}
A graph $H$ is a \textit{minor} of a graph $G$ if it can be obtained from $G$ by deleting edges and vertices and contracting edges. 

We are interested here in families of graph excluding a fixed minor $H$, i.e. none of the graph in the family contains $H$ as a minor. The graphs are weighted: we assume that for each edge, its value is equal to shortest-path distance between its two endpoints.

The goal of this section is to prove the following lemma, analogous to \cref{thm:planar-centroid}.

\begin{lemma}
\mbox{}\label{thm:mf-centroid}
  Let $G = (V, E)$ be an edge-weighted graph that excludes a minor of fixed size, a set $X \subseteq V$ and two positive integers $k$ and $z$. Furthermore, let $\greedy$ be a solution of $(k, z)$-clustering of $X$. 
  
  Then, there exists an $\greedy$-approximate centroid set for $(k,z)$-clustering on $V$ of size  $\exp(O(\log^2 |X| + \log |X| / \eps^4))$.
\end{lemma}

As for the bounded treewidth and planar cases, this lemma implies the following corollary:

\begin{corollary}
\mbox{}\label{cor:coreset-minor}
  Let $G = (V, E)$ be an edge-weighted graph that excludes a fixed minor, and two positive integers $k$ and $z$ . 
  
  There exists an algorithm with running time $\tilde O(nk)$ that constructs an $\eps$-coreset for $(k, z)$-clustering on $V$ with size 
  \[O\left(\frac{\log^5 1/\eps}{2^{O(z\log z)}\min(\eps^2, \eps^z)}\left(k \log^2 k \log(1/\eps) + \frac{k \log k}{\eps^4} + \log 1/\pi\right)\right)\qedhere\]
\end{corollary}

The big picture is the same as for planar graphs. Minor-free graphs have somewhat nice separators, that we can use to select centers. However, those separators are not shortest paths in the original graph, as described in the next structural lemma.

%
%

\begin{lemma}[Lemma 4.12 in \cite{BravermanJKW21}, from Theorem 1 in \cite{AbrahamG06}]
\mbox{}\label{lem:mf-decomp}
For every edge-weighted graph $G = (V, E)$ excluding some fixed minor, and subset $X \subseteq V$, there exists a collection of subsets of $V$ $\Upsilon := \{\Pi_i\}$ with $|\Upsilon| = \poly(|X|)$ and $\cup \Pi_i = V$ such that, for every $\Pi_i \in \Upsilon$:
\begin{itemize}
\item $|\Pi_i \cap X| = O(1)$, and
\item there exists a groups of paths $\{\calP^i_j\}$ with $|\cup \calP^i_j| = O(\log |X|)$ such that removing the vertices of all paths of $\calP_i$ disconnects $\Pi_i$ from $V\setminus \Pi_i$, and such that paths in $\calP^i_j$ are shortest-paths in the graph $G^i_j := G \setminus \left(\cup_{j' < j} \calP^i_{j'}\right)$.
\end{itemize}
\end{lemma}

The general sketch of the proof is as follows: we consider the boundary $B$ of a region $\Pi_i$, and enumerate all possible tuple of distances from a point inside the leaf to the boundary. For each tuple, we include in $\cand$ a point realizing it. Of course, this would lead to a set $\cand$ way too big: the boundary of each leaf consists of too many points, and there are too many distances possible. For that, we show how to discretize the boundary, and how to round distances from a point to the boundary.

Discretizing the boundary is not as easy as in the planar case, as the separating paths are not shortest paths in the original graph $G$. A separating path $P \in \calP_j$, however, is a shortest path in the graph $G^i_j := G \setminus \left(\cup_{j'<j} \calP^i_{j'} \right)$.

As in the planar case, we therefore start from the point $q$ closest to $s$ in the graph $G^i_j$. Note here that we cannot infer much on the distances in the original graph $G$: for this reason, we are not able to apply \cref{lem:noproblem}, and we need to present a whole different argument. 

We will assume that we know $D = \dist_j(q, s)$, where $\dist_j$ is the distance in the graph $G^i_j$. In that case, we can simply take an $\eps D$-net of $P \cap B_j(q, D)$, where $B_j(q, D)$ is the ball centered at $q$ and of radius $D$ in $G^i_j$. This net has size $O(1/\eps^2)$, as $P$ is a shortest path in $G^i_j$. Then, if $\tilde s$ has same distances to this net as $s$, we are able to show as in the previous cases that for any point separated from $s$ by $P$, $\dist(p, \tilde s) \lessapprox \dist(p, s)$; and for any point separated from $\tilde s$ by $P$, $\dist(p, s) \lessapprox \dist(p, \tilde s)$.

To estimate $\dist_i(q, s)$, we proceed as follows: either $\dist_i(q, s) \approx \dist_i(q, q_2)$ for some $q_2 \in X$, or not. In the first case, we can pick such a $q_2$. In the second case, we will need to ensure that when $p$ is such that $\dist_i(p, q) \gg \dist_i(q, s)$, then $\tilde s$ stays close to $q$. When $p$ is such that $\dist_i(q, p) \ll \dist_i(q, s)$, then $p$ and $q$ are essentially located at the same spot, and we ensure that $\tilde s$ stays far from $q$.

\subsection{Construction of the centroid set.}
From \cref{lem:mf-decomp}, we have a decomposition into regions $\Upsilon = \{\Pi_i\}$. In this argument, we fix a region $\Pi_j \in \Upsilon$. $\Pi_j$ is bounded by $O(\log |X|)$ paths $P_1, ..., P_m$ and $P_i$ is a shortest path in some graph $G_i$, subgraph of $G$: if $P_i \in \calP^j_\ell$, then $G_i := G^j_\ell$. We change the indexing for simplicity, and let $\Pi = \Pi_j$. We let $\dist_i$ be the distances in the graph $G_i$.


We consider two ways of rounding the distances. The first starts from a point $q_1 \in X$, and is useful when there is $q_2 \in X$ such that $\eps \dist_i(q_1, s) \leq \dist_i(q_1, q_2) \leq \frac{1}{\eps}\dist_i(q_1, s)$.

Along each paths, we designate portals as follows. Consider a path $P_i$. For any pair of vertices $q_1, q_2 \in X$, let $D = \dist_i(q_1, q_2) + \dist(q_2, \greedy)$ and let $N_{i, q_1, q_2}$ be an $\eps^2 D$-net of $P_i \cap B_i(q_1, \frac{D}{\eps^2})$, where $B_i(q, \frac{D}{\eps^2})$ is the ball centered at $q$ and of radius $\frac{D}{\eps^2}$ in $G_i$.

For each possible $q_1, q_2$ and any point $s \in \Pi$, we consider the following distance tuple: $\left(\dist_i(s, n),~\forall n \in N_{i,q_1, q_2}\right) \cup \left(\dist_i(s, q_1)\right) \cup \left(\dist_i(x, s),~\forall x \in \Pi \cap X\right)$. 
We define the rounded tuple $\tilde d^1(q_1, q_2) :=  \left(\tilde d^1(s, n),~\forall n \in N_{i,q_1, q_2}\right)\cup\left(\tilde d^1(s, q_1)\right)\cup \left(\tilde d^1(x, s),~\forall x \in \Pi \cap X\right)$, where {\setlength{\emergencystretch}{2.5em}\par}
\begin{itemize}
\item $\tilde d^1(s, n)$ is the multiple of $\eps^2 D$ closest to $\min\left(\frac{3D}{\eps^2}, \dist_i(s, n)\right)$.
\item $\tilde d^1(s, q_1)$ is the multiple of $\eps D$ closest to $\dist_i(s, q_1)$ and smaller than $\frac{3D}{\eps}$.
\item for any $x \in \Pi \cap X$, $\tilde d^1(x, s)$ is the closest multiple of $\eps \dist(x, \greedy)$ to $\dist_i(x, s)$ smaller than $\frac{1}{\eps} \cdot \dist(x, \greedy)$.
\end{itemize}

We also consider another rounding, which will be helpful when for all points	, $\dist(p, \greedy) + \dist_i(q, q_1) \notin [\eps \dist_i(q_1, s), \frac{1}{\eps}\dist_i(q_1, s)]$.

For any $q_1, q_3$, and $q_4$ in $X$, 
 $\tilde d^2(q_1, q_3, q_4) = \top$ when $\frac{1}{\eps}\cdot(\dist_i(q_1, q_4) +\dist(q_4, \greedy)) < \dist_i(q_1, s) < \eps \cdot (\dist_i(q_1, q_3) + \dist(q_3, \greedy))$, and $\tilde d^2(q_1, q_3, q_4) = \bot$ otherwise.
 $q_3$ or $q_4$ may be unspecified. In that case, the corresponding part of the inequality is dropped.\footnote{When $q_3$ is unspecified,  $\tilde d^2(q_1, q_3, q_4) = \top$ when $\frac{1}{\eps}\cdot(\dist_i(q_1, q_4) +\dist(q_4, \greedy)) < \dist_i(q_1, s)$, and $\tilde d^2(q_1, q_3, q_4) = \bot$ otherwise. When $q_4$ is unspecified,  $\tilde d^2(q_1, q_3, q_4) = \top$ when $\dist_i(q_1, s) < \eps \cdot (\dist_i(q_1, q_3) + \dist(q_3, \greedy))$, and $\tilde d^2(q_1, q_3, q_4) = \bot$ otherwise.}

 To construct $\cand$, we proceed as follows: for any region $\Pi \in \Upsilon$ given by \cref{lem:mf-decomp}, and for any path $P_i$ in the boundary of $\Pi$, select a rounding $\tilde d_i^1(q^i_1, q^i_2)$ or $\tilde d_i^2(q^i_1, q^i_3, q^i_4)$. If there is any, pick one point $s$ achieving all those rounding distances, and add $s$ to $\cand$.
   
   We will show \cref{thm:mf-centroid} using this centroid set. For that, we break the proof into two parts: first, the size of $\cand$ is the desired one; then, $\cand$ is indeed an approximate centroid set.
   
\subsection{$\cand$ has Small Size}
\begin{lemma}
\mbox{}\label{lem:minorSize}
$\cand$ constructed as previously has size $\exp\left(O(\log^2 |X| + \log|X| / \eps^{-4})\right)$.
\end{lemma}
\begin{proof}
Fix a region $\Pi$, a path $P_i$ on $\Pi$'s boundary, and points $q_1, q_2$. There are $O\left(1/\eps^4\right)$ points in the net $N_{i,q_1,q_2}$, and $O(1)$ in $\Pi \cap X$. For each of those points, there are at most $3/\eps^4$ many choices of distances. 

For a fixed region $\Pi$, path $P_i$ on $\Pi$'s boundary, and points $q_1, q_2, q_3$, there only $2$ possible different $\tilde d^2(q_1,q_2,q_3)$.

Now, there are $\poly(|X|)$ many regions $\Pi$, and for each of them $O(\log |X|)$ many paths $P_i$. For each path, there are at most $|X|^3$ choices of $q_j$ for it, so in total $|X|^{O(\log |X|)}$ possible choices. Each choice gives rise to $O(\log |X|) \cdot O\left(1/\eps^4\right)$ many net points, each having at most $3/\eps^4$ many choices of distances.


So, in total, there are 
\[|X|^{O(\log |X|)} \cdot (1/\eps)^{O\left(\log |X|/\eps^4\right)}\]
many choices of rounded distances tuples. That upper bounds the size of $\cand$, as there is at most one point per rounded distance tuple.
\end{proof}

\subsection{$\cand$ is an Approximate Centroid Set}
\paragraph{Construction of solution $\tilde \calS$.}
Now, for a point $s \in \calS$, we construct $\tilde s$ as follows, using the rounded distance tuples. Let $\Pi$ be a region of $\Upsilon$ that contains $s$. For each path $P_i$ in the boundary of $\Pi$, we define the tuple $\tilde d_i$ as follows. Let $q_1^i$ be the point minimizing $\dist_i(p, s)$. Now, we distinguish two cases:
\begin{itemize}
\item either there is some $q^i_2$ such that $\eps \dist_i(q^i_1, s) \leq \dist_i(q^i_1, q^i_2) + \dist(q^i_2, \greedy) \leq \frac{1}{\eps}\dist_i(q^i_1, s)$. Then $\tilde d_i$ is the tuple $\tilde d^1(q^i_1, q^i_2)$.
\item or there exists points $q$ with $\dist_i(q^i_1, q) + \dist(q, \greedy) > \frac{1}{\eps}\dist_i(q^i_1, s)$: let $q^i_3$ be such a point, with smallest $\dist_i(q^i_1, q^i_3)+ \dist(q^i_3, \greedy) $ value. If there are no such points, $q^i_3$ is unspecified.

If there are points $q$ with $\dist_i(q^i_1, q) + \dist(q, \greedy) < \eps\dist_i(q^i_1, s)$, then let $q^i_4$ be the point with largest $\dist_i(q^i_1, q^i_4) + \dist(q^i_4, \greedy) $ value. Otherwise, $q^i_4$ is unspecified. Note that since we are not in the first case, either $q^i_3$ are $q^i_4$ is specified.

Then $\tilde d_i$ is the tuple $\tilde d^2(q^i_1, q^i_3, q^i_4)$.
\end{itemize}

$\tilde s$ is chosen to be in $\cand \cap \Pi$ and to have the same rounded distance tuples as $s$, for \textit{all} the rounded tuples $\tilde d_i$. $\tilde \calS$ is the union of all those $\tilde s$ for $s \in \calS$.

\begin{lemma}
\mbox{}\label{lem:minorGood}
Let $\calS$ be a solution, and $s \in \calS$. Let $\tilde s$ defined as previously.
For any point $p \in X$, either $|\cost(p, s) - \cost(p, \tilde s)| \leq \eps(\cost(p, s) + \cost(p, \greedy))$ or both $\dist(p, s)$ and $\dist(p, \tilde s)$ are bigger than $\frac{10z\cdot \dist(p, \greedy)}{\eps}$.
\end{lemma}
\begin{proof}
Fix $s \in \calS \cap \Pi$, and let $\tilde s$ be its corresponding point in $\tilde \calS$. Let $s_1 \in \{s, \tilde s\}$, and $s_2$ the other choice: we will show that $\dist(p, s_1) \leq (1+\eps) \dist(p, s_2) + \eps \dist(p, \greedy)$. This implies that the costs verify the same inequality, which will allow us to conclude, switching the roles of $s_1$ and $s_2$.

First, in the case where $p \in \Pi \cap X$, then the rounding directly ensures that either $\dist(p, s) > 1/\eps \cdot \dist(x, \greedy)$, in which case it holds as well than $\dist(p, \tilde s) > 1/\eps \cdot \dist(x, \greedy)$, or $|\dist(p, s) - \dist(p, \tilde s)| \leq \eps \dist(x, \greedy))$.

Otherwise, $p$ is separated from $s_1$ by some path among $\{P_1, ..., P_m\}$.
Let $i$ be the smallest integer such that $P_i$ intersects the shortest path between $p$ and $s_1$. 
Our argument depends on the type of tuple $\tilde d_i$ chosen for $s$. Since $s$ and $\tilde s$ have the same rounded distance tuples $\tilde d_1, \tilde d_2,...$, they have in particular the same rounded distance $\tilde d_i$.
Let $q^i_1$ be the point with smallest $\dist_i(q, s)$ value (importantly, the $q^i_1, q^i_2, q^i_3$ and $q^i_4$ appearing in the proof are defined with respect to $s$, not to $s_1$).

\paragraph{If we can estimate $\dist_i(q^i_1, s)$.} In the first case, there is a $q^i_2$ such that $\eps \dist_i(q^i_1, s) \leq \dist_i(q^i_1, q^i_2) + \dist(q^i_2, \greedy) \leq \frac{1}{\eps}\dist_i(q^i_1, s)$. We let  $D := \dist_i(q^i_1, q^i_2) + \dist(q^i_2, \greedy) $ our (rough) estimate on the distance $\dist_i(q^i_1, s)$.

Then, our argument goes as follows. Let $x$ be a point in the intersection of $P_i$ and the shortest path $s_1 \leadsto p$. We have the following properties: by choice of $i$, $\dist_i(p, s_1) = \dist(p, s_1)$ and $\dist_i(x, s_1) = \dist(x,s_1)$. By choice of $x$, $\dist(p, x) + \dist(x, s_1) = \dist(p, s_1)$. Last, by choice of $q^i_1$, $\dist_i(q^i_1, s) \leq \dist_i(p, s)$, and $D \leq \frac{\dist_i(p, s)}{\eps}$.
\begin{itemize}
\item First, if $\dist_i(x, q^i_1) \leq \frac{D}{\eps^2}$. Then there is a point $n$ from $N_{i,q^i_1,q^i_2}$ with $\dist_i(n, x) \leq \eps^2 D$. Furthermore, $\dist_i(s_1, n) = \dist_i(s_2, n) \pm \eps^2 D$, as $\dist_i(s,n) \leq \dist_i(s, x) + \dist_i(x, q^i_1) + \dist_i(q^i_1, s) \leq \frac{3D}{\eps^2}$ and so $n$ has same rounded distances to $s_1$ and $s_2$. Hence, we get:
\begin{align*}
 \dist_i(p, s_2) &\leq \dist_i(p, x) + \dist_i(x, n) + \dist_i(n, s_2)\\
&\leq \dist_i(p, x) + \dist_i(x, n) + \dist_i(n, s_1) + \eps^2 D\\
&\leq \dist_i(p, x) +\dist_i(x, s_1) +2\dist_i(x, n)  + \eps^2 D\\
&\leq \dist_i(p, s_1) + 3\eps^2 D\\
&\leq \dist(p, s_1) + 3\eps \dist_i(p, s).
\end{align*}
Now, two cases: either $s = s_1$ and $\dist_i(p, s) = \dist(p, s)$, and then we get $\dist(p, s_2) \leq (1+3\eps) \dist(p, s_1)$. Or $s = s_2$, and we have $(1-3\eps) \dist(p, s) \leq \dist(p, \tilde s)$ which implies $\dist(p, s_2) \leq (1+6\eps) \dist(p, s_1)$.

\item Otherwise, $\dist_i(x, q^i_1) > \frac{D}{\eps^2}$: we first show that $\dist(s, \tilde s) \leq 3\eps (\dist_i(p, s) + \dist(p,\greedy))$, which will allow to conclude. It holds that $\dist_i(q^i_1, s) = \dist_i(q^i_1, \tilde s) \pm \eps D$, as by definition of $D$, $\dist_i(q^i_1, s) \leq D/\eps$. Hence, 
\begin{align*}
\dist_i(s, \tilde s) &\leq \dist_i(s, q^i_1) + \dist_i(\tilde s, q^i_1) \leq 2\dist_i(s, q^i_1) + \eps D\\
&\leq \frac{2+\eps^2}{\eps} D \leq 3\eps \dist_i(x, q^i_1)\\
&\leq 3\eps(\dist_i(x, s_1) + \dist_i(s_1, s_2) + \dist_i(s, q^i_1)) \\
\Rightarrow  \dist(s, \tilde s) &\leq 9\eps (\dist(p, s_1) + \dist_i(p,s)).
\end{align*}
Hence,
\begin{align*}
\dist_i(p, s_2) &\leq \dist_i(p, s_1) + \dist_i(s, \tilde s)\\
&\leq \dist(p, s_1) + 9\eps (\dist(p, s_1) + \dist_i(p,s))
\end{align*}
Similarly as in the previous case, either $s_1 = s$ and the right hand side is $(1+18\eps) \dist(p, s_1)$, or $s_2 = s$ and we infer $\dist(p, s_2) \leq (1+27\eps) \dist(p, s_1)$.
\end{itemize}

\paragraph{When we can only overestimate or underestimate $\dist_i(q^i_1, s)$ }
In the second case, $q^i_3$ is such that $\dist(q^i_3, \greedy) + \dist_i(q^i_1, q^i_3) > \frac{1}{\eps}\dist_i(q^i_1, s)$, and has minimal $\dist_i(q^i_1, q^i_3)+ \dist(q^i_3, \greedy)$ value among those. Similarly,
$q^i_4$ is the point with largest $\dist_i(q^i_1, q^i_4) + \dist(q^i_4, \greedy)$ value among those verifying $\dist_i(q^i_1, q) + \dist(q, \greedy)< \eps\dist_i(q^i_1, s)$.

By choice of $q^i_3$ and $q^i_4$, it must be that 
\[\frac{1}{\eps}\cdot(\dist_i(q^i_1, q^i_4) +\dist(q^i_4, \greedy)) < \dist_i(q^i_1, s) < \eps \cdot (\dist_i(q^i_1, q^i_3) + \dist(q^i_3, \greedy)).\]

Hence, $\tilde d^2(q^i_1, q^i_3, q^i_4) = \top$, and $\tilde s$ is chosen such that 
\[\frac{1}{\eps}\cdot(\dist_i(q^i_1, q^i_4) +\dist(q^i_4, \greedy)) < \dist_i(q^i_1, \tilde s) < \eps \cdot (\dist_i(q^i_1, q^i_3) + \dist(q^i_3, \greedy)).\]

Since we are not in the first case where we can estimate $\dist_i(q^i_1, s)$, $p$ verifies either $\dist(p, \greedy) + \dist_i(p, q^i_1) < \eps \dist_i(q_1^i, s)$ or $\dist(p, \greedy) + \dist_i(p, q^i_1) > \frac{1}{\eps} \dist_i(q_1^i, s)$. 

First, if $\dist_i(p, q^i_1) + \dist(p, \greedy) > \frac{1}{\eps} \dist_i(q_1^i, s)$. Then we have, by choice of $q^i_3$:
\begin{align*}
\dist_i(s, \tilde s) &\leq \dist_i(s, q^i_1) + \dist_i(\tilde s, q^i_1)\\
&\leq 2\eps (\dist_i(q^i_1, q^i_3) + \dist(q^i_3, \greedy))\\
&\leq 2\eps (\dist_i(q^i_1, p) + \dist(p, \greedy))\\
&\leq 2\eps (\dist_i(p, s) + \dist_i(q^i_1, s) + \dist(p, \greedy))\\
&\leq 4\eps (\dist_i(p, s) + \dist(p, \greedy))\\
\end{align*}
and therefore, we can conclude just as before (distinguishing whether $s = s_1$ or $s=s_2$) that
\[\dist(p, s_2) \leq (1+12\eps) \dist(p, s_1) + 12\eps \dist(p, \greedy).\]

Lastly, in the case where  $\dist_i(p, q^i_1)+ \dist(p, \greedy)  < \eps \dist_i(q_1^i, s)$, we use that 
$\dist_i(q^i_1, s_1) > \frac{1}{\eps}\cdot(\dist_i(q^i_1, q^i_4) +\dist(q^i_4, \greedy))$ (as both $s$ and $\tilde s$ verifies this) to get:
\begin{align*}
\dist(p, s_1) &= \dist_i(p, s_1) \geq \dist_i(q^i_1, s_1) - \dist_i(p, q^i_1)\\
&\geq \frac{1}{\eps}\cdot (\dist_i(q^i_1, q^i_4) + \dist(q^i_4, \greedy)) - \dist_i(p, q^i_1)\\
&\geq \frac{1}{\eps}\cdot (\dist_i(q^i_1, p) + \dist(p, \greedy)) - \dist_i(p, q^i_1)\\
&\geq \frac{\dist(p, \greedy)}{\eps}.
\end{align*}

\paragraph{Conclusion.}
Rescaling $\eps$ by $1/27z$, the previous inequalities gives us that either $\dist(p, s_1) \geq \frac{27z \cdot \dist(p, \greedy)}{\eps}$, or $\dist(p, s_2) \leq (1+\eps/z) \dist(p, s) + \eps/z\cdot \dist(p, \greedy)$. The second inequality combined with \cref{lem:weaktri} implies that $\cost(p, s_2) \leq (1+\eps) \cost(p, s_1) + \eps\cdot \cost(p, \greedy)$. Therefore, using this result
with $s_1 = s, s_2 = \tilde s$ and then $s_1 = \tilde s, s_2 = s$  shows that:
\begin{itemize}
\item either $\dist(p, s) \geq \frac{27z \cdot \dist(p, \greedy)}{\eps}$, or $\cost(p, \tilde s) \leq (1+\eps) \cost(p, s) + \eps\cdot \cost(p, \greedy)$
\item either $\dist(p, \tilde s) \geq \frac{27 z\cdot \dist(p, \greedy)}{\eps}$, or $\cost(p, s) \leq (1+\eps) \cost(p, \tilde s) + \eps\cdot \cost(p, \greedy)$.
\end{itemize}

Therefore, if $p$ is such that $\dist(p, s) \leq \frac{5z\cdot \dist(p, \greedy)}{\eps}$, then $\dist(p, \tilde s) \leq  \frac{10z\cdot \dist(p, \greedy)}{\eps}$, and reciprocally when $\dist(p, \tilde s) \leq \frac{5z\cdot \dist(p, \greedy)}{\eps}$, then $\dist(p, s) \leq \frac{10z\cdot \dist(p, \greedy)}{\eps}$.

Thus we conclude: either both $\dist(p, \tilde s)$ and $\dist(p, s)$ are bigger than $\frac{10z\cdot \dist(p, \greedy)}{\eps}$, and we are done. Or both are smaller than $\frac{20z\cdot \dist(p, \greedy)}{\eps}$, and then using the previous inequalities we get:
\[|\cost(p, s) - \cost(p, \tilde s)| \leq 2\eps(\cost(p, s) +\cost(p, \tilde s) + \cost(p, \greedy))\]
Which, using $\cost(p, \tilde s) \leq \cost(p, s) + |\cost(p, s) - \cost(p, \tilde s)|$, yields
\[|\cost(p, s) - \cost(p, \tilde s)| \leq 7\eps(\cost(p, s) + \cost(p, \greedy)).\]
\end{proof}

\cref{lem:minorGood} gives exactly the same guarantee as \cref{eq:centroid-treewidth}: hence, as in the proof for treewidth, we can conclude from that inequality that for any solution $\calS$ and any interesting point $p$, $|\cost(p, \calS) - \cost(p, \tilde \calS)| \leq \eps(\cost(p, \calS) + \cost(p, \greedy))$.

Combining the guarantees from \cref{lem:minorGood} and \cref{lem:minorSize} concludes the proof of \cref{thm:mf-centroid}.
\section{A Note on Euclidean Spaces}\label{sec:dimension}

Lastly, we briefly want to survey the state of the art results for eliminating the dependency on the dimension in Euclidean spaces.

In a nutshell, the frameworks by both Feldman and Langberg~\cite{FeldmanL11} and us only yield coresets of size $O(k d \text{poly}(\log k,\varepsilon^{-1}))$.
To eliminate the dependency on the dimension, we typically have to use some form of dimension reduction.

In a landmark paper,~\cite{FeldmanSS20} showed that one can replace the dependency on $d$ with a dependency on $k/\varepsilon^{2}$ for the $k$-means problem, see also ~\cite{CEMMP15} for further improvements on this idea. Subsequently, Sohler and Woodruff~\cite{SohlerW18} gave a construction for arbitrary $k$-clustering objectives which lead to the first existence proof of dimension independent coresets for these problems. Unfortunately, there were a few caveats; most notably a running time exponential in both $k$. 
Huang and Vishnoi~\cite{huang2020coresets} showed that the mere existence of the Sohler-Woodruff construction was enough to compute coresets of size $\text{poly}(k/\varepsilon)$.
Recently, the Sohler-Woodruff result was made constructive in the work of Feng, Kacham and Woodruff~\cite{FKW19}.

Having obtained a $\text{poly}(k/\varepsilon)$-sized coreset, one can now use a terminal embedding to replace the dependency on $d$ by a dependency $\varepsilon^{-2} \log k/\varepsilon$. Terminal embeddings are defined as follows:
\begin{definition}[Terminal Embeddings]
Let $\varepsilon\in (0,1)$ and let $A\subset \mathbb{R}^d$ be arbitrary with $|A|$ having size $n>1$. Define the Euclidean norm of a $d$-dimensional vector $\|x\| = \sqrt{\sum_{i=1}^d x_i^2}$. Then a mapping $f:\mathbb{R}^d\rightarrow\mathbb{R}^m$ is a \emph{terminal embedding} if
\[\forall x\in A,~\forall y\in \mathbb{R}^d,~(1-\varepsilon)\cdot\|x-y\| \leq \|f(x)-f(y)\| \leq (1+\varepsilon)\cdot \|x-y\|.\]
\end{definition}

Terminal embeddings were studied by~\cite{ElkinFN17,MahabadiMMR18,NaN18}, with Narayanan and Nelson \cite{NaN18} achieving an optimal target dimension of $O(\varepsilon^{-2}\log n)$, where $n$ is the number of points\footnote{See the paper by Larsen and Nelson for a matching lower bound~\cite{LarsenN17}}.

It was first observed by Becchetti et al.~\cite{BecchettiBC0S19} how terminal embeddings can be combined with the Feldman-Langberg~\cite{FeldmanL11} (or indeed our) framework. Specifically, given the existence of a $\text{poly}(k/\varepsilon)$-sized coreset, applying a terminal embedding with $n$ being the number of distinct points in the coreset now allows us to further reduce the dimension. At the time, the only problem with such a coreset bound was $k$-means. The generalization to arbitrary $k$-clustering objectives is now immediate following the results by Huang and Vishnoi~\cite{huang2020coresets} and Feng et al.~\cite{FKW19}.

It should be noted that more conventional Johnson-Lindenstrauss type embeddings proposed in \cite{BecchettiBC0S19,CEMMP15,MakarychevMR19} do not (obviously) imply the same guarantee as terminal embeddings. We appended a short proof showing that terminal embeddings are sufficient at the end of this section. For a more in-depth discussion as to why normal Johnson-Lindenstrauss transforms may not be sufficient, we refer to Huang and Vishnoi~\cite{huang2020coresets}.

Combining our $O(k (d+\log k)\cdot \varepsilon^{-\max(2,z)})$ bound for general Euclidean spaces with either the Huang and Vishnoi~\cite{HuangJV19} or the Feng et al.~\cite{FKW19} constructions and terminal embeddings now immediately imply the following corollary.

\begin{corollary}\label{cor:coreset-euclidean}
There exists a coreset of size $$O\left(k \log k \cdot  \left(\varepsilon^{-2-\max(2,z)}\right)\cdot 2^{O(z\log z)}\cdot \mathrm{polylog}(\varepsilon^{-1})\right)$$ for $(k,z)$-clustering in Euclidean spaces.
\end{corollary}

Huang and Vishnoi further considered clustering in $\ell_p$ metrics for $p\in[1,2)$, i.e. non-Euclidean spaces. For this they reduced constructing a coreset for $(k,z)$ clustering in an $\ell_p$ space to constructing a constructing a coreset for $(k,2z)$ clustering in Euclidean space. Plugging in our framework into their reduction then yields the following corollary:

\begin{corollary}\label{cor:coreset-lp}
There exists a coreset of size $$O\left(k \log k \cdot  \left(\varepsilon^{-2-2z}\right)\cdot 2^{O(z\log z)}\cdot \mathrm{polylog}(\varepsilon^{-1})\right)$$ for $(k,z)$-clustering in any $\ell_p$ space for $p\in [1,2)$.
\end{corollary}

\begin{proposition}
\label{prop:coresetpreserve}
Suppose we have a (possibly weighted) point set $A$ in $\mathbb{R}^d$. Let $f:\mathbb{R}^d\rightarrow \mathbb{R}^{m}$ with $m\in O(\varepsilon^{-2}\cdot z^2 \log n)$ be a terminal embedding for $A$ and let $f(A)$ be the projected point set. Then if $f(P)\subset f(A)$ is an $\varepsilon$-coreset for $f(A)$, $P\subset A$ is an $O(\varepsilon)$-coreset for $A$.
Conversely, if $P\subset A$ is an $\varepsilon$-coreset for $A$, then $f(P)\subset f(A)$ is an $O(\varepsilon)$-coreset for $f(A)$
\end{proposition}
\begin{proof}
We prove the result for the first direction, the other direction is analogous.
Consider an arbitrary solution $S$ in $\mathbb{R}^d$.
We first notice that for any point $p\in A$, we have
\begin{equation*}
(1-\varepsilon/2z)^z \cdot \cost(f(p),f(\calS)) \leq (1-\varepsilon) \cdot \cost(f(p),f(\calS)) 
\end{equation*}
and  
\begin{equation*}
(1+\varepsilon/2z)^z \cdot \cost(f(p),f(\calS)) \geq (1+\varepsilon) \cdot \cost(f(p),f(\calS)) 
\end{equation*}
Therefore,
\begin{equation}
\label{eq:termin1}
  (1-\varepsilon) \cdot \cost(f(p),f(\calS)) \leq \cost(p,\calS) \leq (1+\varepsilon)\cdot \cost(f(p),f(\calS)).
\end{equation} 
Now suppose $f(P)$ is a coreset for $f(A)$, which means for any set of $k$ points $f(S)\subset \mathbb{R}^m$ 
\begin{eqnarray}
\left\vert \sum_{p\in f(A)} w_p \cdot \cost(p,f(\calS)) - \sum_{q\in f(P)}w'_q \cdot \cost(q,f(\calS))\right\vert \label{eq:termin2}
\leq \varepsilon\cdot \sum_{p\in f(A)} w_p \cdot \cost(p,f(\calS)),
\end{eqnarray}
where $w$ and $w'$ are the weights assigned to points in $f(A)$ and $f(P)$, respectively.
Let us now consider a solution $S$ in the original $d$-dimensional space. Since $P$ is a subset of $A$, we have by combining Equations~\ref{eq:termin1} and \ref{eq:termin2}
\begin{eqnarray*}
& &\left\vert \sum_{p\in A} w_p \cdot \cost(p,\calS) - \sum_{q\in P}w'_q \cdot \cost(p,\calS)\right\vert \\
&\leq & \varepsilon\cdot \sum_{p\in A} w_p \cdot \cost(f(p),f(\calS)) +  \varepsilon\cdot \sum_{q\in P}w'_q \cdot \cost(f(q),f(\calS)) \\
& & + \left\vert \sum_{p\in A} w_p \cdot \cost(f(p),f(\calS)) - \sum_{q\in P}w'_q \cdot \cost(f(q),f(\calS))\right\vert \\
&\leq & 2\varepsilon\cdot \sum_{p\in A}w_p \cdot \cost(f(p),f(\calS)) + \varepsilon\cdot \sum_{q\in P}w'_q \cdot \cost(f(q),f(\calS)) \\
&\leq & (3 + \varepsilon)\varepsilon\cdot \sum_{p\in A} w_p \cdot \cost(f(p),f(\calS)) \\
&\leq &  (3 + 3\varepsilon)\varepsilon\cdot \sum_{p\in A} w_p \cdot \cost(p,\calS),
\end{eqnarray*}
where the second inequality uses Equation \ref{eq:termin2} and the triangle inequality and the last inequality uses Equation \ref{eq:termin1}. 
\end{proof}

\ifstoc
\bibliographystyle{alpha}
\bibliography{references}
\else
{\small
\bibliographystyle{alpha}
\bibliography{references}

\newcommand{\etalchar}[1]{$^{#1}$}
\begin{thebibliography}{MMMR18}

\bibitem[AG06]{AbrahamG06}
Ittai Abraham and Cyril Gavoille.
\newblock Object location using path separators.
\newblock In Eric Ruppert and Dahlia Malkhi, editors, {\em Proceedings of the
  Twenty-Fifth Annual {ACM} Symposium on Principles of Distributed Computing,
  {PODC} 2006, Denver, CO, USA, July 23-26, 2006}, pages 188--197. {ACM}, 2006.

\bibitem[BBC{\etalchar{+}}19]{BecchettiBC0S19}
Luca Becchetti, Marc Bury, Vincent Cohen{-}Addad, Fabrizio Grandoni, and Chris
  Schwiegelshohn.
\newblock Oblivious dimension reduction for \emph{k}-means: beyond subspaces
  and the johnson-lindenstrauss lemma.
\newblock In {\em Proceedings of the 51st Annual {ACM} {SIGACT} Symposium on
  Theory of Computing, {STOC} 2019, Phoenix, AZ, USA, June 23-26, 2019}, pages
  1039--1050, 2019.

\bibitem[BBH{\etalchar{+}}20]{baker2020coresets}
Daniel Baker, Vladimir Braverman, Lingxiao Huang, Shaofeng H.~C. Jiang, Robert
  Krauthgamer, and Xuan Wu.
\newblock Coresets for clustering in graphs of bounded treewidth, 2020.

\bibitem[BEHW89]{BlumerEHW89}
Anselm Blumer, Andrzej Ehrenfeucht, David Haussler, and Manfred~K. Warmuth.
\newblock Learnability and the vapnik-chervonenkis dimension.
\newblock {\em J. {ACM}}, 36(4):929--965, 1989.

\bibitem[BEL13]{BalcanEL13}
Maria{-}Florina Balcan, Steven Ehrlich, and Yingyu Liang.
\newblock Distributed k-means and k-median clustering on general communication
  topologies.
\newblock In {\em Advances in Neural Information Processing Systems 26: 27th
  Annual Conference on Neural Information Processing Systems 2013. Proceedings
  of a meeting held December 5-8, 2013, Lake Tahoe, Nevada, United States},
  pages 1995--2003, 2013.

\bibitem[BFL{\etalchar{+}}17]{BravermanFLSY17}
Vladimir Braverman, Gereon Frahling, Harry Lang, Christian Sohler, and Lin~F.
  Yang.
\newblock Clustering high dimensional dynamic data streams.
\newblock In {\em Proceedings of the 34th International Conference on Machine
  Learning, {ICML} 2017, Sydney, NSW, Australia, 6-11 August 2017}, pages
  576--585, 2017.

\bibitem[BFLR19]{BravermanFLR19}
Vladimir Braverman, Dan Feldman, Harry Lang, and Daniela Rus.
\newblock Streaming coreset constructions for m-estimators.
\newblock In {\em Approximation, Randomization, and Combinatorial Optimization.
  Algorithms and Techniques, {APPROX/RANDOM} 2019, September 20-22, 2019,
  Massachusetts Institute of Technology, Cambridge, MA, {USA}}, pages
  62:1--62:15, 2019.

\bibitem[BJKW19]{BravermanJKW19}
Vladimir Braverman, Shaofeng~H.{-}C. Jiang, Robert Krauthgamer, and Xuan Wu.
\newblock Coresets for ordered weighted clustering.
\newblock In {\em Proceedings of the 36th International Conference on Machine
  Learning, {ICML} 2019, 9-15 June 2019, Long Beach, California, {USA}}, pages
  744--753, 2019.

\bibitem[BJKW21]{BravermanJKW21}
Vladimir Braverman, Shaofeng~H.{-}C. Jiang, Robert Krauthgamer, and Xuan Wu.
\newblock Coresets for clustering in excluded-minor graphs and beyond.
\newblock In D{\'{a}}niel Marx, editor, {\em Proceedings of the 2021 {ACM-SIAM}
  Symposium on Discrete Algorithms, {SODA} 2021, Virtual Conference, January 10
  - 13, 2021}, pages 2679--2696. {SIAM}, 2021.

\bibitem[BLHK17]{BachemLH017}
Olivier Bachem, Mario Lucic, S.~Hamed Hassani, and Andreas Krause.
\newblock Uniform deviation bounds for k-means clustering.
\newblock In Doina Precup and Yee~Whye Teh, editors, {\em Proceedings of the
  34th International Conference on Machine Learning, {ICML} 2017, Sydney, NSW,
  Australia, 6-11 August 2017}, volume~70 of {\em Proceedings of Machine
  Learning Research}, pages 283--291. {PMLR}, 2017.

\bibitem[BLL18]{BachemLL18}
Olivier Bachem, Mario Lucic, and Silvio Lattanzi.
\newblock One-shot coresets: The case of k-clustering.
\newblock In {\em International Conference on Artificial Intelligence and
  Statistics, {AISTATS} 2018, 9-11 April 2018, Playa Blanca, Lanzarote, Canary
  Islands, Spain}, pages 784--792, 2018.

\bibitem[CEM{\etalchar{+}}15]{CEMMP15}
Michael~B. Cohen, Sam Elder, Cameron Musco, Christopher Musco, and Madalina
  Persu.
\newblock Dimensionality reduction for k-means clustering and low rank
  approximation.
\newblock In {\em Proceedings of the Forty-Seventh Annual {ACM} on Symposium on
  Theory of Computing, {STOC} 2015, Portland, OR, USA, June 14-17, 2015}, pages
  163--172, 2015.

\bibitem[Che09]{Chen09}
Ke~Chen.
\newblock On coresets for k-median and k-means clustering in metric and
  {E}uclidean spaces and their applications.
\newblock {\em {SIAM} J. Comput.}, 39(3):923--947, 2009.

\bibitem[CL19]{Cohen-AddadL19}
Vincent Cohen{-}Addad and Jason Li.
\newblock On the fixed-parameter tractability of capacitated clustering.
\newblock In Christel Baier, Ioannis Chatzigiannakis, Paola Flocchini, and
  Stefano Leonardi, editors, {\em 46th International Colloquium on Automata,
  Languages, and Programming, {ICALP} 2019, July 9-12, 2019, Patras, Greece},
  volume 132 of {\em LIPIcs}, pages 41:1--41:14. Schloss Dagstuhl -
  Leibniz-Zentrum f{\"{u}}r Informatik, 2019.

\bibitem[CMK19]{CsikosMK19}
M{\'{o}}nika Csik{\'{o}}s, Nabil~H. Mustafa, and Andrey Kupavskii.
\newblock Tight lower bounds on the vc-dimension of geometric set systems.
\newblock {\em J. Mach. Learn. Res.}, 20:81:1--81:8, 2019.

\bibitem[CPP18]{ceccarello2018fast}
Matteo Ceccarello, Andrea Pietracaprina, and Geppino Pucci.
\newblock Fast coreset-based diversity maximization under matroid constraints.
\newblock In {\em Proceedings of the Eleventh ACM International Conference on
  Web Search and Data Mining}, pages 81--89, 2018.

\bibitem[CS17]{Cohen-AddadS17}
Vincent Cohen{-}Addad and Chris Schwiegelshohn.
\newblock On the local structure of stable clustering instances.
\newblock In {\em 58th {IEEE} Annual Symposium on Foundations of Computer
  Science, {FOCS} 2017, Berkeley, CA, USA, October 15-17, 2017}, pages 49--60,
  2017.

\bibitem[EA07]{EisenstatA07}
David Eisenstat and Dana Angluin.
\newblock The {VC} dimension of k-fold union.
\newblock {\em Inf. Process. Lett.}, 101(5):181--184, 2007.

\bibitem[EFN17]{ElkinFN17}
Michael Elkin, Arnold Filtser, and Ofer Neiman.
\newblock Terminal embeddings.
\newblock {\em Theor. Comput. Sci.}, 697:1--36, 2017.

\bibitem[EKM14]{EisenstatKM14}
David Eisenstat, Philip~N. Klein, and Claire Mathieu.
\newblock Approximating \emph{k}-center in planar graphs.
\newblock In Chandra Chekuri, editor, {\em Proceedings of the Twenty-Fifth
  Annual {ACM-SIAM} Symposium on Discrete Algorithms, {SODA} 2014, Portland,
  Oregon, USA, January 5-7, 2014}, pages 617--627. {SIAM}, 2014.

\bibitem[FGS{\etalchar{+}}13]{FGSSS13}
Hendrik Fichtenberger, Marc Gill{\'{e}}, Melanie Schmidt, Chris Schwiegelshohn,
  and Christian Sohler.
\newblock {BICO:} {BIRCH} meets coresets for k-means clustering.
\newblock In {\em Algorithms - {ESA} 2013 - 21st Annual European Symposium,
  Sophia Antipolis, France, September 2-4, 2013. Proceedings}, pages 481--492,
  2013.

\bibitem[FKW19]{FKW19}
Zhili Feng, Praneeth Kacham, and David~P. Woodruff.
\newblock Strong coresets for subspace approximation and k-median in nearly
  linear time.
\newblock {\em CoRR}, abs/1912.12003, 2019.

\bibitem[FL11]{FeldmanL11}
Dan Feldman and Michael Langberg.
\newblock A unified framework for approximating and clustering data.
\newblock In {\em Proceedings of the 43rd {ACM} Symposium on Theory of
  Computing, {STOC} 2011, San Jose, CA, USA, 6-8 June 2011}, pages 569--578,
  2011.

\bibitem[FMS07]{FMS07}
Dan Feldman, Morteza Monemizadeh, and Christian Sohler.
\newblock A {PTAS} for k-means clustering based on weak coresets.
\newblock In {\em Proceedings of the 23rd {ACM} Symposium on Computational
  Geometry, Gyeongju, South Korea, June 6-8, 2007}, pages 11--18, 2007.

\bibitem[FS05]{FrahlS2005}
Gereon Frahling and Christian Sohler.
\newblock Coresets in dynamic geometric data streams.
\newblock In {\em Proceedings of the 37th Annual {ACM} Symposium on Theory of
  Computing ({STOC})}, pages 209--217, 2005.

\bibitem[FSS20]{FeldmanSS20}
Dan Feldman, Melanie Schmidt, and Christian Sohler.
\newblock Turning big data into tiny data: Constant-size coresets for k-means,
  pca, and projective clustering.
\newblock {\em {SIAM} J. Comput.}, 49(3):601--657, 2020.

\bibitem[GKL03]{GuptaKL03}
Anupam Gupta, Robert Krauthgamer, and James~R. Lee.
\newblock Bounded geometries, fractals, and low-distortion embeddings.
\newblock In {\em 44th Symposium on Foundations of Computer Science {(FOCS}
  2003), 11-14 October 2003, Cambridge, MA, USA, Proceedings}, pages 534--543,
  2003.

\bibitem[HCB16]{huggins2016coresets}
Jonathan Huggins, Trevor Campbell, and Tamara Broderick.
\newblock Coresets for scalable bayesian logistic regression.
\newblock In {\em Advances in Neural Information Processing Systems}, pages
  4080--4088, 2016.

\bibitem[HJLW18]{HuangJLW18}
Lingxiao Huang, Shaofeng~H.{-}C. Jiang, Jian Li, and Xuan Wu.
\newblock Epsilon-coresets for clustering (with outliers) in doubling metrics.
\newblock In {\em 59th {IEEE} Annual Symposium on Foundations of Computer
  Science, {FOCS} 2018, Paris, France, October 7-9, 2018}, pages 814--825,
  2018.

\bibitem[HJV19]{HuangJV19}
Lingxiao Huang, Shaofeng~H.{-}C. Jiang, and Nisheeth~K. Vishnoi.
\newblock Coresets for clustering with fairness constraints.
\newblock In {\em Advances in Neural Information Processing Systems 32: Annual
  Conference on Neural Information Processing Systems 2019, NeurIPS 2019, 8-14
  December 2019, Vancouver, BC, Canada}, pages 7587--7598, 2019.

\bibitem[HK07]{HaK07}
Sariel Har{-}Peled and Akash Kushal.
\newblock Smaller coresets for k-median and k-means clustering.
\newblock {\em Discrete {\&} Computational Geometry}, 37(1):3--19, 2007.

\bibitem[HM01]{HaM01}
Pierre Hansen and Nenad Mladenovic.
\newblock J-m\({}_{\mbox{eans}}\): a new local search heuristic for minimum sum
  of squares clustering.
\newblock {\em Pattern Recognition}, 34(2):405--413, 2001.

\bibitem[HM04]{HaM04}
Sariel Har{-}Peled and Soham Mazumdar.
\newblock On coresets for k-means and k-median clustering.
\newblock In {\em Proceedings of the 36th Annual {ACM} Symposium on Theory of
  Computing, Chicago, IL, USA, June 13-16, 2004}, pages 291--300, 2004.

\bibitem[HV20]{huang2020coresets}
Lingxiao Huang and Nisheeth~K. Vishnoi.
\newblock Coresets for clustering in euclidean spaces: importance sampling is
  nearly optimal.
\newblock In Konstantin Makarychev, Yury Makarychev, Madhur Tulsiani, Gautam
  Kamath, and Julia Chuzhoy, editors, {\em Proccedings of the 52nd Annual {ACM}
  {SIGACT} Symposium on Theory of Computing, {STOC} 2020, Chicago, {IL}, USA,
  June 22-26, 2020}, pages 1416--1429. {ACM}, 2020.

\bibitem[IMGR20]{IndykMGR20}
Piotr Indyk, Sepideh Mahabadi, Shayan~Oveis Gharan, and Alireza Rezaei.
\newblock Composable core-sets for determinant maximization problems via
  spectral spanners.
\newblock In Shuchi Chawla, editor, {\em Proceedings of the 2020 {ACM-SIAM}
  Symposium on Discrete Algorithms, {SODA} 2020, Salt Lake City, UT, USA,
  January 5-8, 2020}, pages 1675--1694. {SIAM}, 2020.

\bibitem[IMMM14]{IndykMMM14}
Piotr Indyk, Sepideh Mahabadi, Mohammad Mahdian, and Vahab~S. Mirrokni.
\newblock Composable core-sets for diversity and coverage maximization.
\newblock In Richard Hull and Martin Grohe, editors, {\em Proceedings of the
  33rd {ACM} {SIGMOD-SIGACT-SIGART} Symposium on Principles of Database
  Systems, PODS'14, Snowbird, UT, USA, June 22-27, 2014}, pages 100--108.
  {ACM}, 2014.

\bibitem[LL06]{LiL06a}
Yi~Li and Philip~M. Long.
\newblock Learnability and the doubling dimension.
\newblock In {\em Advances in Neural Information Processing Systems 19,
  Proceedings of the Twentieth Annual Conference on Neural Information
  Processing Systems, Vancouver, British Columbia, Canada, December 4-7, 2006},
  pages 889--896, 2006.

\bibitem[LLS01]{LiLS01}
Yi~Li, Philip~M. Long, and Aravind Srinivasan.
\newblock Improved bounds on the sample complexity of learning.
\newblock {\em J. Comput. Syst. Sci.}, 62(3):516--527, 2001.

\bibitem[LN17]{LarsenN17}
Kasper~Green Larsen and Jelani Nelson.
\newblock Optimality of the {J}ohnson-{L}indenstrauss {L}emma.
\newblock In {\em 58th {IEEE} Annual Symposium on Foundations of Computer
  Science, {FOCS} 2017, Berkeley, CA, USA, October 15-17, 2017}, pages
  633--638, 2017.

\bibitem[LS10]{LS10}
Michael Langberg and Leonard~J. Schulman.
\newblock Universal $\varepsilon$-approximators for integrals.
\newblock In {\em Proceedings of the Twenty-First Annual {ACM-SIAM} Symposium
  on Discrete Algorithms, {SODA} 2010, Austin, Texas, USA, January 17-19,
  2010}, pages 598--607, 2010.

\bibitem[Mat00]{Mat00}
Jir{\'{\i}} Matousek.
\newblock On approximate geometric k-clustering.
\newblock {\em Discrete {\&} Computational Geometry}, 24(1):61--84, 2000.

\bibitem[MJF19]{maalouf2019fast}
Alaa Maalouf, Ibrahim Jubran, and Dan Feldman.
\newblock Fast and accurate least-mean-squares solvers.
\newblock In {\em Advances in Neural Information Processing Systems}, pages
  8307--8318, 2019.

\bibitem[MMK18]{MK18}
Alejandro Molina, Alexander Munteanu, and Kristian Kersting.
\newblock Core dependency networks.
\newblock In Sheila~A. McIlraith and Kilian~Q. Weinberger, editors, {\em
  Proceedings of the Thirty-Second {AAAI} Conference on Artificial
  Intelligence, (AAAI-18), the 30th innovative Applications of Artificial
  Intelligence (IAAI-18), and the 8th {AAAI} Symposium on Educational Advances
  in Artificial Intelligence (EAAI-18), New Orleans, Louisiana, USA, February
  2-7, 2018}, pages 3820--3827. {AAAI} Press, 2018.

\bibitem[MMMR18]{MahabadiMMR18}
Sepideh Mahabadi, Konstantin Makarychev, Yury Makarychev, and Ilya~P.
  Razenshteyn.
\newblock Nonlinear dimension reduction via outer bi-lipschitz extensions.
\newblock In {\em Proceedings of the 50th Annual {ACM} {SIGACT} Symposium on
  Theory of Computing, {STOC} 2018, Los Angeles, CA, USA, June 25-29, 2018},
  pages 1088--1101, 2018.

\bibitem[MMR19]{MakarychevMR19}
Konstantin Makarychev, Yury Makarychev, and Ilya~P. Razenshteyn.
\newblock Performance of johnson-lindenstrauss transform for \emph{k}-means and
  \emph{k}-medians clustering.
\newblock In {\em Proceedings of the 51st Annual {ACM} {SIGACT} Symposium on
  Theory of Computing, {STOC} 2019, Phoenix, AZ, USA, June 23-26, 2019}, pages
  1027--1038, 2019.

\bibitem[MP04]{MettuP04}
Ramgopal~R. Mettu and C.~Greg Plaxton.
\newblock Optimal time bounds for approximate clustering.
\newblock {\em Mach. Learn.}, 56(1-3):35--60, 2004.

\bibitem[MS18]{MunteanuS18}
Alexander Munteanu and Chris Schwiegelshohn.
\newblock Coresets-methods and history: {A} theoreticians design pattern for
  approximation and streaming algorithms.
\newblock {\em K{\"{u}}nstliche Intell.}, 32(1):37--53, 2018.

\bibitem[MSSW18]{MunteanuSSW18}
Alexander Munteanu, Chris Schwiegelshohn, Christian Sohler, and David~P.
  Woodruff.
\newblock On coresets for logistic regression.
\newblock In Samy Bengio, Hanna~M. Wallach, Hugo Larochelle, Kristen Grauman,
  Nicol{\`{o}} Cesa{-}Bianchi, and Roman Garnett, editors, {\em Advances in
  Neural Information Processing Systems 31: Annual Conference on Neural
  Information Processing Systems 2018, NeurIPS 2018, December 3-8, 2018,
  Montr{\'{e}}al, Canada}, pages 6562--6571, 2018.

\bibitem[NN19]{NaN18}
Shyam Narayanan and Jelani Nelson.
\newblock Optimal terminal dimensionality reduction in euclidean space.
\newblock In Moses Charikar and Edith Cohen, editors, {\em Proceedings of the
  51st Annual {ACM} {SIGACT} Symposium on Theory of Computing, {STOC} 2019,
  Phoenix, AZ, USA, June 23-26, 2019}, pages 1064--1069. {ACM}, 2019.

\bibitem[Pol12]{pollard2012}
David Pollard.
\newblock {\em Convergence of stochastic processes}.
\newblock Springer Science \& Business Media, 2012.

\bibitem[SSS19]{SSS19}
Melanie Schmidt, Chris Schwiegelshohn, and Christian Sohler.
\newblock Fair coresets and streaming algorithms for fair k-means.
\newblock In {\em Approximation and Online Algorithms - 17th International
  Workshop, {WAOA} 2019, Munich, Germany, September 12-13, 2019, Revised
  Selected Papers}, pages 232--251, 2019.

\bibitem[SW18]{SohlerW18}
Christian Sohler and David~P. Woodruff.
\newblock Strong coresets for k-median and subspace approximation: Goodbye
  dimension.
\newblock In {\em 59th {IEEE} Annual Symposium on Foundations of Computer
  Science, {FOCS} 2018, Paris, France, October 7-9, 2018}, pages 802--813,
  2018.

\bibitem[T{\etalchar{+}}96]{talagrand1996majorizing}
Michel Talagrand et~al.
\newblock Majorizing measures: the generic chaining.
\newblock {\em The Annals of Probability}, 24(3):1049--1103, 1996.

\bibitem[Vit85]{Vitter85}
Jeffrey~Scott Vitter.
\newblock Random sampling with a reservoir.
\newblock {\em {ACM} Trans. Math. Softw.}, 11(1):37--57, 1985.

\end{thebibliography}
}
\fi

\appendix
\section{Missing Proof}
\weaktri*
\begin{proof}
The proof of the first inequality is appears in \cite{MakarychevMR19}, Corollary A.2.

For the second part, let $S(a), S(b)$ be the closest point to $a$ and $b$ from $S$, and assume that $d(b, S) \leq d(a, S)$. Then:
\begin{align*}
d(a, S)^z &\leq d(a, S(b))^z\\
&\leq \left(1+\frac{\eps}{2z}\right)^{z-1} \cdot d(b, S(b))^z +  \left(1+\frac{2z}{\eps}\right)^{z-1} \cdot d(a, b)^z\\
&\leq (1+\eps)\cdot  d(b, S(b))^z +  \left(1+\frac{2z}{\eps}\right)^{z-1} \cdot d(a, b)^z\\
&\leq d(b, S)^z + \eps \cdot d(a, S(a))^z + \left(1+\frac{2z}{\eps}\right)^{z-1}\cdot  d(a, b)^z,
\end{align*}
and so 
\[\left\vert d(a,S)^z - d(b, S)^z\right\vert = d(a,S)^z - d(b, S)^z \leq \varepsilon \cdot d(a,S)^z + \left(\frac{2z+\varepsilon}{\varepsilon}\right)^{z-1} d(a, b)^z.\]

In the other case, when $d(a, S) \leq d(b, S)$:
\begin{align*}
d(b, S)^z &\leq d(b, S(a))^z\\
&\leq \left(1+\frac{\eps}{2z}\right)^{z-1} \cdot d(a, S(a))^z +  \left(1+\frac{2z}{\eps}\right)^{z-1} \cdot d(a, b)^z\\
&\leq (1+\eps) \cdot d(a, S)^z +  \left(1+\frac{2z}{\eps}\right)^{z-1} \cdot d(a, b)^z,
\end{align*}
and so
\[\left\vert d(a,S)^z - d(b, S)^z\right\vert = d(b,S)^z - d(a, S)^z \leq \varepsilon \cdot d(a,S)^z + \left(\frac{2z+\varepsilon}{\varepsilon}\right)^{z-1} d(a, b)^z.\]

\end{proof}

\section{A Coreset of Size $k^2\eps^{-2}$} \label{sec:ksquare}
In this section, we show how to trade a factor $\eps^{-z}$ for a factor $k$ in the coreset size. 

\begin{lemma}\label{lem:k2sampling}
Let $(X, \dist)$ be a metric space, $P$ be a set of points, $k, z$ two positive integers and $\greedy$ a set of $O(k)$ centers such that each for each cluster with center $c$ induced by $\greedy$, all points of the cluster are at distance between $\left(\frac{\eps}{z}\right)^2\Delta_C$ and $\left(\frac{z}{\eps}\right)^{2} \Delta_C$, for some $\Delta_C$.

Suppose there exists an $\greedy$-approximate centroid set $\cand$ for $P$. 

Then, there exists an algorithm running in time $O(|P|)$ that constructs a set $\coreset$ of size $O(k\cdot 2^{O(z)} \frac{\log^3(1/\eps)}{\eps^2}\left(k \log k + k \log |\cand| + \log(1/\pi)\right)$ such that, with probability $1 - 1/\pi$, for any set $\calS$ of $k$ centers,
\[\left\vert \cost(\calS) - \cost(\coreset, \calS)\right\vert = O(\eps)\cost(\calS).\]
\end{lemma}

Suppose we initially computed a set of $k'$ centers $\greedy$. 
Our aim is to define a sampling distribution that approximates the cost of any solution $\calS$ with high probability.
While the basic idea is related to importance sampling (i.e. sampling proportionate to $\cost(p,\greedy)$), we add a few modifications that are crucial.

Compared to the framework described in the main body, we change slightly the definition of ring.
For every cluster $\mathcal{C}_i$ of $\greedy$, we partition the points of $\mathcal{C}_i$ into rings $R_{i,j}$ from between distances $[\left(\frac{\eps}{z}\right)^2  \Delta_C \cdot 2^j, \left(\frac{\eps}{z}\right)^2 \Delta_C  \cdot 2^{j+1}]$, for $j\in\{1,\ldots 4z\log (z/\eps)\}$.

The algorithm is as follows: from every $R_{i,j}$, sample $\delta$ points uniformly at random (if $|R_{i,j}|\leq \delta$, simply add the whole $R_{i,j}$).

The analysis of this algorithm follows the same line as the main one. Rings are divided into tiny, interesting and huge types; tiny and huge are dealt with as in Lemmas~\ref{lem:kepstiny} and~\ref{lem:khuge}, and interesting points slightly differently.

From the definition of $R_{i,j}$, we immediately get the following observation.
\begin{fact}
\label{fact:k2ringcount}
For every cluster we have at most $O(z\cdot \log z/\eps)$ non-empty rings in total.
\end{fact}

Given a solution $\calS$, we consider the groups $I_{i,j,\ell}\subset \mathcal{C}_i$ consisting of the points of $R_{i,j}$ served in $\calS$ by a center at distance $[\eps \cdot 2^{\ell},\eps\cdot 2^{\ell+1}]$.
As before, we let
$\cost(I_{i,j,\ell},\calS) = \sum_{p\in I_{i,j,\ell}} \cost(p,\calS)$ and $\cost(I_{j,\ell},\calS) = \sum_{i=1}^{k'} \cost(I_{i,j,\ell},\calS)$.

Our analysis will distinguish between three cases:
\begin{enumerate}
\item $\ell \leq j+ \log \eps $, in which case we say that $I_{i,j,\ell}$ is \emph{tiny}.
\item  $j\cdot \log \eps \leq \ell \leq j + \log (4z/\eps)$, in which case we say $I_{i,j,\ell}$ is \emph{interesting}. 
\item $\ell \geq  j +\log (16z/\eps)$, in which case we say $I_{i,j,\ell}$ is \emph{huge}.
\end{enumerate}

We first consider the huge case.
For this, we show that the weight of every ring is preserved with high probability, which implies that the huge groups are well approximated.
\begin{lemma}
\label{lem:k2huge}
It holds that, for any $R_{i,j}$  and for all solutions $S$ with at least one non-empty huge group $I_{i,j,\ell}$
$$\left\vert\cost(R_{i,j},\calS) - \sum_{p\in \Omega \cap R_{i,j}} \frac{|R_{i,j}|}{\delta} \cdot \cost(p,\calS) \right\vert \leq  3\eps\cdot \cost(R_{i,j},\calS).$$
\end{lemma}
\begin{proof}
Fix a ring $R_{i,j}$ and let $I_{i,j,\ell}$ be a huge group. 
First, the weight of  $R_{i,j}$ is preserved in $\coreset$: since $\delta$ points are sampled from $R_{i,j}$, it holds that 
$$ \sum_{p\in \Omega\cap R_{i,j}} \frac{|R_{i,j}|}{\delta} = |R_{i,j}|$$

Now, let $\calS$ be a solution, and $p\in I_{i,j,\ell}$ with $I_{i,j,\ell}$ being huge. This implies, for any $q \in R_{i,j}$: 
$\cost(p,q) \leq (2\cdot \eps\cdot 2^{j+1})^z\leq 4^z \cdot \eps^z\cdot 2^{(\ell - \log (16z/\eps))z} \leq \left(\frac{\eps}{4z}\right)^{z}\cdot \cost(p,\calS)$. 
By Lemma~\ref{lem:weaktri}, we have therefore for any point $q\in R_{i,j}$
\begin{eqnarray*}
\cost(p,\calS) &\leq & \left(1+\eps/2z \right)^{z-1} \cost(q,\calS) + \left(1+2z/\eps\right)^{z-1}\cost(p,q) \\
&\leq & \left(1+\eps \right) \cost(q,\calS) + \varepsilon\cdot \cost(p,\calS) \\
\Rightarrow \cost(q,\calS) &\geq & \frac{1-\varepsilon}{1+\eps}\cost(p,S) \geq  (1-2\eps) \cost(p, \calS)
\end{eqnarray*}

Moreover, by a similar calculation, we can also derive an upper bound of $\cost(q,\calS)\leq \cost(p,\calS)\cdot (1+2\eps)$. Hence, combined with $\sum_{p\in \Omega\cap R_{i,j}} \frac{|R_{i,j}|}{\delta} = |R_{i,j}|$, this is sufficient to approximate $\cost(R_{i,j}, \calS)$. 

Therefore, the cost of $R_{i,j}$ is well approximated for any solution $\calS$ such that there is a non-empty huge group $I_{i,j,\ell}$. 
\end{proof}

Next, we consider the interesting cases. The main observation here is that there are only $O(\log 1/\eps)$ many rings per cluster, hence a coarser estimation using Bernstein's inequality is actually sufficient to bound the cost.

\begin{lemma}%
\label{lem:k2interesting}
Consider an $R_{i,j}$ and any solution $\calS$ such that all huge $I_{i,j,\ell}$ are empty. 
It holds with probability at least
$1- \log(z/\eps)\exp(-\frac{\eps^2}{2\cdot 16^z\log^2 z/\eps} \cdot \delta)$ that, for all interesting $I_{i, j, \ell}$:
\begin{eqnarray*}
& &\left\vert\cost(I_{i, j, \ell},\calS) - \sum_{p\in I_{i, j, \ell}\cap \Omega} \cost(p,\calS)\cdot\frac{|R_{i,j}|}{\delta} \right\vert \leq  \frac{\eps}{\log(z/\eps)} \cdot \left(\cost(R_{i,j},\greedy) + \cost(R_{i,j},\calS)\right) .
\end{eqnarray*}
\end{lemma}
\begin{proof}

We start by bounding $|R_{i,j}|\cdot (\eps \cdot 2^{\ell})^z$ in terms of $\cost(R_{i,j},\calS)+\cost(R_{i,j},\greedy)$. 

If $I_{i,j,\ell}$ for some $\ell\geq j+3$ is non-empty, then  $\eps \cdot 2^{\ell}-\eps\cdot 2^{j+2} \leq d(q,\calS)$, for any point $q$. Hence, $|R_{i,j}|\cdot (\eps \cdot 2^{\ell})^z \leq \cost(R_{i,j},\calS)\cdot 2^z$. 
If $\ell\leq j+2$, then $|R_{i,j}|\cdot(\eps\cdot 2^{\ell})^z \leq |R_{i,j}|\cdot(\eps\cdot 2^{j+2})^z \leq \cost(R_{i,j},\greedy)\cdot 4^z$. Putting both bounds together, we have
\begin{equation}
\label{eq:k2interesting_1}
|R_{i,j}|\cdot (\eps \cdot 2^{\ell})^z \leq 4^z (\cost(R_{i,j},\calS)+\cost(R_{i,j},\greedy))
\end{equation}

Since we aim to apply Bernstein's inequality, we now require a bound on the second moment of our cost estimator. 
We have for a single randomly chosen point $P$:
$$\mathbb{E}\left[\sum_{p\in I_{i,j,\ell}\cap P}  \cost(p,\calS)\cdot |R_{i,j}|\right] = \cost(I_{i,j,\ell},\calS)$$ 

and
\begin{eqnarray}
\nonumber
\mathbb{E}\left[\left(\sum_{p\in I_{i,j,\ell}\cap P} 
    \cost(p,\calS)\cdot |R_{i,j}|\right)^2\right] &=& \mathbb{E}\left[\sum_{p\in I_{i,j,\ell}\cap P} 
   \cost(p,\calS)^2\cdot |R_{i,j}|^2 \right]  
   \text{ since } |P| = 1\\
\nonumber
&=& \sum_{p\in I_{i,j,\ell}\cap P} \cost(p,\calS)^2\cdot|R_{i,j}| \leq  |R_{i,j}|\cdot |I_{i,j,\ell}|\cdot (\eps 2^{\ell})^{2z} 4^z \\
\label{eq:k2intersting_2}
& \leq & \cost(I_{i,j,\ell},\calS) \cdot
   (\cost(R_{i,j},\calS) + \cost(R_{i,j},\greedy) ) \cdot 16^z
\end{eqnarray}
where the final equation follows from by lower bounding the cost in $\calS$ of any point in $I_{i,j,\ell}$ with $(\eps\cdot 2^{\ell})^z$ and using Equation~\ref{eq:k2interesting_1}.

Furthermore, by the same reasoning and again using Equation~\ref{eq:k2interesting_1},  we have the upper bound $M$ on the (weighted) cost in $\calS$ of every sampled point in every ring:
\begin{equation}
\label{eq:k2interesting_3}
M\leq (\eps\cdot 2^{\ell+1})^z \cdot |R_{i,j}| \leq (\cost(R_{i,j},\calS) + \cost(R_{i,j},\greedy))\cdot 8^z
\end{equation}

Applying Bernstein's inequality and Equations~\ref{eq:k2intersting_2} and \ref{eq:k2interesting_3}, we now have

\begin{eqnarray}
\label{eq:k2interesting_5}
\notag
& &\mathbb{P}\left[\left\vert \delta\cdot \cost(I_{i,j,\ell},\calS) - \sum_{p\in I_{i,j,\ell}\cap \Omega} 
              \cost(p,\calS)\cdot |R_{i,j}| \right\vert > \frac{\eps\cdot \delta}{r}\cdot (\cost(R_{i,j},\calS) + \cost(R_{i,j},\greedy))\right] \\
\nonumber
&\leq & \exp\left(-\frac{\frac{\eps^2\cdot \delta}{r^2}\cdot (\cost(R_{i,j},\calS) + \cost(R_{i,j},\greedy))}
                        {\cost(I_{i,j,\ell},\calS) \cdot 16^z + 
                           \frac{4\eps}{3r}\cdot (\cost(R_{i,j},\calS)+\cost(R_{i,j,\ell}))\cdot 8^z }\right) \leq  \exp\left(-\frac{\eps^2\cdot \delta}{2r^2 16^z}\right) ,
\end{eqnarray}
where the last line uses $\cost(I_{i,j, \ell}, \calS) \leq \cost(R_{i,j}, \calS)$.
Applying a union bound over all $r$ interesting sets $I_{i,j,\ell}$, we obtain the above guarantee for all $I_{i,j,\ell}$ simultaneously with probability 
\begin{equation*}
1-r\cdot \exp\left(-\frac{\eps^2\cdot \delta}{2r^2 16^z}\right).
\end{equation*}

%
\end{proof}

Finally, we conclude:

\begin{proof}[Proof of \cref{lem:k2sampling}]
As in the proof of \cref{lem:coreset-reduc}, we decompose $|\cost(\calS) - \cost(\coreset,\calS)|$ into terms corresponding to points of tiny, interesting or huge groups. We only sketch the proof here, the details are the same as for \cref{lem:coreset-reduc}.
We condition on event $\calE$ happening. Let $\calS$ be a set of $k$ points, and $\tilde \calS \in \cand^k$ that approximates best $\calS$, as given by the definition of $\cand$ (see \cref{def:centroid-set}). This ensures that 
for all points $p$ with $\dist(p, \calS) \leq \frac{8z}{\eps}\cdot \dist(p, \greedy)$ or $\dist(p, \tilde \calS) \leq \frac{8z}{\eps}\cdot \dist(p, \greedy)$ , we have $|\cost(p, \calS) - \cost(p, \tilde \calS)|\leq \eps(\cost(p, \calS) + \cost(p, \greedy))$.

Our first step is to deal with points that have $\dist(p, \calS) > \frac{8z}{\eps}\cdot \dist(p, \greedy)$, using \cref{lem:k2huge}. All other points have distance well approximated by $\tilde \calS$.
Then, we can apply \cref{lem:kepstiny} and \cref{lem:k2interesting} to $L_{\tilde \calS}$, since all points in $L_{\tilde \calS}$ have $\dist(p, \tilde \calS) \leq \frac{4z}{\eps}\cdot \dist(p, \greedy)$, and so $\dist(p, \calS) \leq \frac{8z}{\eps}\cdot \dist(p, \greedy)$ and were not removed by the previous step.
Remaining points are those which have $\dist(p, \tilde \calS) > \frac{4z}{\eps}\cdot \dist(p, \greedy)$ and $\dist(p, \calS) \leq \frac{8z}{\eps}\cdot \dist(p, \greedy)$, i.e., their distance is preserved in $\tilde \calS$ and they are huge with respect to $\tilde \calS$. We apply \cref{lem:khuge} to them as well.
\end{proof}

Combining this lemma and \cref{lem:preprocess} gives an analogous to \cref{thm:main}. Now, using this lemma instead of \cref{thm:main} in all proofs of section \cref{sec:centroid} gives bound with a factor $k$ instead of a $\eps^{-z}$.

\end{document}